\documentclass[12pt, a4paper]{article}

\newcommand{\NoDataAvailability}{}
\newcommand{\ForthcomingNote}{\\[0.6em] \normalsize Forthcoming, \emph{Review of Economic Studies}}

\usepackage{natbib}
\usepackage{graphicx}
\usepackage[margin=1.00in]{geometry}
\usepackage{float}                
\usepackage{amsmath}
\usepackage{amsfonts}
\usepackage{amssymb}
\usepackage{bbm}
\usepackage{booktabs}
\usepackage{lscape}
\usepackage{enumerate}
\usepackage{caption}
\usepackage{subcaption}
\usepackage{setspace}
\usepackage{dsfont}
\usepackage{amsthm}
\usepackage{lmodern}
\usepackage{bm}

\usepackage{colortbl}

\newtheorem{theorem}{Theorem}
\newtheorem{prop}[theorem]{Proposition} 
\newtheorem{assumption}{Assumption} 
\newtheorem{lem}{Lemma}

\usepackage[hyperfootnotes=true]{hyperref}

\hypersetup{
     pdftitle={Structured Payment in Pawnshop Borrowing},    %
     pdfauthor={Francis J. DiTraglia, Craig McIntosh, Isaac Meza, Joyce Sadka, Enrique Seira}     %
}

\usepackage{bookmark}

\usepackage{bibunits}
\defaultbibliographystyle{authordate1}
\defaultbibliography{References}

\usepackage{authblk}
\makeatletter
\renewcommand\AB@affilsepx{, }
\makeatother

\begin{document}

\title{Structured Payment in Pawnshop Borrowing: Mandates vs.\ Choice\thanks{We thank Johannes Abeler, Jose Maria Barrero, Chen Huang, Camilo García-Jimeno, Andrei Gomberg, Emilio Gutierrez, Anett John, David Laibson, Aprajit Mahajan, Matt Rabin, Mauricio Romero, Charlie Sprenger, Séverine Toussaert, Chris Udry, Jonathan Zinman, and seminar participants at Dartmouth, ITAM, MSU, UC Berkeley, UC Davis, UCSD, USC, and Wisconsin, who provided valuable feedback. Ricardo Olivares, Gerardo Melendez, and Alonso de Gortari provided excellent research assistance and Erick Molina helped with formatting.  Research assistance was financed through faculty grants at ITAM. Our research partner had no say in the design or results of the experiment.}}
\author[1]{Francis J.\ DiTraglia}
\author[2]{Craig McIntosh}
\author[3]{Isaac Meza}
\author[4]{Joyce Sadka}
\author[5]{Enrique Seira\thanks{Corresponding Author: \href{mailto:enrique.seira@gmail.com}{enrique.seira@gmail.com}}}
\affil[1]{\normalsize University of Oxford}
\affil[2]{\normalsize University of California San Diego}
\affil[3]{\normalsize Harvard University}
\affil[4]{\normalsize ITAM}
\affil[5]{\normalsize University of Notre Dame}

\providecommand{\ForthcomingNote}{}
\date{August 1, 2026 \ForthcomingNote \\[2 cm]}

\maketitle
\vspace{-0.75in}

\begin{abstract}
\singlespacing
Pawn loans offer borrowers a substantial degree of repayment flexibility in exchange for a harsh penalty in case of default: forfeit of collateral worth more than the loan amount along with any payments made toward recovery. 
Using a large RCT conducted in Mexico City, we document key stylized facts about pawn lending and explore the merits of replacing flexibility with structured repayment contracts in this important but understudied form of credit. Our experimental design includes a mandatory frequent-payments arm, a (status quo) flexible payments arm, and a choice between the two.  This design point-identifies not only the average treatment effect, but also the effects of treatment on the treated and the untreated along with the average selection on gains, allowing a rigorous study of mandates versus choice. Although the average treatment effect of assigning borrowers to structured payments is a 19\% decrease in their financial cost and a 17.5\% decrease in the probability of default, only 11\% of borrowers choose structured repayment contracts voluntarily. We show that structured repayment reduces financial costs for nearly all borrowers, including those who would not freely choose it, and find no evidence of selection on gains in cost savings.

\end{abstract}


\begin{bibunit}
\newpage

\section{Introduction}

Pawn lending is one of the oldest and most prevalent forms of credit in the world \citep{carter2012pawnshops}. Pawn loans are typically small and short-term. They involve no credit checks or proof of income, and are widely used by people without access to other forms of credit. Borrowers surrender highly liquid collateral as the pawn, typically gold, in return for a loan that is substantially smaller than the value of their pawn.  Pawn loans are highly flexible: there are no scheduled payments or prepayment penalties, and interest is charged only on the outstanding balance. But this flexibility comes at a cost.
Unlike mortgage lenders, who recover only the outstanding loan balance in foreclosure, pawn lenders keep the collateral's full value in default. Borrowers also forfeit any payments already made.
Pawn lending operates at substantial scale worldwide. In the US, more than 11,000 shops serve 30 million clients with \$14 billion in yearly revenues; in China, pawn lending is a \$43 billion industry.\footnote{Section \ref{context} provides the relevant sources for these statistics.}
In Mexico, where our study takes place, the scale of pawn lending rivals that of microfinance.\footnote{During the past three years a single large lender in Mexico made over 4 million loans to more than a million clients, compared to a total of 2.3 million microfinance clients in all of Mexico during 2009 \citep{Pedroza:2010}.}
In spite of this, pawnshops have received little attention in the literature to date. 

Pawn loans differ from conventional credit in a number of important ways. First, pawn loans have extremely high default rates: 44\% of loans made by our partner lender are not repaid, and 29\% of defaulters make payments before they default---losing this money along with their collateral. Second, lender profits are higher under default than repayment. An industry standard contract in Mexico City lends 70\% of the value of gold collateral at an interest rate of 7\% per month (compounded daily) for three months. Because gold is highly liquid, the profit from a borrower who defaults is 43\% of the loan amount compared to 22.5\% in interest if a loan is paid in full on the last day. Third, pawn borrowers are typically economically vulnerable and in our context are also naively optimistic: borrowers in our sample report a subjective probability of repayment that averages 92\%. The high degree of flexibility in these contracts would benefit borrowers if they were neoclassical agents, but the prevalence of costly default that raises financial costs for apparently naive borrowers but benefits lenders raises concerns about borrower welfare. Lending modalities with lower default rates, such as microfinance, typically rely on highly structured repayment. 

This paper asks whether structured payments can help pawnshop borrowers avoid default. To do so, we implement a unique multi-armed experiment covering around 4,500 pawnshop clients across six of our partner lender's Mexico City branches. Borrowers in the Control arm received the status quo pawn contract, described above, allowing us to document payment trajectories and default under standard pawn contracts. In contrast, borrowers in the ``Structured payments'' arm were required to make three monthly installment payments, each including the accrued interest at that time. These borrowers faced a fee of 2\% of that month's installment if the payment was delinquent. The fee could only be collected from borrowers who ultimately repaid their loans. Finally, borrowers in the ``Choice'' arm were allowed to choose between the Standard and Structured payments contracts. By juxtaposing mandatory contract assignment and choice, this design allows us to go beyond the average treatment effect (ATE) and point identify treatment effects \emph{separately} for borrowers who would voluntarily choose structured payments (choosers) and those who would not (non-choosers). This allows us to answer a fundamental question: do the ``right borrowers'' choose structured payments? More precisely: are reductions in financial costs largest for borrowers who would freely choose structured payments? 

Our experimental design, the ``Mandates versus Choice'' design, involves three comparisons. First, comparing the Structured payments and Control arms gives the ATE of structured payments: this comparison constitutes a standard randomized experiment with full compliance. A comparison of the Choice and Control arms, on the other hand, can be viewed as an experiment with one-sided non-compliance: assignment to the Choice arm exogenously shifts some borrowers (the choosers) into structured payments, since no one in the Control arm is treated. A comparison of the Choice and Structured payments arms can be viewed as a \emph{second} experiment with one-sided non-compliance: assignment to the Structured payments arm exogenously shifts some borrowers (the non-choosers) into structured payments, since everyone in the Structured payments arm is treated. Under an exclusion restriction, we can apply standard instrumental variables logic to our latter two comparisons. The Choice versus Control comparison identifies the effect of treatment on the treated (TOT), while the Structured payments versus Choice comparison identifies the effect of treatment on the untreated (TUT). This single experiment thus reveals the selection on gains (TOT$-$TUT) as well as the magnitude of impacts among non-choosers---the keys to understanding when voluntary policies suffice versus when mandates might be justified.

Our ATE results show that mandatory structured payments are strongly effective in lowering financial costs to borrowers and preventing default in pawnshop lending. The average borrower in the Structured payments arm is 7.7 percentage points less likely to default (17.5\% of the mean), and pays financing costs inclusive of fees that are 19\% lower than the control. 
In short, structured payments save borrowers money by charging them fees. The mechanism by which our monthly payment contract achieves these cost savings is intriguing: not only does it improve repayment and speed up the pace of payments on average, but it also decreases the amount of money paid by borrowers who ultimately default. While borrowers in the Structured payments arm are 6 percentage points more likely to borrow again from the same lender, a back-of-the-envelope calculation suggests that the structured payments contract nevertheless lowers lender profits by 12\%--19\% per borrower. It is thus unsurprising that structured payments contracts were unavailable in Mexico City pawnshops outside of our experiment: lenders profit from default.\footnote{Unlike most pawn lenders, our partner operates as a non-profit, donating all of their earnings to philanthropic causes.}

Despite these large financial cost savings, only 11\% of borrowers in the Choice arm choose structured payments. 
Does this low take-up reflect efficient selection into treatment---where ``choosers'' are the only borrowers who experience cost savings from structure despite the large ATE? 
Our estimated TUT effects suggest that the answer is no.
The causal effect for non-choosers is even larger than the ATE: structure lowers their financial costs by 192 pesos, a 6.8 percentage point reduction in cost as a proportion of the loan.
We cannot reject the null of no selection on gains.

Could our average TUT effect mask substantial heterogeneity, with some non-choosers incurring large financial costs when structure is imposed on them? We explore this question using two approaches. First, non-parametric bounds on the distribution of individual treatment effects, following \cite{fan2010sharp}, show that at least 28\% of borrowers see reduced costs---almost three times the 11\% who chose structure. Second, causal forest estimates using baseline survey data, following the approach of \cite{atheygrf}, reveal that only 1\% of borrowers have negative estimated conditional TUT effects. It is difficult to identify groups who would experience higher financial costs under mandatory structure, even among non-choosers, deepening the puzzle of low take-up.

So why do borrowers seem to be leaving money on the table?  
One explanation could be impatience: a sufficiently impatient borrower might prefer the status quo contract despite its higher probability of default. We find, however, that implausibly large discount rates (over 1000\%) would be needed to rationalize this preference.  An analysis that incorporates travel costs, the subjective value of pawns, and the opportunity cost of liquidity does not overturn the financial cost savings from the product. This motivates us to explore two behavioral mechanisms: present bias and overconfidence. Using measures from our baseline survey, we find that neither mechanism predicts take-up of structured payments, but subjective certainty about pawn recovery predicts the strength of the TUT effect. To better understand this finding, we provide a simple model illustrating how the state-contingent features of the structured payments contract could interact with overconfidence to generate heterogeneous TUT effects in the face of low take-up. While we cannot definitively identify the mechanism behind the mandate-choice gap that we document in our experiment, our results suggest that overconfident borrowers are unwilling to use a tool designed to prevent default because they incorrectly believe that they are not at risk. We discuss this interpretation in light of the existing literature, along with limitations of our present-bias measure and the potential relationship between overconfidence and present bias \citep{allcott2022high}.

We make several contributions to the literature. First, we identify \emph{both} the TOT and TUT effects in a single experiment, without the need for auxiliary structural modeling assumptions or incentivized elicitation of willingness to pay. Our three-armed experimental design also identifies the average selection on gains, average selection bias (ASB)--the difference in untreated potential outcomes for choosers versus non-choosers--and the average selection on levels (ASL)--the analogous comparison for treated potential outcomes. Like us, \cite{fowlie2021default} employs a one-stage, three-armed experimental design. Because they identify two TOT effects for different treated groups, however, rather than a TUT and TOT effect, their design cannot speak to the question of selection on gains. The two-stage designs of \cite{KarlanZinman2009} and \cite{Beaman2023} likewise do not identify TUT effects. 

Closest in spirit to our approach are studies that link incentivized elicitation of willingness to pay (WTP) with random assignment, formalized as ``selective trials'' by \citet{chassang2012selective}. This approach has been used to study the introduction of new technologies such as water filtration \citep{berry2020eliciting} and energy-efficient cookstoves \citep{berkouwer2022credit}, to examine WTP for exercise incentives \citep{CarreraEtAl2022}, and to study the welfare effects of social media \citep{allcott2020welfare}. \citet{allcott2022high} and \citet{bronchetti2023attention} use a related approach that elicits WTP via incentivized price lists but assigns treatment primarily via standard random assignment: only a small fraction of subjects are treated based on their WTP. Our design can be viewed as a simpler variation on these approaches. Rather than eliciting willingness to pay, we compare mandatory assignment with voluntary choice when structured payments are offered free of charge---the relevant policy counterfactual in our setting---so that take-up is measured from an actual contract choice, with no need to elicit valuations or to assume that elicited valuations would predict real-world take-up. Our approach therefore collects less information than these more involved designs, identifying treatment effects at a single point on the demand curve rather than across the distribution of valuations. Its advantage is simplicity, which can be crucial for large scale implementation in field settings such as ours. We view the two approaches as complementary.

The Mandates versus Choice design could be useful in other circumstances where researchers aim to test whether self-selection is efficient: whether the ``right people''---those with the largest treatment effects---choose treatment. This question arises when designing educational interventions (mandatory versus optional homework assignments), employer-provided retirement plans (opt-in or automatic enrollment), and when studying adherence to medical treatments.
By point-identifying both the TUT and TOT effects, our design directly addresses the question of how treatment effect heterogeneity relates to real-world compliance decisions.

Second, we contribute to the relatively small literature on paternalism and selection-on-gains. \cite{Laibson2018} studies private paternalism from a theoretical perspective and explains how principals may conceal commitment features that agents need but do not demand. Our results suggest that this picture may be inverted in the upside-down world of pawn lending: principals provide unstructured contracts that \textit{induce} default, though borrowers would face lower borrowing costs with more structure. Relatedly, \cite{heidhues2010exploiting} present a model where lenders design a contract that ``is flexible in a way that induces the borrower to unexpectedly change her mind regarding how she repays'', resulting in postponed payments, higher financial costs, and lower welfare.
On the empirical side, \cite{Sprenger} show that individuals with the most time-inconsistent preferences are the least likely to demand commitment. Our papers are complementary; while they design the experiment to identify one particular mechanism generating low take-up, our design focuses on quantifying treatment effects for choosers and non-choosers. We measure the financial cost savings that non-choosers forgo while remaining largely agnostic on the mechanism. Relatedly, whereas \cite{Walters} combines a distance-based instrument with structural modeling assumptions to show that students who select into more effective schools have smaller treatment effects ($\text{TOT} < \text{TUT}$), we identify TUT and TOT effects without the need for a structural model.

Third, our paper provides a detailed picture of pawn lending, a widespread and important financial service that has received little scholarly attention despite its prevalence and impact on economically vulnerable borrowers.\footnote{\cite{Caskey1991} complained that ``economists have produced hundreds of theoretical and empirical studies of banks while completely ignoring pawnshops'', which, according to him, are used by 10\% of Americans. \cite{Bos2012} write that ``while controversy surrounds payday loans, pawnshops have stayed out of public scrutiny''.} \cite{oeltjen1991pawnbroking} gives a history of pawn lending, while \cite{Caskey1991} and \cite{Bos2012} describe the industry in the US and Sweden. None of these papers estimates effects of contractual features on pawn lending outcomes. An older literature considers the exploitative potential of over-collateralization \citep{basu1984implicit}, but the behavioral implications of such contracts remain largely unexplored. 

Fourth, we contribute to the literature on the effects of structure and plans by showing that payment structure works even with modest fees \citep{nickerson2010voting,milkman2011using,milkman2013planning}. Consistent with \cite{beshears2016beyond}, a modest fee could circumvent limited attention, set a reference point or a mental goal, fight forgetfulness, serve as a recommendation to act, or increase the salience of deadlines.

Fifth, our paper contributes to a large literature on commitment devices: our Structured payments arm adds commitment features to an already high-interest loan. \cite{CarreraEtAl2022} lists 14 penalty-based commitment studies, and finds take-up rates ranging from 11\% to 73\%, with an average of 22\%. Our 11\% take-up rate is at the low end of this range.\footnote{\label{commitment_footnote} \cite{Ashraf}, \cite{Gine}, \cite{Ted}, \cite{Royer}, \cite{Sprenger} find modest take-up of commitment, similar to our results, while  \cite{Kremer}, \cite{Casaburi}, \cite{Alcohol}, \cite{AprajitP&P}, and \cite{Pascaline} find more robust take-up.} Unlike existing commitment studies, we estimate the effects of structure separately for borrowers who would and would not choose it. \cite{CarreraEtAl2022} documents over-optimism in people's gym attendance and low correlation between measured present bias and commitment take-up, both of which are consonant with our findings. 

An important point of comparison for our work is \cite{allcott2022high}, who study US payday borrowers' demand for a commitment-style incentive to avoid future borrowing. Both their study and ours reveal a mismatch between demand and treatment effects, but in opposite directions. In \citeauthor{allcott2022high}'s setting, only 3\% of loans end in default and the payday borrowers who self-select into the experiment are fairly sophisticated: their predictions about future borrowing are only 4 percentage points below the actual rate. This high degree of borrower sophistication is central to the key result of the authors' structural model, namely that banning payday loans or imposing tighter loan-size caps is unlikely to improve welfare. While borrowers value the incentive, it is less effective than they expect. In stark contrast, 44\% of pawn loans in our experiment end in default under the status quo but borrowers are highly overoptimistic, predicting on average a 92\% probability of recovering their pawn. And while voluntary take-up of structured repayment is low in our setting, mandated structure substantially reduces both default rates and financial costs. The contractual features of pawn lending may make overconfident borrowers especially responsive to late fees and imposed structure. \citeauthor{allcott2022high}'s borrowers want commitment but do not benefit from it; ours do not want structure even though it reduces default and financial costs.

Finally, we contribute to a growing literature on the effects of payment frequency in lending. In the context of microfinance, \citep{Pande} find no effects from frequent repayment schemes on loan default. Other studies find that increasing repayment flexibility in microfinance (by giving grace periods to pay or options to delay payment) improves business performance \citep{mcintosh2008estimating,Field,barboni2023flexible,flexibilty2024}. Our experiment, in contrast, highlights some of the dangers of flexibility in the context of overcollateralized pawn lending. See \cite{vihriala2023self,BernsteinKoudijs2024} for studies of payment structure in mortgage lending.

The remainder of the paper is organized as follows. Section \ref{context} provides background on pawn lending and defines our main outcome variables. Section \ref{Design} describes our experiment and data sources, and shows pre-treatment balance across arms. Section \ref{sec:randchoice} formally presents the Mandates versus Choice design and establishes identification of treatment effects for choosers (TOT) and non-choosers (TUT), along with average selection on gains and related selection measures. Section \ref{Experiment} presents average treatment effects on financial costs and their components, then applies the Mandates versus Choice design to identify treatment effects for choosers and non-choosers and test for selection on gains.
Section \ref{sec:harm} explores heterogeneity in TUT effects, using partial identification bounds and causal forests to assess whether causal effects may be negative for some non-choosers. Section \ref{why_paternalism} explores potential mechanisms for our results, and Section \ref{conclusion} concludes.

\section{Context} \label{context}

\subsection{Pawnshop borrowing}
    
Pawn loans involve individuals leaving valuable liquid assets, typically jewelry, as collateral in exchange for an immediate cash loan. The collateral is typically more valuable than the loan amount, allowing lenders to give the loan immediately without checking a borrower's credit history. This makes pawn loans a popular way to get cash to pay for emergencies. In fact, they are one of the most prevalent forms of borrowing. There are more than 11,000 pawnshops across the US, with 30 million clients and \$14 billion in yearly revenues.\footnote{See the website for the National Pawnbrokers Association (\href{https://tinyurl.com/ybm56dpe}{https://tinyurl.com/ybm56dpe}),  information on the items pawned by Americans (\href{https://tinyurl.com/y9zdcgws}{https://tinyurl.com/y9zdcgws}), and a Bloomberg article on regulation of the pawn industry in China (\href{https://tinyurl.com/y59ptdam}{https://tinyurl.com/y59ptdam}).}  
Our partner pawn lender in Mexico alone served more than 1 million clients in the last 3 years with more than 4 million contracts. For comparison there were 2.3 million microfinance clients across all lenders in Mexico in 2009 \citep{Pedroza:2010}. 

Pawning is also one of the oldest forms of borrowing. Pawn lending existed in antiquity at least since the Roman Empire, and there are records of it in China about 1,500 years ago \citep{PawnShops}. In spite of its prevalence and long history, pawnshop borrowing has not received much attention in the economics literature. The closest widely studied product is payday lending. In developing countries, however, payday lending is arguably less important than pawnshop lending; the latter is faster and requires less documentation, making it more accessible to informal sector workers who receive their salaries in cash.  

To situate pawn loans relative to other, more widely studied credit products, it is useful to highlight their distinctive institutional features. While pawn loans often serve different clientele than credit cards \citep{Caskey1991, Bos2012}, they share some of their flexibility. Unlike credit cards and most lines of credit, however, pawn loans do not impose minimum payments. Pawn loans also differ from payday loans. Payday contracts typically impose short repayment deadlines (the next pay date), rollover fees, and default consequences such as collections and credit reporting. Pawn contracts, in contrast, provide a true walk-away default option.
These contractual differences translate into sharply different default outcomes. \cite{allcott2022high} finds that only 3\% of payday loans end in default, whereas 44\% of pawn loans in our sample do. Microfinance loans typically require frequent installments, often weekly or monthly; pawn loans impose no such structured repayment schedule. This repayment flexibility likely contributes to pawn loans' substantially higher default rates relative to microfinance. Finally, unlike other overcollateralized loans such as mortgages, pawn loans involve strict forfeiture: if the borrower defaults, the lender retains both the pledged collateral and any partial payments already made.

Like payday lending, pawnshop lending is controversial \citep{carter2012pawnshops}. Regulators worry that borrowers using these products may lack financial sophistication, leading them to make suboptimal choices---biases that contract design may exacerbate. Evidence of such concerns exists for payday borrowers, but the pawn lending industry has received little empirical scrutiny despite serving similar clientele.\footnote{The US Congress has banned payday lenders from serving active military personnel, and some states have banned the industry altogether \citep{Payday}. \cite{Bertrand} find that simply disclosing how financing costs accumulate reduced payday loan demand by 11\%, suggesting borrowers underestimate total costs. \cite{Meltzer} finds that payday loan access increases difficulty paying mortgage, rent, and utility bills.}

\subsection{Pawning Logistics and Contracts}
\label{sec:logistics}

To fill this gap in the literature, we partnered with one of the largest pawn lenders in Mexico, an institution with more than one hundred branches spanning multiple states. This lender (whom we refer to as ``Lender P'') has a simple and typical business model. 

\paragraph{Appraising and Lending.} Lender P takes gold jewelry as collateral in exchange for a fraction of the value of the piece, in cash. No other collateral and no credit history checks are required. The transaction takes less than 10 minutes and is conducted at the branch in person between the client and the appraiser (i.e.\ a teller).
The appraiser weighs the gold piece and runs tests on its purity. Based on these tests, she assigns a gold value to the piece, stores it as collateral, and gives the client 70\% of the gold value in cash. The borrower signs a 2-page contract with the conditions of the loan, receives a pawn receipt with amounts owed and date the loan is due, and leaves with the cash. 

\paragraph{Status Quo Contract} Before our experiment, Lender P had only one type of contract, henceforth the \textit{status quo} contract. The interest rate was 7\% \textit{per month} compounded daily on the outstanding amount of the loan. The loan had a 90-day term with a 15-day grace period. The client was free to make early payments at the branch at any time with no penalty, but these prepayments were not required. Under this status quo contract, there were no payment reminders or any other kind of interim contact between the lender and the borrower. Internal lender rules allocated payments first to interest and fees owed, and to capital only after the first had been covered. If the client returned to pay the principal plus the accumulated interest within 105 days, she recovered her pawn; otherwise the pawnbroker kept the piece \textit{and} any payments already made. Before the contract expired, the client had the right to renew for another three months by going to the pawnshop, paying the accumulated interest, and signing a new contract with exactly the same terms and the same piece as the original contract (38\% of borrowers renewed at least once with a given pawn). This contract was standard in the industry. Pawnshops make money in three ways: by reselling the jewelry left as collateral on defaulted loans, by charging interest on non-defaulted loans, and by keeping the payments made on defaulted loans.

\subsection{Measuring Borrowers' Financial Costs} 
\label{costs}

Borrowers' financial costs have three main components: the value of lost collateral on defaulted loans, interest charges, and penalty fees.
For each loan in our sample, we record whether the client lost her pawn, $\mathds{1}(\text{Default}_i)$.\footnote{In our experiment, 13\% of loans are ongoing at the end of our observation window. For our main results in Section \ref{sec:baseline}, we classify these loans as not having defaulted. As explained below, this approach is conservative in our context: it biases treatment effects toward zero. Online Appendix Table \ref{bounding_censoring} presents results under alternative assumptions on outstanding loans.} 
Our administrative data records all payments made by borrowers, classified according to Lender P's payment allocation rules: payments to principal $P^C$, payments on interest $P^I$, and payments on penalty fees $P^F$. We observe the amount and date of each payment in all three categories.

We define a borrower's financial cost as the total monetary outflow---in cash or pawn value---from the borrower to the lender. This includes all payments the borrower made toward interest and fees, but also the net difference between the appraised value of the pawn and the loan amount (Value $-$ Loan) in the event of default. When there is no default, the borrower recovers her pawn and there is no loss of value for the borrower. Payments toward capital are considered a cost only when the borrower defaults, as she is not reimbursed for these payments. Note, however, that when she does not default, payments to capital net to zero against the initial loan disbursement. The formula for financial cost for person $i$ is thus as follows: 
\begin{align*}
    \text{Financial Cost}_i =&  \sum_t P^I_{it} +\sum_t P^F_{it}  
     + \mathds{1}(\text{Default}_i) \times \left(\text{Value}_i-\text{Loan}_i + \sum_t P^C_{it}\right)
\end{align*}
where $t$ indexes days, and $\mathds{1}(\text{Default}_i)$ is the indicator of default. Because the period of the loan is only 90 days, we do not apply discounting in calculating costs. In robustness checks reported below we show that our results are virtually unchanged when applying a wide range of time discounting factors.\footnote{Additional robustness checks discussed in Section \ref{sec:baseline} and Online Appendix Table \ref{table_robustness_fc} below incorporate: (i) borrowers' subjective valuations of their pawns in place of appraised gold value, (ii) travel expenses and opportunity cost of time for in-person payments, and (iii) the liquidity cost of prepayments.}

As a second measure of cost we calculate the Cost Ratio (CR): financial cost divided by loan size, representing the fraction of borrowed capital paid in costs. We report costs as levels and ratios rather than annualized rates. There is no standard approach to annualizing total costs in our setting that would not inappropriately conflate interest and fees (a flow that accumulates with loan duration) and collateral losses in default (a stock reflecting the initial loan-to-value gap that does not vary with the default date).\footnote{Two borrowers who default on loans with identical collateral incur identical collateral losses regardless of whether one defaults sooner than the other. Both have effectively ``sold their pawn'' on the day when they signed their loan agreements.}

\section{Experiment, Data and Stylized Facts} \label{Design}

\subsection{Treatment arms and randomization}
\label{sec:treatment-randomization}

\paragraph{The Structured contract.} For the experiment, we designed a structured payments contract informed by micro-lending practices \citep{morduch1999microfinance, bauer2012behavioral}. The contract is identical to the status quo in most respects: it has the same interest rate (7\% per month), the same loan size/collateral ratio (70\%), the same loan term (90 days with a 15-day grace period), and the same appraisal procedures. The key difference is that the structured payments contract requires three equal monthly payments covering both principal and interest, due on days 30, 60, and 90 after loan disbursement. 
Payment deadlines were clearly specified in the contract and on payment receipts. To make these deadlines salient, borrowers who missed a monthly payment incurred a modest late fee equal to 2\% of the payment due.\footnote{The fee was modest and intended to make the payment deadlines salient. As a benchmark, the transportation cost to visit the branch to make a payment is comparable to the fee, on average. The level of the fee resulted from discussion and focus groups. Too large a fee could generate illiquidity and hurt borrowers who only partially adhere \citep{john,Ted,BeshearsEtAl2025, Gine}.} The fee and its enforcement by the lender may serve as a commitment device. The payment schedule may also encourage borrowers to form a mental payment plan.\footnote{In other domains \cite[e.g.,][]{beshears2016beyond} it has been shown that such plans can circumvent limited attention, set a reference point, and be a cue to act.}

To measure demand for structured payments, we include a \textbf{Choice arm} in which borrowers can opt into the structured contract. The combination of mandatory and voluntary structured payment arms is essential for identifying the range of treatment effects we describe in the next section.

\paragraph{Treatment Arms.} 
Borrowers were randomized at the branch-day level to one of the three treatment arms:

\begin{enumerate}
    \item \textit{Control} arm: borrowers were offered only the status quo contract. 
    \item \textit{Mandatory structured payments} arm:  borrowers were required to use the structured payment contract.  
    \item \textit{Choice} arm: borrowers could choose between the structured payments contract and the status quo contract.
\end{enumerate}

\begin{figure}[htbp]    
 \caption{Experiment description}
\begin{center}
\includegraphics[width=0.9\textwidth]{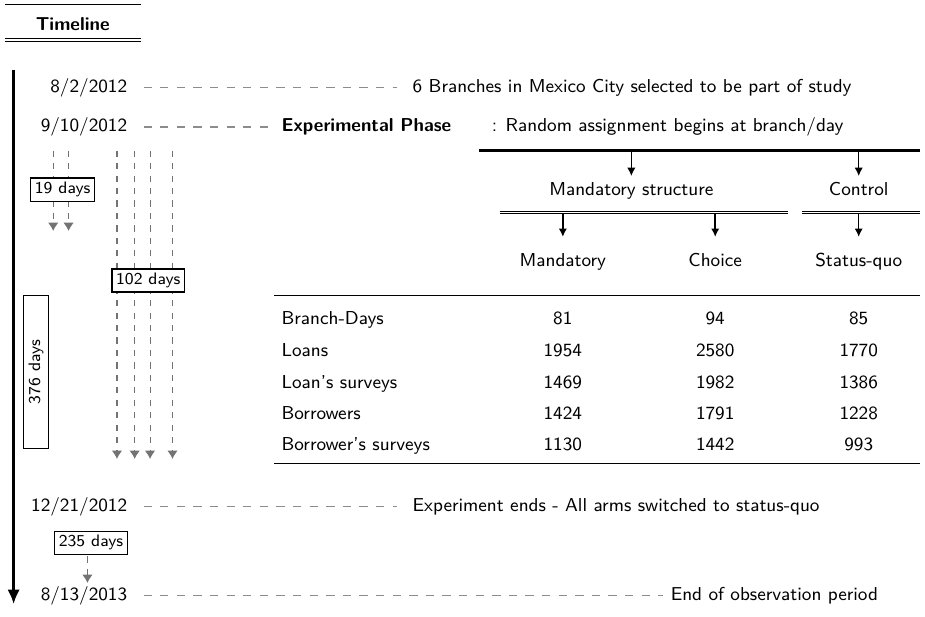}
  \end{center}
     \label{exp_description}
\end{figure}

\paragraph{Randomization.} We implemented the experiment in six branches of Lender $P$ beginning on September 6, 2012. The branches were selected by Lender $P$ to be dispersed across Mexico City and have varying sizes. In four of them the experiment ran for 102 days; in 2 of them it ran for a shorter time to economize on data collection costs. 
Branches are more than 5 km apart from each other, and there is no substitution between them; no borrower appeared at more than one of our study branches.

Each day, a computer randomly assigned which contract would be available in each study branch, and the branch IT system could only offer the assigned contract. We did not allocate an equal number of days across arms, since we were interested in having more power in some of them. We allocated 85 branch-days to control, 81 to mandated structured payments, and 94 to choice. See Figure \ref{exp_description} for a CONSORT (Consolidated Standards of Reporting Trials) diagram of the study design and recruitment. The study was designed in 2012 before pre-registration was common in Economics, and hence does not have a pre-analysis plan.

Branch personnel did not know the daily treatment assignment or the study's objective. They were told that there were 3 different ``types of contract-days,'' that the system assigned randomly each day, and that consecutive days could receive the same contract assignment. They were also told that this procedure was in effect at several of Lender P's branches (they did not know which ones) and would continue for several months.

Some clients pawned on multiple occasions during the experiment: 14\% pawned twice in study branches and 8\% pawned three or more times. To have a clean comparison, we consider only the first pawn conducted during the experimental window. Moreover, 21\% of borrowers' first visits involved multiple loans because borrowers submitted two or more pieces of gold. We treat each piece as a separate loan. In the appendix we show that our results are robust to this analysis choice.

\paragraph{Promise Arms}
The experiment included two additional independent arms that involved a personal promise to pay but no pecuniary penalties. In one arm individuals were asked if they wanted to make such a promise, and in the other they were required to do so. These were intended to investigate the extent to which personal morality and promises can drive high-stakes behavior. Both promise arms proved ineffective and so we do not highlight them here. In the results section we discuss what we can learn from the fact that the structured payments arm was effective while the promise arms were not. 

\paragraph{Timeline.} Figure \ref{exp_description} shows the experimental timeline along with the length of time for which we observe payments. For loans made in the first week of the experiment, we observe up to 338 subsequent days of loan information; for loans made in the last week we observe up to 235 days. Figure \ref{exp_description} also illustrates the number of branch-days per arm, the number of loans, and the number of surveys. 

\paragraph{Explaining the Contracts.}
We devoted substantial effort to helping clients understand the contract terms. First, full-time enumerators explained contract terms to clients. The explanation took about 3-5 minutes and continued until the client said she understood. Enumerators then asked clients to explain the contract back to them before correcting any misunderstandings. Second, the appraiser gave clients the ``Contract Terms Summary'' and read it out loud to them when their piece had been appraised but before they signed the contract. 

After the borrower pawned the piece, she received a contract receipt. Those assigned to the structured payments arm received a receipt indicating the contract terms and specifying a late payment fee of 2\% of the monthly payment amount. The receipt included a payment calendar with exact dates and amounts due each month (Figure \ref{PaperSlip}). Those in the status quo arm received analogous receipts that included only a final payment date.

\paragraph*{Appraisers and IT system lockout.} 
The IT system was hard-coded with the treatment assigned to a given branch on a given day; appraisers could not override the system. For this reason, we can be confident that borrowers in the mandatory arms received their experimentally-assigned treatments. 
To further ensure fidelity to the experimental protocol, our enumerators were present at each study branch from opening to close of business on every day of the experiment. They did not report a single instance of appraisers attempting to influence borrowers' decisions in the choice arm. This is unsurprising: since appraisers were paid a flat wage independent of the type or number of contracts that borrowers signed, they had no financial incentive to influence borrowers' choices. Consistent with this, adding or removing appraiser fixed effects leaves our results unchanged.
We cannot completely exclude the possibility that appraisers subtly influenced decisions in the choice arm. An advantage of our real-world setting, however, is that it features the same incentives that would be present in a full-scale implementation. As such, our results remain relevant for policymakers considering the wider roll-out of contract choice.

\subsection{Data}
\label{sec:data}

\paragraph{Administrative Data.}
The study exploits two types of data: administrative data from the lender and a short survey that we implemented. The administrative data contain a unique identifier for each client, an identifier for the piece she is pawning, and the transactions relating to that piece. These transactions include the value of the item as assessed by the appraiser, the amount of money loaned, the date of the pawn transaction, and the type of contract for that pawn: structured payments or status quo. During the loan period, we followed each transaction related to that piece in the administrative data: when payments were made and for what amounts, whether there was default (i.e.\ the client lost her pawn), and whether any late-payment fees were imposed. After each experimental loan, we were able to track subsequent behavior and see whether the borrower returned for another loan.
We have this information for all the pawns that occurred in the experiment's 6 branches between August 2, 2012 and August 13, 2013. This includes all the pawns that took place during our experiment along with those that originated one month before and eight months after our experiment. 
The experiment comprises 6,304 pawns while our administrative data cover a total of 26,180 pawns.

\paragraph{Survey Data.} An additional team of enumerators stationed at the entrance of each study branch asked clients to complete a 5-minute survey before going to the teller window to appraise their piece and \textit{before} they learned which contracts would be available on that day. The survey allowed us to measure loan take-up: we can identify borrowers who entered the branch and completed the survey but then decided not to pawn after learning the contract terms. Only 4\% of respondents declined to borrow, suggesting that the decision to pawn is made before entering the branch and is insensitive to contract terms.
The survey was intentionally brief to avoid interfering with branch operations. We collected survey data at the pawn level. Some borrowers completed multiple surveys across repeated pawns. The survey measured demographics, income and education, present bias, prior pawn experience, subjective probability of recovery, and subjective valuation of collateral, among other variables.\footnote{The complete survey instrument appears in the Online Appendix.} 
The survey data allow us to explore treatment effect heterogeneity in Section \ref{sec:RF} and behavioral mechanisms in Section \ref{why_paternalism}. For these sections we link the survey with administrative data, obtaining a 78\% survey response rate for our main sample and balance across arms. Importantly, our core experimental results from Section \ref{Experiment} do not depend on survey data, relying only on the administrative records described above. 

\subsection{Stylized facts}

\paragraph{Borrowers.} Borrowers are not inexperienced: 87\% of survey respondents report having pawned before. 
However, they are economically vulnerable:  20\% of them could not pay either utility bills or rent in the past 6 months. In terms of triggers and uses of the loan, 89\% said they were pawning because of an emergency, and only 12\% stated it would be used for a ``non-urgent expense''. When asked why they were pawning, 5\% cited a family member losing a job, 11\% a medical emergency, and 72\% an urgent expense. Borrowers were mostly women (73\%), with an average age of 43 years. Only 32\% reported education beyond high school.

\paragraph{Many borrowers lose their pawn.}
Default rates are high in our context: 41\% of clients lose their pawn within 230 days of pawning. The default rate for the experimental period and branches is similar to that observed for all branches in subsequent periods (see Figure \ref{external_validity}). One potential explanation for high default is that clients are knowingly selling their gold piece through a pawn contract on which they intend to default. This appears unlikely both because the reported subjective value of the pawn was larger than the appraised value for 83\% of clients, and because clients could easily sell the gold and obtain the full value of the pawn at gold-buying stores, which were located close to our partner lender's branches. 
Branches have an average of 2.8 such shops within a 15-minute walking distance.\footnote{There are 2.8 shops within a 1.25km radius of the branch on average.} 
Given that they would receive only 70\% of the pawn's value at the lender while shops next door would pay 100\%, we consider the ``intentional default'' explanation \emph{prima facie} unlikely.

\paragraph{Paying without recovering.} Among those who lost their pawn, 29\% paid a positive amount toward its recovery. On average this subset of borrowers paid 34\% of the value of their loan (see Figure \ref{proxy_naive} in Appendix), suggesting that they expected to recover. 

\paragraph{Overconfidence.} 
Our baseline survey asked borrowers for their subjective probability of recovering the pawn. The average self-assessed recovery probability was 92\%. This contrasts sharply with the actual non-default rate for these same control-group borrowers: 54\%. Overall 72\% of borrowers reported a 100\% probability of repaying their loan.\footnote{\cite{allcott2022high} find that inexperienced payday borrowers are overconfident about  getting out of debt in the future, but they do not study overconfidence in repaying the current loan.}

\subsection{No differential selection or imbalance}
\label{sec:integrity}

The validity of our experimental design depends critically on the absence of differential selection: borrowers must not selectively avoid the branch on days when certain contracts are offered. The overwhelming majority (96.1\%) of those who answered the baseline survey took a loan after learning the contract terms, and this rate is virtually identical across treatment arms: 96.7\% in the control arm versus 95.6\% and 96.2\% in the structured payments and choice arms, respectively (p=0.86). These identical take-up rates demonstrate that borrowers did not differentially select into treatment conditions. Online Appendix \ref{app:integrity} presents two further pieces of evidence---balance in the average number of borrowers, pawns per borrower, and pawns per day across arms, and balance in borrowers' baseline characteristics---that together constitute strong evidence of no differential borrower selection.

\section{The Mandates versus Choice Design}
\label{sec:randchoice}

The ``Mandates versus Choice'' design randomly assigns participants to one of three experimental conditions: mandatory control, mandatory treatment, or free choice between the two.
As we now present more formally, this design allows researchers to test a fundamental assumption of many economic models: that those who voluntarily choose a product or contract are those who benefit from it the most.

\begin{assumption}[Mandates versus Choice Design]
Each participant is randomized into one of three experimental arms: $Z_i \in \{0, 1, 2\}$. Those with $Z_i = 0$ receive the control $(D_i = 0)$, those with $Z_i = 1$ receive the treatment $(D_i = 1)$, and those with $Z_i = 2$ are free to choose either the treatment $(D_i = 1)$ or control $(D_i=0)$ conditions.
\label{assump:design}
\end{assumption}

In our empirical context $Z_i = 0$ denotes the mandatory Status quo arm, $Z_i = 1$ denotes the mandatory structured payments arm, and $Z_i = 2$ denotes the Choice arm. 
In Assumption \ref{assump:design}, $D_i$ denotes the treatment that participant $i$ actually \emph{received}, where $D_i = 0$ is the Status quo contract and $D_i = 1$ is the structured payments contract. 
Besides random assignment of $Z_i$, the key ingredient in Assumption \ref{assump:design} is perfect compliance in the $Z_i = 0$ and $Z_i = 1$ arms. 
As explained in Section \ref{sec:integrity} above, this holds in our context: only participants in the $Z_i = 2$ are free to choose between alternative contracts. 
In addition to Assumption \ref{assump:design} we assume that there are no spillovers from treatment.

\begin{assumption}[No Spillovers]
The stable unit treatment value assumption (SUTVA) holds: each participant's potential outcomes depend only on her own values of $Z_i$ and $D_i$, not those of any other participant. 
\label{assump:SUTVA}
\end{assumption}

Let $C_i \in \{0, 1 \}$ denote a participant's ``choice type.'' If $C_i = 1$ then she \emph{would choose structure}, given the option; if $C_i = 0$ she would not. 
We call participants with $C_i = 1$ ``choosers'' and those with $C_i = 0$ ``non-choosers.''
Note that a participant's choice type $C_i$ is only observed if she is allocated to the choice arm ($Z_i = 2$).
Under Assumption \ref{assump:design}, 
\begin{equation}
D_i = \mathbbm{1}(Z_i \neq 2) Z_i + \mathbbm{1}(Z_i = 2) C_i.
\label{eq:potentialTreatments}
\end{equation}
Under Assumption \ref{assump:SUTVA}, a fully general model for the potential outcomes would take the form $Y_i(d, z)$ for $d\in \{0,1\}$ and $z \in \{0, 1, 2\}$.
The following assumption restricts the potential outcomes to depend on $d$ only.

\begin{assumption}[Exclusion Restriction]
Each participant's potential outcomes depend only on the treatment she actually received $(D_i)$, not on her experimental arm $(Z_i)$, i.e.\ $Y_i(d,z) = Y_i(d) \equiv Y_{id}$ for $z \in \{0, 1, 2\}$ and $d\in \{0, 1\}$.
\label{assump:exclusion}    
\end{assumption}

Assumption \ref{assump:exclusion} has the same mathematical structure as the familiar LATE exclusion restriction but its interpretation is somewhat different: there is no explicit distinction between ``chosen'' versus ``mandatory'' treatment in the usual LATE setup.
Assumption \ref{assump:exclusion} implies that being assigned a particular treatment has the same result as choosing it for yourself.\footnote{Technically, our results only require that the exclusion restriction holds when individuals are assigned the same treatment that they would have chosen freely: $Y_i(0,2)=Y_{i0}$ and $Y_i(1,2) = Y_{i1}$.}
If the mere fact of being given a choice directly affects a person's potential outcomes, our exclusion restriction fails. 
For example, someone who would voluntarily choose drug rehabilitation might respond differently when mandated to attend the same treatment program. 
In our experiment, however, assumption \ref{assump:exclusion} is plausible.\footnote{Closely related assumptions are common, if often tacit. \cite{chamberlain2011bayesian} explicitly assumes that choosing and being assigned a treatment have the same effect. A significant literature using compulsory schooling laws to estimate the returns to education tacitly assumes that schooling \emph{in general} has the same returns, regardless of whether it is mandated or chosen freely.}
Moreover, it has testable implications that we fail to reject in our application.\footnote{Details available upon request.}

Assumptions \ref{assump:design}--\ref{assump:exclusion} imply that the observed outcome $Y_i$ is related to $(Y_{i0}, Y_{i1})$ by 
\begin{equation}
    Y_i = \mathbbm{1}(Z_i =0) Y_{i0} + \mathbbm{1}(Z_i = 1)  Y_{i1}  + \mathbbm{1}(Z_i = 2) \left[(1 - C_i) Y_{i0} + C_i Y_{i1} \right].
\label{eq:potentialOutcomes}
\end{equation}
Equation \ref{eq:potentialOutcomes} is the key to understanding the mandates versus choice design. It shows how observed outcomes decompose into potential outcomes based on treatment assignment and choice type. 
Random assignment of $Z_i=0$ and $Z_i = 1$ identifies the marginal distributions of $Y_{i0}$ and $Y_{i1}$ for the population as a whole, point identifying the average treatment effect (ATE).\footnote{The ATE is point identified regardless of whether the exclusion restriction in Assumption \ref{assump:exclusion} holds.}
Random assignment of $Z_i=2$ likewise point identifies the share of choosers ($C_i = 1$), the distribution of $Y_{i1}$ for choosers, and the distribution of $Y_{i0}$ for non-choosers ($C_i = 0$).
As a direct consequence, our design point identifies a number of economically-relevant causal quantities beyond the ATE. 
First among these are TUT and TOT effects: $\text{TUT} \equiv \mathbbm{E}(Y_{i1} - Y_{i0} | C_i = 0)$ and $\text{TOT} \equiv \mathbbm{E}(Y_{i1} - Y_{i0} | C_i = 1)$.
In our experiment the TUT is the effect of structure on borrowers who would not voluntarily choose it; the TOT is the effect on borrowers who would choose structure.
By separately identifying these effects, we can assess whether borrowers with the largest cost reductions self-select into structured payments.
We return to this question in Section \ref{TuT_ToT} below.

To see why these effects are identified, notice that the Mandates versus Choice design can be seen as a \emph{pair} of randomized encouragement designs: RCTs subject to one-sided non-compliance. 
The first compares $Z_i = 1$ to $Z_i=2$ and identifies the TUT.
In particular, by a standard argument (see Online Appendix \ref{sec:proofs})
\begin{equation}
\frac{\mathbbm{E}(Y_i|Z_i=1) - \mathbbm{E}(Y_i|Z_i =2)}{\mathbbm{E}(D_i|Z_i=1)-\mathbbm{E}(D_i|Z_i=2)} = 
\frac{\mathbbm{E}(Y_i|Z_i=1) - \mathbbm{E}(Y_i|Z_i =2)}{1 - \mathbbm{E}(D_i|Z_i=2)} = \text{TUT}
\label{eq:TUT}
\end{equation}
since everyone with $Z_i=1$ is treated.
Similarly, since no one with $Z_i = 0$ is treated, a comparison of $Z_i=0$ to $Z_i = 2$ identifies the TOT:
\begin{equation}
\frac{\mathbbm{E}(Y_i|Z_i=2) - \mathbbm{E}(Y_i|Z_i =0)}{\mathbbm{E}(D_i|Z_i=2)-\mathbbm{E}(D_i|Z_i=0)} = 
\frac{\mathbbm{E}(Y_i|Z_i=2) - \mathbbm{E}(Y_i|Z_i =0)}{\mathbbm{E}(D_i|Z_i=2)} = \text{TOT}.
\label{eq:TOT}
\end{equation}
The expression for the TOT in \eqref{eq:TOT} divides by the share of choosers, $\mathbbm{E}(D_i|Z_i=2)$. Without this adjustment, $\mathbbm{E}(Y_i|Z_i=2) - \mathbbm{E}(Y_i|Z_i=0)$ gives the intent to treat (ITT) effect of the structured contract, i.e.\ the causal effect of \emph{offering} this contract.\footnote{Note that identification of the ITT does not rely on assumption \ref{assump:exclusion}.} Figure \ref{tot_tut_graph} provides graphical intuition for equations \eqref{eq:TUT} and \eqref{eq:TOT}.

Because our design identifies both the TUT and TOT, it allows us to calculate their difference: the average selection on gains $\text{ASG} \equiv \text{TOT} - \text{TUT}$.
In a canonical Roy model $\text{TOT} > \text{TUT}$ so the ASG would be positive.
The Mandates versus Choice design allows us to test this implication directly.
Identifying both TOT and TUT effects in the same setting typically requires strong assumptions. One approach combines an excluded instrument with additional structural modeling assumptions \citep{cornelissen2018benefits,Walters}.\footnote{While the marginal treatment effects (MTE) approach \citep{heckman2007econometric} can in principle be used to identify the TOT and TUT with an exclusion restriction alone, this requires instrumental variable $Z$ with sufficiently rich support that the probability of treatment take-up given $Z$ varies continuously between zero and one. In practice, instrumental variables are usually discrete and, even when continuous, typically have a more modest effect on take-up.} 
An alternative calculates conditional local average treatment effects (LATE) given observed covariates $X$ and re-weights them according to the distribution of $X$ in some population of interest \citep{aronow2013beyond,angrist2013extrapolate}. 
While this ``LATE and re-weight'' strategy avoids parametric modeling assumptions, it relies on a strong assumption. Specifically, it assumes that there is no selection-on-gains ($\text{ASG} = 0$) conditional on $X$. 
Unlike both of these approaches, the Mandates versus Choice design identifies the TOT and TUT from assumption \ref{assump:exclusion}. It requires neither additional parametric restrictions nor the assumption of no unobserved selection-on-gains.

\begin{figure}
\caption{Graphical intuition for the Mandates vs.\ Choice design} 
    \begin{center}
        \centering
        \includegraphics[width=1.0\textwidth]{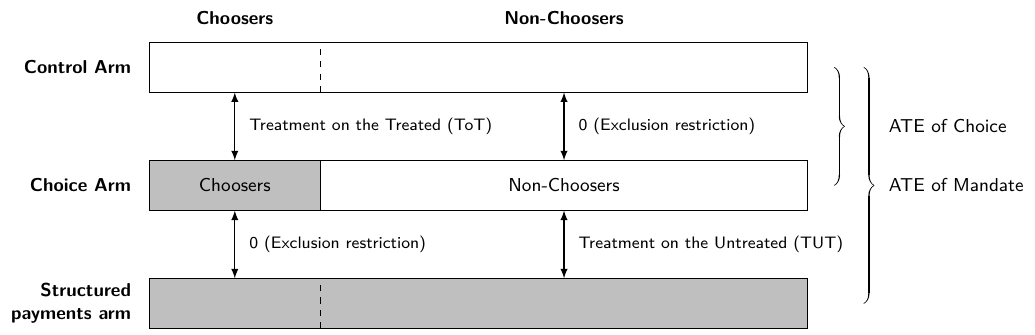}
    \end{center}
 \scriptsize{
    Notes: Gray rectangles denote borrowers with a structured contract; white rectangles denote borrowers with a status quo contract. 
    A comparison of means across control and structured payments arms identifies the ATE of Structure. 
    The difference of mean outcomes across the choice and control ``nets out'' the non-choosers, and hence equals the TOT multiplied by the share of choosers. 
    Similarly, the difference of means across the structured and choice arms ``nets out'' the choosers and equals the TUT multiplied by the share of non-choosers.
        The share of choosers, illustrated using dashed vertical lines, is equal on average across arms under random assignment.}
    \label{tot_tut_graph}
\end{figure}

In addition to the TUT, TOT and ASG, the Mandates vs.\ Choice design also identifies the average selection bias (ASB) and average selection on levels (ASL), namely 
\begin{align}
\label{eq:ASB}
    \text{ASB} &\equiv \mathbbm{E}(Y_{i0}|C_i=1) - \mathbbm{E}(Y_{i0}|C_i = 0) = \frac{\mathbbm{E}(Y|Z=0) - \mathbbm{E}(Y|Z=2,D=0)}{\mathbbm{E}(D|Z=2)}\\
    \label{eq:ASL}
    \text{ASL} &\equiv \mathbbm{E}(Y_{i1}|C_i = 1) - \mathbbm{E}(Y_{i1}|C_i=0) = \frac{\mathbbm{E}(Y|Z=2,D=1) - \mathbbm{E}(Y|Z=1)}{1 - \mathbbm{E}(D|Z=2)}
\end{align}
as shown in Online Appendix \ref{sec:proofs}.
These measures provide additional detail about the patterns of selection in our experiment.
Specifically, the ASB tells us whether borrowers who voluntarily choose structure are those with higher financial costs, on average, under the \emph{status quo}.
Similarly, the ASL tells us whether borrowers who voluntarily choose structure are those with lower financial costs, on average, under structured payments. 

The Mandates versus Choice design also point identifies $\mathbbm{E}(Y_{i0}|C_i,X_i)$ and $\mathbb{E}(Y_{i1}|C_i,X_i)$ for $C_i=0$ and $C_i=1$ where $X_i$ denotes any collection of pre-treatment covariates: see Lemma \ref{lem_randchoice} and \eqref{eq:Y0C1}--\eqref{eq:Y1C0} from the proof of Proposition \ref{pro:randchoice} in Online Appendix \ref{sec:proofs}.
This allows us to compute \emph{conditional} versions of the effects described above and to decompose these conditional treatment effects into their constituent parts.
We exploit this unique feature of the Mandates versus Choice design in our exploration of heterogeneity and behavioral mechanisms in Section \ref{why_paternalism} below.

\bigskip

\section{Results} \label{Experiment}

\subsection{Basic Treatment Effects and Take-up}
\label{sec:baseline}

We begin our presentation of results with a standard estimation of experimental assignment to the Mandatory structured payments and the Choice arms, relative to the Control arm. Table \ref{main_impact_table} presents estimates and standard errors from the regression 
\begin{equation} \label{basic_reg}
    y_{ij} = \alpha + \beta^{SP} T_{i}^{SP} + \beta^C T_{i}^C +  \varepsilon_{ij}
\end{equation}
where $i$ indexes the borrower, $j$ indexes the branch, $T_{i}^{SP}$ is an indicator for assignment to the structured payments arm, $T_{i}^C$ is an indicator for assignment to the Choice arm. 
Standard errors are clustered at the branch-day level, the unit of treatment assignment.
In this expression $\beta^{SP}$ is the ATE of the structured payments contract while $\beta^C$ is the ITT of structured repayment, i.e.\ the ATE of \emph{offering} this contract (see Section \ref{sec:randchoice}). 
Our two primary outcome variables are financial cost in pesos and the Cost Ratio (CR), as defined in Section \ref{costs}. 
Results for these outcomes appear in columns (1) and (7) of Table \ref{main_impact_table}.
The remaining columns decompose the financial cost into its components. Columns 2-4 capture payment flows: interest payments, payments of fees incurred, and principal payments. Column 5 shows the value of lost pawns conditional on default, column 6 uses an indicator of default as the outcome, and column 7 scales financial costs by loan size.\footnote{Loans can be extended in 3-month increments by paying accrued interest, maintaining the same treatment conditions. We track costs across the entire loan duration, including any extensions.}

Mandatory structure causes economically large and statistically significant decreases in borrowing costs, as measured by both financial cost and CR. Despite increasing fees, mandatory structure decreases borrowing costs by 183 pesos---a 19\% reduction from the control mean of 942 pesos---with a 95\% confidence interval of 83 to 283 pesos. These savings arise from two sources. First, the probability of default falls from a baseline of 44\% by 7.7 percentage points (95\% CI: 2.8 to 12.6 pp), generating savings of 91 pesos, i.e., about half of the total financial cost savings. Second, faster repayment reduces interest charges, as the structured payments contract causes borrowers to pay down the principal more quickly. 
These effects translate into a 6 percentage point reduction in CR from mandatory structure (95\% CI: 3.8 to 8.6 pp).

\begin{table}
\caption{Effects on financial cost}
\label{main_impact_table}
\begin{center}
\resizebox{1.0\textwidth}{!}{
\footnotesize{
\begin{tabular}{lcccccccc}
\toprule
      &       & \multicolumn{5}{c}{Components of FC}  &       &  \\
\cmidrule{3-7}      & FC    & Interest pymnt & Fee pymnt & Def$\times$Ppl pymnt & Lost pawn value & Default &       & CR \\
\cmidrule{2-2}\cmidrule{9-9}      &       & $\sum_t P^I_{it}$ & $\sum_t P^F_{it}$ & $\mathds{1}(\text{Def}_i)\times\sum_t P^C_{it}$ & $\mathds{1}(\text{Def}_i)\times \text{Value-Loan}_i$ & $\mathds{1}(\text{Def}_i)$ &       &  \\
\midrule
      & (1)   & (2)   & (3)   & (4)   & (5)   & (6)   &       & (7) \\
\midrule
\midrule
Mandatory structured & -183.0*** & -124.2*** & 32.4*** & -1.62 & -91.3*** & -0.077*** &       & -0.062*** \\
      & (50.9) & (40.0) & (1.54) & (3.14) & (34.5) & (0.025) &       & (0.012) \\
Choice  & -11.9 & 8.16  & 1.60*** & -0.36 & -21.7 & -0.040* &       & -0.0011 \\
      & (56.6) & (49.5) & (0.29) & (3.30) & (34.0) & (0.023) &       & (0.013) \\
      &       &       &       &       &       &       &       &  \\
\midrule
Observations & 6304  & 6304  & 6304  & 6304  & 6304  & 6304  &       & 6304 \\
R-squared & 0.005 & 0.005 & 0.147 & 0.000 & 0.002 & 0.004 &       & 0.016 \\
Control Mean & 942.4 & 545.9 & 0     & 5.96  & 396.5 & 0.44  &       & 0.43 \\
\bottomrule
\bottomrule
\end{tabular}%
}
}
\end{center}
\vspace{-.1in}
\scriptsize{Notes: This table shows the treatment effects for our core pecuniary outcomes. Each column presents results from a regression of the specified outcome variable on an indicator for the mandated and choice arms, following the specification from equation \ref{basic_reg}. Column (1) analyzes our core financial cost measure, while the columns 2-6 decompose these into their constituent parts.
Some borrowers take more than one loan on their first visit to an experimental branch. We treat these loans as separate observations. Additional results, available upon request, show that our results are robust to different ways of handling multiple loans for each borrower. Standard errors are clustered at the branch-day level. }

\end{table}

In contrast, the Choice arm generates no significant cost reductions. There is an increase in fees, an outcome for which we have high statistical power because fees are zero for every borrower in the control arm, and there is also suggestive evidence of a decrease in default probability (95\% CI: -8.5 to 0.5 pp). 
Confidence intervals for the Mandatory and Choice arms have similar widths, indicating that the null result in the Choice arm reflects a smaller estimated effect rather than insufficient precision.
Despite structure's large effects when mandated, the Choice arm generates no significant cost reductions due to low take-up: only 11\% of borrowers choose the structured payments contract.
As shown in Online Appendix Figure \ref{determinants_choose}, observables from our baseline survey are only weak predictors of take-up: borrowers who say they have trouble paying bills are less likely to choose structure while more educated borrowers are more likely.

\paragraph{Promise arms.} As a further point of comparison, Appendix Table \ref{sec_impact_table} presents results analogous to Table \ref{main_impact_table} for the promise arms.  In the equivalent of the mandated structured payments arm, borrowers were told to give us their word that they would pay one third of the loan amount each month. These are the same monthly payments as required by the structured payments contract. Promise arm borrowers, however, knew that they would face no lender fees for missed monthly payments.\footnote{Clients signed a document that said \textit{``I promise to pay every month the corresponding sum of \makebox[0.2cm]{\hrulefill}, on the dates \makebox[0.2cm]{\hrulefill}, \makebox[0.2cm]{\hrulefill}, and \makebox[0.2cm]{\hrulefill}. This is \underline{not} a legal document and cannot be used in court. It is just a personal promise. If you do not pay on time you will not have kept your word.''}
Dates were pre-printed (\ref{PaperSlip}).} 
In a second promise arm clients were given the option to make such a promise; 33\% chose to do so. 
Having borrowers promise to pay monthly generates economically small and statistically insignificant effects, whether the promise is chosen or required. For example, using the financial cost outcome, the estimated effect of the Promise Mandate is -18 pesos (95\% CI: -142 to 106) compared to -183 pesos for mandatory structure. More generally, confidence intervals for the Promise Arms are wide and include substantial effects in both directions. Personal promises alone are insufficient to change behavior; the combination of a pecuniary fee and externally imposed structure is responsible for our treatment effects.

\paragraph*{Payment speed} To shed light on the mechanisms behind these treatment effects,  Table \ref{mechanisms} examines intermediate outcomes. Panel A shows that mandatory structure accelerates repayment. Under the status quo borrowers wait an average of 83 days (out of a 90-day contract) before making their first payment; those in the structured payments arm make their first payment an average of 12.2 days earlier. These earlier payments are also 9.7\% larger (column 2). Moreover, around 37\% of borrowers in the structured payments arm pay in full on their first visit, compared to 30\% under the status quo contract (column 3). Overall, mandatory structure shortens loan duration by 22.6 days (column 4) and by 14 days conditional on recovery (column 5), noting that recovery is an endogenous control.

\paragraph{Binding monthly deadlines.} Figure \ref{porc_payment_over_time} (a) plots cumulative loan repayment over time by treatment arm. In the structured payments arm, a disproportionate number of payments occur near monthly payment deadlines. We see no such bunching of payments around these dates in the status quo arm. Differences in accumulated payments across arms seem to emerge precisely right around these dates.

\paragraph{Payment Bifurcation} Panel B reveals that mandatory structure reduces wasteful partial payments. Without structure, 12\% of borrowers make payments toward their loans but ultimately default, losing both their collateral and any payments made. Mandatory structured payments reduce the probability of this costly outcome by 6.7 percentage points (column 6). Among those who do default, borrowers in the structured payments arm pay 3.5 percentage points less of their loan amount (column 7) and are 12 percentage points more likely to pay nothing at all (column 8). One interpretation of this is that the monthly payment schedule in the structured payments arm forces borrowers to engage earlier with the question of whether recovery is realistic. Some decide that it is and commit fully. Others realize that it is not and decide to cut their losses sooner by making no payments, thus saving them money that would likely be lost. This mechanism would also explain the pattern of branch visits that we observe: no overall increase in the structured payments arm (column 10) but 0.17 fewer visits on average from those who default (column 11).
Column 9 shows that mandated structured payments do not affect the fraction of borrowers who make at least one payment.

\begin{landscape}
  
\begin{table}
\caption{Effects on intermediate outcomes}
\label{mechanisms}
\begin{center}

\scriptsize{
\begin{tabular}{lccccccc}
\toprule
      & \multicolumn{7}{c}{Panel A  : Speed of payment} \\
\cmidrule{2-8}      & Days to & \% of payment & $\Pr($Recovery & Loan duration  &       & Loan duration &  \\
      &  1st payment &  in 1st visit &  in 1st visit) & (days) &       &  $|$ recovery &  \\
\cmidrule{1-5}\cmidrule{7-8}      & (1)   & (2)   & (3)   & (4)   &       & (5)   &  \\
\cmidrule{1-5}\cmidrule{7-8}Mandatory structured & -12.2*** & 9.69*** & 0.073** & -22.6*** &       & -14.0*** &  \\
      & (1.86) & (3.13) & (0.029) & (4.94) &       & (4.28) &  \\
Choice  & -3.10* & -0.23 & -0.013 & 3.03  &       & 0.41  &  \\
      & (1.72) & (2.49) & (0.024) & (5.73) &       & (4.79) &  \\
      &       &       &       &       &       &       &  \\
\cmidrule{1-5}\cmidrule{7-8}Observations & 4412  & 6304  & 6304  & 6304  &       & 3031  &  \\
R-squared & 0.026 & 0.008 & 0.007 & 0.021 &       & 0.013 &  \\
Control Mean & 82.8  & 45.8  & 0.30  & 136.6 &       & 103.9 &  \\
      &       &       &       &       &       &       &  \\
\midrule
      & \multicolumn{4}{c}{Panel B  : Variables related to default} &       & \multicolumn{2}{c}{Panel C  : Visit variables} \\
\cmidrule{2-8}      & $\Pr($+ payment & \% of pay & $\Pr($payment $=0$ &       &       &       & \# of visits \\
      &  \& default) &  $|$ def  &  $|$ def) & $\Pr($ payment $=0$) &       & \# of visits &  $|$ def \\
\midrule
\midrule
      & (6)   & (7)   & (8)   & (9)   &       & (10)  & (11) \\
\midrule
\midrule
Mandatory structured & -0.067*** & -3.50*** & 0.12*** & -0.011 &       & 0.016 & -0.17*** \\
      & (0.014) & (1.25) & (0.032) & (0.024) &       & (0.056) & (0.049) \\
Choice  & -0.022 & -0.93 & 0.027 & -0.018 &       & 0.14** & -0.046 \\
      & (0.015) & (1.20) & (0.034) & (0.024) &       & (0.062) & (0.049) \\
      &       &       &       &       &       &       &  \\
\cmidrule{1-5}\cmidrule{7-8}Observations & 6304  & 2492  & 2492  & 6304  &       & 6304  & 2492 \\
R-squared & 0.008 & 0.007 & 0.014 & 0.000 &       & 0.003 & 0.011 \\
Control Mean & 0.12  & 9.68  & 0.71  & 0.31  &       & 1.14  & 0.39 \\
\bottomrule
\bottomrule
\end{tabular}%
}

\end{center}
\footnotesize{Notes: This table explores treatment effects on intermediate outcomes. Each column in the three panels corresponds to a different outcome variable in the regression from \eqref{basic_reg}. The outcomes in Panel A concern the speed of repayment, those in Panel B relate to default, and those in Panel C relate to visits to a branch. 
Days to 1st payment (col 1) is defined only for borrowers who made some payment; probability of recovery in the first visit (col 3) equals one for those who recovered in the first visit, and zero if they did not recover; loan duration (col 4) is the number of days the borrower took to repay her loan for those who recover, the number of days until default for those who default, and the maximum number of days we observe them in the sample for those who neither default nor recover during our observation window.
Standard errors are clustered at the branch-day level.}

\end{table}
\end{landscape}

\normalsize
\normalsize

\paragraph{Other Costs} 
To show that our results are robust to alternative cost measures and to account for additional costs that structured payments contracts impose, Table \ref{table_robustness_fc} presents three adjustments. First, we use borrowers' \emph{subjective} value of the pawn rather than its appraised value, since jewelry has sentimental value. Self-reported valuations are nearly twice the appraised gold value on average, making default even more costly. This adjustment (column 2) \emph{increases} the estimated treatment effect from -183 to -293 pesos (95\% CI: -469 to -117). Second, we add transaction costs from branch visits, including self-reported transport costs and the opportunity cost of borrowers' time (conservatively charged at a full day's minimum wage per visit). Despite requiring three visits instead of one, the savings remain essentially unchanged at -183 pesos (column 3). Third, we account for liquidity costs by applying compound daily interest to all pre-payments, assuming that borrowers would need to borrow at an interest rate of 7\% per month (column 4). Even with this conservative adjustment, the effect remains large and significant at -107 pesos (CI: -183 to -32). Comparing columns (1) and (4) of Table \ref{table_robustness_fc} shows that less than half (42\%) of the total savings come from reduced interest charges. About half of the financial cost savings come from avoiding default. This means the intervention still generates substantial net cost savings even if clients have to borrow to pay earlier. 

This analysis shows that the reduction in undiscounted financial costs to the borrower is robust to the inclusion of several additional cost measures. In Section \ref{why_paternalism} we address the role of time discounting and examine potential implications of our findings, particularly in a context where so few people choose a contract that so many appear to benefit from.

\paragraph{Censoring of Loan Completion}
Our observation window was limited in each branch, preventing us from following some loans to completion. Overall 13\% of experimental loans are censored---neither defaulting nor repaying within our observation window. This censoring particularly affects pawns that were rolled over for additional 90-day periods. In the results presented above, we addressed this issue conservatively by using outcomes such as ``did not default'' which remain well-defined for censored loans. This approach biases our estimates toward zero because control loans, which repay more slowly, are more likely to be censored and thus coded as not having defaulted even though some would have eventually defaulted. In the Online Appendix we consider the issue of censoring in more detail. Table \ref{bounding_censoring} presents the results of a bounding exercise that checks the robustness of our results to alternative assumptions about repayment for censored loans. Even under the most conservative assumptions (Panel B), the causal effect of structure on financial costs remains large and significant at -167 pesos (95\% CI: -276 to -58).

\paragraph{Lender profit} 
If structured payments reduces lenders' profits, they will have little incentive to offer it. For a single loan, the interaction between borrowers and lenders is zero-sum: the borrower's financial cost equals the lender's profit. Table \ref{main_impact_table} showed that structured payments contracts reduced borrowers' financial costs, and hence lenders' profits, on a per loan basis. Table \ref{repeat_loans} from the Online Appendix, however, shows that borrowers in the structured payments arm are 6.7 percentage points more likely to become repeat customers. We now describe a back-of-the-envelope calculation that adjusts for this offsetting effect to determine whether structured contracts increase or decrease total lender profits on net. 
As shown in column 1 of Table \ref{main_impact_table}, structured payments contracts generate approximately \$759 MXN in lender profit per loan compared to \$942 MXN for status quo contracts. As shown in Table \ref{repeat_loans}, 38.7\% of borrowers in the structured payments arm return to take out another loan compared to 32.0\% under the status quo. To determine which contract type is more profitable, we calculate the expected lifetime value of a customer under each:
\[
\text{Profit(status quo)} = 942 \cdot \sum_{t=0}^T\delta^t(0.32)^t, \quad
\text{Profit(structured)} = 759 \cdot \sum_{t=0}^T\delta^t(0.387)^t
\]
where $\delta$ is the discount factor and $T$ is the time horizon. 
This calculation makes three assumptions. First, it assumes that the number of times a borrower returns is a geometric random variable with success probability equal to the overall return rate of their contract type in our experiment.\footnote{Results are virtually unchanged if we replace the geometric model with a binomial model in which customers need to borrow on $n$ future occasions and return with probability $p$ for each.} Second, it assumes that borrowers receive the same contract each time they return. Third, it assumes that financial cost is statistically independent of the decision to return. Under this model, status quo contracts generate higher profits for all $\delta \in [0,1]$ and all $T$ (see Figure \ref{profit_mandatory_control}). The higher rate of repeat business from structured payments contracts does not compensate for their lower per-loan profits. In Online Appendix \ref{append:lender-profit}, we present two alternative versions of this exercise that relax the assumption of independence between per-loan profits and the decision to borrow again. As shown in Figure \ref{bound_lender_profit}, the results are unchanged.

\subsection{Estimating Treatment Effects by Contract Choice} \label{TuT_ToT}

The results discussed so far present a puzzle: structured payments contracts substantially reduce financial costs yet only 11\% of borrowers choose them when given the opportunity. If the effect of structure were homogeneous, this would be enough to conclude that the 89\% who did not choose it would have been financially better off if they had. In a world of heterogeneous treatment effects, however, low demand for structure could still be consistent with rational choice---the borrowers who don't choose it could simply be those who don't need it. This raises a fundamental question: do the ``right people'' choose structure, or would a structured payments contract lower financial costs for borrowers who would not choose it? 

Table \ref{tot_tut} presents the causal effects identified by the Mandates versus Choice design---TOT, TUT, ASG, ASB, and ASL---along with standard errors clustered at the branch-day level.  For ease of interpretation, we redefine outcomes in this section and those that follow so that positive values indicate beneficial effects: cost savings, lower default rates, and CR reductions.
For reference, row 1 shows ATE results and rows 4--5 present the average potential outcomes $\mathbbm{E}[Y_0]$ and $\mathbbm{E}[Y_1]$. 

Row 3 of Table \ref{tot_tut} contains our central finding: non-choosers experience substantial cost savings from structure. The TUT effect is positive, statistically significant, and economically large across all outcomes.
Non-choosers save an average of 192 pesos (95\% CI: 92 to 291) and experience a 6.8 percentage point reduction in CR (95\% CI: 4.0 to 9.6)---cost savings they forgo by choosing the status quo contract. Indeed, our estimated TUT effects are larger than the corresponding ATE estimates for all outcomes except default. Put simply, the 89\% of choice arm borrowers who reject structure are leaving money on the table. 

\begin{table}[htbp]
\caption{TOT, TUT, ASG, ASB, and ASL estimates}
\label{tot_tut}
\begin{center}
\footnotesize{
\begin{tabular}{lcccc}
\toprule
      & CR \% benefit & FC benefit & \% (1-Default) & \% (1-Refinance) \\
\midrule
      & (1)   & (2)   & (3)   & (4) \\
\midrule
\midrule
ATE   & 6.16*** & 183.0*** & 7.74*** & 6.34** \\
      & (1.25) & (50.8) & (2.50) & (2.90) \\
ToT   & 1.00  & 111.9 & 37.4* & -25.9 \\
      & (12.6) & (528.3) & (21.6) & (29.1) \\
TuT   & 6.77*** & 191.5*** & 4.20* & 10.2*** \\
      & (1.44) & (50.8) & (2.41) & (2.90) \\
$\mathbb{E}[Y_1]$ & -36.4*** & -759.4*** & 64.2*** & 67.2*** \\
      & (0.85) & (27.3) & (1.69) & (1.70) \\
$\mathbb{E}[Y_0]$ & -42.5*** & -942.4*** & 56.4*** & 60.9*** \\
      & (0.91) & (42.9) & (1.84) & (2.35) \\
\midrule
ASG   & -5.77 & -79.6 & 33.2  & -36.1 \\
      & (13.4) & (556.2) & (22.6) & (30.6) \\
ASB   & 8.89  & 291.5 & -39.1* & 22.7 \\
      & (13.1) & (551.2) & (22.3) & (30.1) \\
ASL   & 3.12  & 211.9*** & -5.90 & -13.4*** \\
      & (2.15) & (59.5) & (4.29) & (4.20) \\
\midrule
Observations & 6304  & 6304  & 6304  & 6304 \\
$H_0$ : ASG=0 & 0.67  & 0.89  & 0.14  & 0.24 \\
$H_0$ : ASG$\leq$ 0 & 0.67  & 0.56  & 0.071 & 0.88 \\
\bottomrule
\bottomrule
\end{tabular}%
}
\end{center}
\scriptsize{Notes: This table presents results computed using the derivations from Section \ref{sec:randchoice}. The CR and financial cost outcomes have been multiplied by $-1$ so that causal effects represent benefits to the borrower in each of the four columns. FC benefits are measured in pesos, while default and refinance outcomes are measured in percentage points. The bottom panel presents p-values for two hypothesis tests of treatment effect homogeneity.} 
\end{table}

Standard models of rational choice predict that choosers should benefit more than non-choosers (TOT$>$TUT)---positive selection on gains. While the low take-up rate of structure in the choice arm (11\%) limits our power, we find no evidence in support of this prediction. Point estimates for the average selection on gains ($\text{ASG} = \text{TOT} - \text{TUT}$) are negative for all outcomes except default probability, but none are statistically significant. For example, the estimated ASG for our CR outcome is -6 percentage points with a 95\% confidence interval of -32 to 20 percentage points. The imprecision in our ASG estimates stems from the small number of choosers, leading to large standard errors for TOT effects. The most precisely estimated of our TOT effects (column 3), still has a wide 95\% confidence interval: -5 to 80 percentage points.

We do find some evidence of other forms of selection. First, the average selection bias (ASB) for the recovery outcome (column 3) is -39 percentage points, with a 90\% confidence interval of -76 to -2 percentage points. This means that choosers are estimated to have a higher risk of default under the status quo than non-choosers. Thus, we find some suggestive evidence of selection on baseline default risk $Y_0$ but no evidence of selection on gains $Y_1 - Y_0$. This pattern would be consistent with choosers recognizing their elevated default risk under the familiar status quo contract but being unable to predict how much the unfamiliar structured payments contract would help. Second, we find evidence of average selection on levels (ASL) for the FC outcome (column 2) and the ``did not refinance'' outcome (column 4). 

This experiment identifies treatment effects for choosers (TOT) and non-choosers (TUT) simultaneously, showing that non-choosers may experience cost savings as large as or larger than those who self-select into treatment.
The Mandates versus Choice design could be valuable in other settings with low rates of voluntary adoption. Non-adherence to treatment is widespread in medical applications \citep{non_adherence} yet existing work rarely asks whether voluntary adherers benefit more than non-adherers. Similarly, while several studies have shown low take-up of commitment devices \citep[e.g.,][]{Ashraf, Gine, Ted, Royer}, existing work rarely estimates treatment effects for non-takers.\footnote{We note that some other papers find higher rates of commitment take-up \citep{Kremer, Casaburi, Alcohol, AprajitP&P, Pascaline}, see \cite{CarreraEtAl2022} for more detail.}

\section{Heterogeneous Treatment Effects} 
\label{sec:harm}
In Section \ref{Experiment} we estimated large and statistically significant average effects of structured payments, both on average and for borrowers who chose the status quo contract. Average effects, however, could mask substantial heterogeneity. Despite the large and positive TUT effect, some non-choosers could still experience negative treatment effects, validating their decision to opt for the status quo. We now present two exercises that explore whether anyone experiences \emph{higher financial costs} under structured payments. The first computes partial identification bounds for the distribution of \emph{individual} treatment effects. Even under the most conservative assumptions, at least 28\% of borrowers experience lower financial costs under structure---more than twice the 11\% who choose it. The second uses causal forests and our borrower survey to estimate conditional average treatment effects. We find that only 1\% of estimated conditional TUT effects are negative, suggesting that cost increases are rare among non-choosers. 

\subsection{Bounding the Distribution of Individual Treatment Effects}
\label{sec:bounds}
While the distribution of \emph{individual} treatment effects cannot be point identified, it can be bounded, allowing us to go beyond the average effects from Section \ref{sec:baseline} to compute a lower bound for the \emph{fraction} of borrowers who would experience cost reductions under mandated structure.
This exercise relies only on Assumptions \ref{assump:design} and \ref{assump:SUTVA}; unlike our TUT and TOT results, it does \emph{not} rely on an exclusion restriction.

As above, let $(Y_{i0}, Y_{i1})$ be $i$'s potential outcomes under the status quo and mandated structured payments conditions and define $\Delta_i \equiv Y_{i1} - Y_{i0}$ as borrower $i$'s treatment effect.
The marginal distributions $F_0, F_1$ of $Y_{i0},Y_{i1}$ are observed, but the marginal distribution $F_\Delta$ of $\Delta_i$ is not.
To bound the latter, we consider all joint distributions of the potential outcomes that agree with $F_0$ and $F_1$.
As shown in \cite{fan2010sharp}, for any fixed $\delta$ the sharp bounds for $F_\Delta(\delta)$ are $[\underline{F}(\delta), \,\overline{F}(\delta)]$ where 
\[
\underline{F}(\delta) \equiv \max \left\{0, \sup_y F_1(y) - F_0(y - \delta)  \right\}, \, \,
\overline{F}(\delta) \equiv 1 + \min \left\{0, \inf_y F_1(y) - F_0(y-\delta) \right\}.
\]
Since $F_0$ and $F_1$ are identified under Assumptions \ref{assump:design}--\ref{assump:SUTVA}, so are $\underline{F}$ and $\overline{F}$.
Setting $\delta = 0$ bounds the fraction of borrowers whose costs \emph{increase} under mandated structure: $F_\Delta(0) = \mathbbm{P}(\Delta_i \leq 0)$.

Using this approach, we estimate a lower bound of $0.04$ and an upper bound of $0.72$ for the share of borrowers whose costs increase under mandated structure, based on the CR outcome.
Corresponding 95\% confidence intervals are $[0.02, \, 0.06]$ for the lower bound and $[0.69, \, 0.75]$ for the upper bound.\footnote{As described in Appendix
\ref{append:bounds}, we compute these bounds adjusting for day of the week and branch dummies to tighten the bounds. See Appendix for details of our inference procedure as well.}
Since at most 72\% experience negative effects, even under the most conservative assumptions, and without imposing an exclusion restriction, we find that \emph{at least} 28\% of borrowers have positive treatment effects from mandatory structured repayment, more than double the fraction who demand this contract.

With stronger assumptions, it is possible to say more.
One such assumption is \emph{rank invariance}, which posits that a person's rank in the distribution of $Y_0$ equals her rank in the distribution of $Y_1$. While strong, this assumption or variants of it have appeared in a number of settings in the literature, e.g.\ \cite{chernozhukov2005iv}.
Under rank invariance, the distribution of treatment effects is point identified and given by $F_\Delta(\delta) = \int_0^1 \mathbbm{1}\{ F_1^{-1}(u) - F_0^{-1}(u)\leq \delta\}\,\mathrm{d}u$ where $F_1^{-1}$ and $F_0^{-1}$ are the quantile functions of $Y_{i1}$ and $Y_{i0}$. 
For the CR outcome, rank invariance implies that around 90\% of borrowers have positive individual treatment effects from mandatory structured repayment.
See Figures \ref{fan_park_bounds} and \ref{te_rankinvariance} in Appendix \ref{append:bounds} for our estimate of the full distribution $F_\Delta(\delta)$ with and without rank invariance.

\subsection{Conditional Average Treatment Effects via Causal Forests} 
\label{sec:RF}
While informative, our bounds for the share of borrowers who experience higher financial costs under structure are wide. 
Rank invariance collapses these bounds to a point, but at the expense of imposing a strong additional assumption.
In this section we take a different approach, using data from our baseline survey to study \emph{conditional average} treatment effects rather than individual effects. 
Specifically, we ask whether it is possible to identify \emph{groups of borrowers} with negative average effects from mandated structure.

All treatment effect results presented thus far have relied exclusively on administrative data; we have only used the baseline survey to assess experimental balance (Online Appendix Table \ref{SS_final}) and to explore predictors of take-up in the choice arm (Online Appendix Table \ref{determinants_choose}).
In the remainder of this section and the next, we incorporate survey responses into our estimates of treatment effects.\footnote{See Table \ref{baseline_survey} in the Online Appendix for the list of the survey variables used in this section.}
The survey response rate was 78\%; as such, the results discussed in this section should be understood as applying to this subsample of respondents. 
Reassuringly, unconditional treatment effects for survey respondents are nearly identical to those for the full sample, as is the take-up rate of structured repayment.
See Online Appendix \ref{append:respondents} and Table \ref{tab:main_effects_respondents} for details.
Because the baseline survey was administered before treatment assignment, random assignment of $Z_i$ remains valid within the survey subsample.\footnote{Assumption \ref{assump:design} continues to hold conditional on pre-treatment covariates. 
Similarly, since Assumptions \ref{assump:SUTVA}--\ref{assump:exclusion} are restrictions on the potential outcomes of \emph{each individual} in the population, they also hold conditional on pre-treatment covariates. 
Thus, the results of Section \ref{sec:randchoice} continue to apply conditional on survey response and information from the baseline survey.}
Before proceeding to estimate heterogeneous treatment effects, we note that the omnibus test of homogeneous treatment effects from \cite{chernozhukov2018generic} strongly rejects the null and that the heterogeneous treatment effect estimates presented below are well-calibrated.\footnote{Details available upon request.}

We consider conditional ATE, conditional TUT, and conditional TOT effects---CATE, CTUT, and CTOT for short---defined as $\text{CATE}(x) \equiv \mathbbm{E}[Y_{i1} - Y_{i0}|X_i = x]$ and
\[
\text{CTUT}(x) \equiv \mathbbm{E}[Y_{i1} - Y_{i0}|C_i = 0, X_i = x], \quad
\text{CTOT}(x) \equiv \mathbbm{E}[Y_{i1} - Y_{i0}|C_i = 1, X_i = x]
\]
respectively.
We estimate these as nonparametric functions of survey responses, using the ``generalized random forest'' approach of \cite{atheygrf}.
See Appendix \ref{app:cate} for implementation details.
Note that our CATE results, like the bounds from Section \ref{sec:bounds}, rely only on Assumptions \ref{assump:design}--\ref{assump:SUTVA} while our CTUT and CTOT results, like their unconditional counterparts, impose the exclusion restriction from Assumption \ref{assump:exclusion}.

\begin{figure}[htbp]
     \caption{Heterogeneous treatment effects} 
     \label{heterogeneous_effects}    
    \begin{center}
     \begin{subfigure}{0.3\textwidth}
       \centering
      \includegraphics[width=\textwidth]{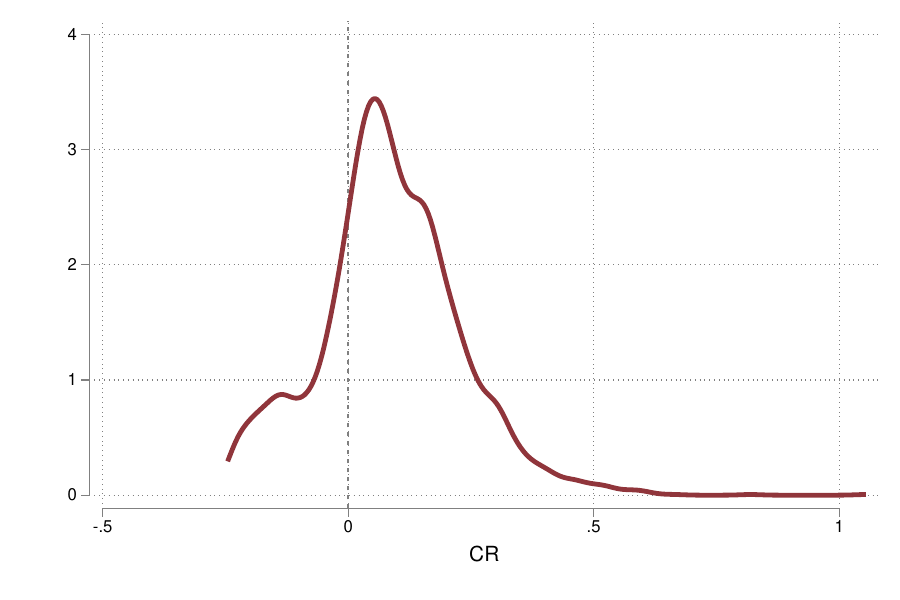}
          \caption{Conditional TOT}
    \end{subfigure}
    \begin{subfigure}{0.3\textwidth}
       \centering
      \includegraphics[width=\textwidth]{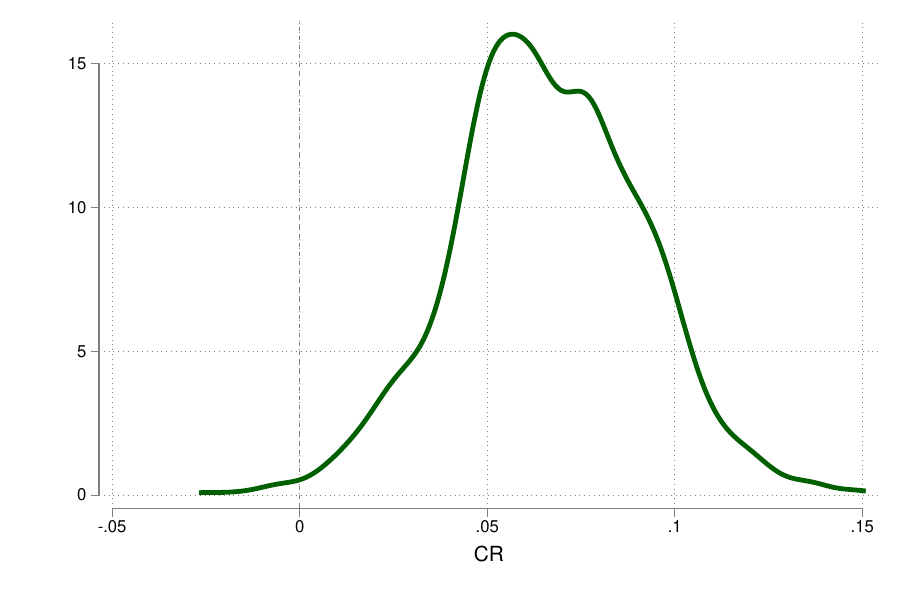}
          \caption{Conditional TUT}
    \end{subfigure} 
       \begin{subfigure}{.3\textwidth}
        \centering
        \includegraphics[width=\textwidth]{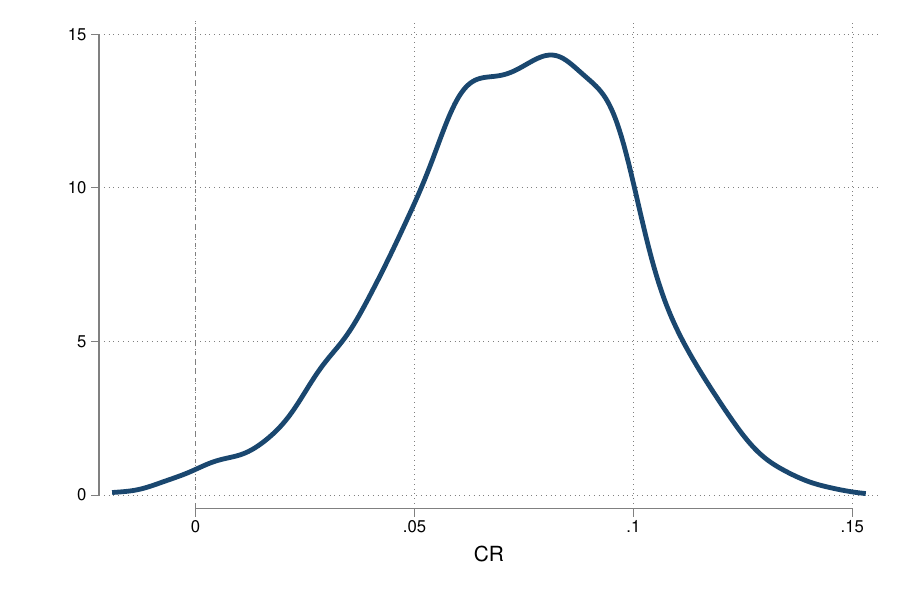}
        \caption{Conditional ATE}
    \end{subfigure} 
    \end{center}
  \scriptsize{Notes: Kernel density plots of estimated conditional average treatment effects across borrowers in our sample, constructed using survey and administrative data using the generalized random forests. See Appendix \ref{app:cate} for implementation details.} 
\end{figure}

Figure \ref{heterogeneous_effects} presents kernel density plots that summarize the estimated CATE, CTOT, and CTUT effects across the individuals in our experiment. For example, panel (c) is produced by defining $\hat{\theta}_i \equiv \widehat{\text{CATE}}(X_i)$ to be the causal forest estimate of participant $i$'s CATE, based on her survey responses, and plotting the density of $\widehat{\theta}_i$ across $i$.\footnote{These densities are merely convenient summaries of the \emph{estimated} conditional average effects across borrowers in our sample; they are not estimators of the underlying distribution of \emph{true average effects} $\theta_i$. Estimating the latter would involve solving a challenging deconvolution problem. Fortunately this is not required to address our question of interest in this section.}
In each panel, the outcome variable is CR benefit, i.e. the \emph{reduction} in CR from a structured repayment contract.

As seen from the figure, estimated conditional average effects are highly heterogeneous but overwhelmingly positive. 
The density of TUT estimates is particularly interesting because it summarizes conditional average effects for borrowers who would not voluntarily choose structure.
Strikingly, only 1\% of our estimated conditional TUTs are negative.
Using a bootstrap approach detailed in Online Appendix \ref{app:cate}, we obtain a 95\% confidence interval of $[0.1\%, \, 1.8\%]$ for the share of non-choosers in the population who have a negative average treatment effect from mandated structure.
To be clear: this is a probability statement about conditional average effects over the distribution of covariates. In particular, we estimate that  $\int \mathbbm{1}\{\mathbb{E}[Y_1 - Y_0| X = \mathbf{x}, C = 0] < 0\} \, f(\mathbf{x}|C=0) \, \mathrm{d}\mathbf{x}$ is approximately $0.008$ with the aforementioned confidence interval. 
In other words, after flexibly accounting for observed heterogeneity using survey responses, it is difficult to find groups of borrowers whose financial costs would increase on average under a policy of mandated structure, even among non-choosers. This need not imply that 99\% of \emph{individual} non-choosers would see lower financial costs under mandated structure. There could remain important sources of unobserved heterogeneity such that the share with negative individual effects, $\mathbbm{P}(Y_1 < Y_0| C = 0)$, is larger (see Section \ref{sec:bounds} for bounds on the distribution of individual effects). Nevertheless, our CTUT results strengthen the case that mandated structure is broadly helpful in reducing financial costs.

The more heterogeneity that $X_i$ explains, the closer conditional average effects become to individual effects. Thus, under the stronger assumption that our survey measures capture the main sources of treatment effect heterogeneity, our causal forest estimates of conditional average effects \emph{approximate} individual-level effects. The remainder of this section proceeds under this stronger assumption, allowing us to consider whether particular borrowers made ``mistakes'' in their choice of contract.\footnote{Here we define a ``mistake'' in terms of forgone \emph{financial cost savings}. Such savings may not imply forgone \emph{welfare improvements}. We discuss this point further below in Section \ref{sec:limitations}.}
We define a mistake for a non-chooser as a conditional TUT estimate that exceeds a specified CR threshold---for example, 5 percentage points---meaning that she would have experienced substantially lower financial costs if she had chosen structure.
Analogously, a mistake for a chooser is a conditional TOT estimate below the \emph{negative} of that same threshold, meaning she \emph{increased} her financial costs by choosing structure.

\vspace{.2in}
\begin{figure}[htbp]
 \caption{Fraction of sample foregoing financial savings}
 \label{choose_wrong}
 \begin{center}
     \begin{subfigure}{0.49\textwidth}
       \centering
      \includegraphics[width=\textwidth]{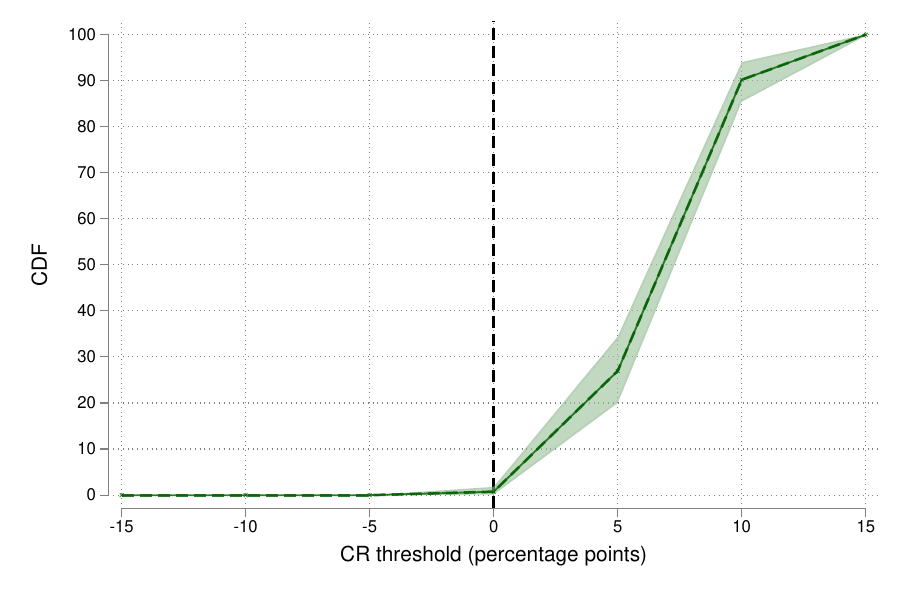}
          \caption{CDF of CTUT effects for non-choosers}
    \end{subfigure}
    \begin{subfigure}{0.49\textwidth}
       \centering
      \includegraphics[width=\textwidth]{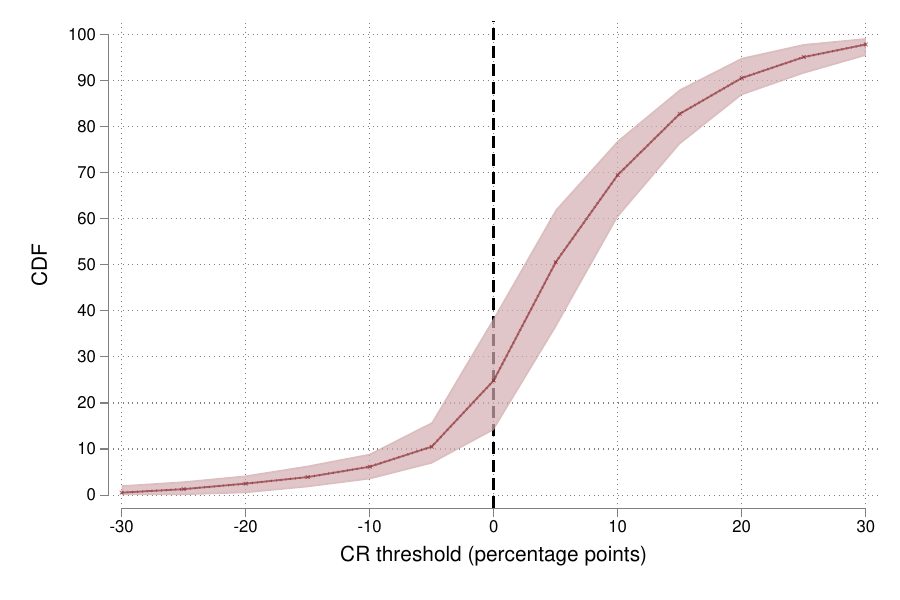}
          \caption{CDF of CTOT effects for choosers}
    \end{subfigure} 
    \end{center}
\scriptsize{Notes: Each panel shows the CDF of estimated conditional average treatment effects on financial costs (CR, in \% points of the loan value), with pointwise 95\% bootstrap confidence bands. The caption's fraction of the sample foregoing financial savings---by declining structure in panel (a), by choosing it in panel (b)---is read from these curves as follows. The fraction of non-choosers estimated to forgo financial savings of at least the threshold $t$ is 100 minus the CDF in panel (a) evaluated at $t$; the fraction of choosers estimated to have increased their financial costs by at least $t$ is the CDF in panel (b) evaluated at $-t$.}
\end{figure}

Figure \ref{choose_wrong} plots the CDFs corresponding to the densities from panels (a) and (b) of Figure \ref{heterogeneous_effects} from above, along with pointwise 95\% confidence bands constructed using the bootstrap procedure described in Online Appendix \ref{app:cate}.
Our results suggest that many non-choosers would have experienced substantially lower financial costs had they chosen structured payments. We estimate that between 98\% and 99.9\% of non-choosers have a positive CTUT effect, and between 66\% and 80\% have a CTUT that \emph{exceeds} 5 percent of the loan amount. In contrast, comparatively few choosers appear to have made mistakes by choosing structure: roughly 14-38\% of choosers have a \emph{negative} CTOT---they shouldn't have chosen structure---and only 7-16\% have a CTOT effect below -5 percentage points (a large mistake). These findings strengthen the case for mandatory structure: not only does it generate large cost savings on average, but most non-choosers belong to groups with substantial positive conditional average effects. As we discuss further in Section \ref{sec:limitations} below, financial costs savings need not imply welfare gains. Nevertheless, the cost savings we document here are large relative to measures of other costs that structure may impose on borrowers, e.g.\ transactions costs and liquidity costs. See \textbf{Other Costs} in Section \ref{sec:baseline} for details.

Given the large variation in conditional average effects, one might consider targeting mandatory structure to borrowers who were predicted to benefit, rather than applying it universally. To be feasible in the real-world, however, such a targeted approach could only rely on covariates that are either verifiable or difficult to manipulate. In our setting, however, it turns out that the narrow set of feasible covariates---age, gender, education, desired loan size, and prior pawn experience---has limited power to inform targeting. An optimal targeting rule based on these narrow covariates reduces the share of borrowers assigned to the ``wrong'' treatment only marginally.\footnote{Details available upon request.} Hence in this context, universal mandated structure may be more attractive than a feasible targeted approach.

\section{Why Does Structure Reduce Financial Costs?}
\label{why_paternalism}

Sections \ref{TuT_ToT} and \ref{sec:harm} showed that non-choosers experience substantial cost savings from structured payments. At least 28\% of borrowers have positive individual effects and 99\% of non-choosers show positive conditional average effects. Yet only 11\% choose structured payments voluntarily. 
As described in \textbf{Other Costs} from Section \ref{sec:baseline} above, the ATE of 183 pesos and TUT of 192 pesos are large relative to plausible estimates of additional liquidity and transactions costs that borrowers may incur under structured repayment.
These results, along with the novelty and unfamiliarity of the structured contract, suggest that many borrowers may have made a costly mistake by opting for the status quo. 
While we cannot definitively identify the mechanism behind this mandate-choice gap, our baseline survey allows us to explore two standard behavioral explanations: present bias and overconfidence. We find that neither measure predicts take-up, but subjective certainty about pawn recovery predicts the strength of the TUT effect. We then provide a simple model illustrating how the state-contingent features of the structured payments contract could interact with overconfidence to generate heterogeneous TUT effects along with low take-up. We close by discussing our findings in the light of the existing literature.

\paragraph{Learning and Time Discounting} Before proceeding, we consider two neoclassical explanations. The first is \emph{learning}. The structured payments contract was unfamiliar to borrowers; perhaps they required more experience to understand its benefits. As detailed in Online Appendix \ref{append:learning}, we find no evidence that borrowers assigned to the mandatory structured payments arm were more likely to choose structured payments in the future. The second is \emph{time discounting}. The monthly payments required under the structured payments contract are front-loaded, while pawn recovery is back-loaded. Perhaps high discount rates among borrowers explain their low demand for structure. In Online Appendix Figure \ref{fc_discount_rates}, we show that implausibly high annual discount rates---above 1000\%---would be required to offset TUT effects of the magnitude we estimate above. While these exercises cannot completely exclude a role for learning or discounting, they motivate us to consider behavioral explanations below.

\subsection{Present Bias}
\label{sec:PB}

Present-biased borrowers overweight immediate costs relative to future ones. In the pawn setting, this may lead to delayed repayment and, ultimately, default. 
Sophisticated present-biased borrowers realize that they stand to benefit from structure and choose it when given the option. Their na\"ive present-biased counterparts also benefit from structure but fail to realize this. When offered the choice, these borrowers opt for the status quo. If most borrowers are na\"ive present-biased, this would explain our key findings: few choose structured payments (the sophisticated) and structure benefits even the non-choosers (the na\"ive). 

\begin{figure}[htbp]
\caption{Differences in TUT effects by behavioral variables (financial cost outcome)}
    \begin{center}
        \centering
\includegraphics[width=0.65\textwidth]{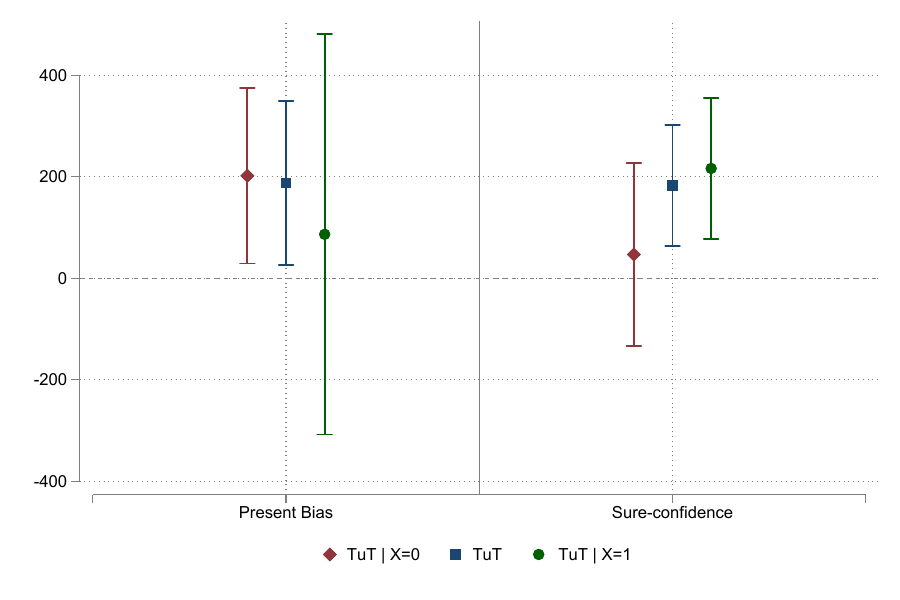} 
    \end{center}  
\footnotesize{Notes: Each panel in this figure shows how the estimated treatment on the untreated (TUT) effect on financial cost (in pesos) varies with a binary survey variable $X_i$. Error bars show 95\% confidence intervals with standard errors clustered at the branch-day level. In the left panel, $X_i = 1$ if borrower $i$ is present-biased based on her responses to the time preference questions from our survey. In the right panel $X_i = 1$ if borrower $i$ reported that she was certain to recover her pawn. The larger confidence intervals in the left panel reflect the smaller sample of borrowers who answered the PB questions. Results are similar when restricted to the subset of borrowers who answered both PB and sure-confidence questions.}
    \label{tut_beh_partition}
\end{figure}

Under this explanation, non-choosers are either time-consistent or na\"ive present-biased. Time-consistent borrowers do not benefit from structure but na\"ive present-biased borrowers do. Thus, if na\"ive hyperbolic discounting drives our results, the TUT effect for the present-biased should be larger than for other borrowers. 
To evaluate this prediction, we measure present bias using standard hypothetical choice questions from our baseline survey: comparing a decision between payment tomorrow versus payment one month in the future against a decision between payment three months versus four months in the future. Based on this measure, we classify borrowers who exhibit more impatience over immediate delays as present biased.\footnote{While this measure of present-bias is widely used, it has known limitations compared to more elaborate incentivized measures. See Section \ref{sec:limitations} for more discussion of this point.} 
Taking our survey measure at face value, we find no evidence that present-bias explains our positive estimated TUT results from above.
As shown in the left panel of Figure \ref{tut_beh_partition}, the TUT estimate for financial cost among present-biased borrowers (at right) is 87 pesos (95\% CI: -307 to 481), less than half as large as the estimate of 202 pesos (95\% CI: 29 to 374) for time-consistent borrowers (at left). This pattern is the opposite of what the na\"ive present bias explanation would predict, although the difference between TUT effects is not statistically significant.

In standard models, only sophisticated agents demand commitment. Because we cannot distinguish sophisticated from na\"ive borrowers in our survey, however, we cannot test this theoretical prediction. Empirically, our measure of present bias does not predict take-up of structured payments in our experiment (Online Appendix \ref{determinants_choose}).

\subsection{Overconfidence}
\label{sec:overconf}

Overconfidence offers another potential explanation of the low demand for structured payments despite financial cost savings from mandating it. Whereas present bias concerns preferences, overconfidence concerns \emph{beliefs}.\footnote{While conceptually distinct, present bias and overconfidence may be related in practice: see Section \ref{sec:limitations} for more discussion of this point.}
Overconfident borrowers are overly-optimistic about their probability of recovering their pawn under the status quo. They may be counting on a financial windfall that is unlikely to materialize, or they may expect that setbacks experienced in the past will not persist into the future. A highly confident borrower, one who is effectively \emph{certain} that she will recover her pawn, sees no value in structured payments and opts for the status quo. Why sacrifice repayment flexibility to avoid an outcome that she sees as impossible? But if this confidence is misplaced, she may still benefit from mandated structure.

Our baseline survey also elicited borrowers' subjective probability of recovery before they knew their treatment assignment, asking them to mark a point along an interval ``where 0 is impossible and 100 is completely certain''. In total, 72\% of survey respondents were \emph{completely certain} they would recover. We call these borrowers \emph{sure confident}.\footnote{Figure \ref{determinants_sure} in the Online Appendix shows the correlates of sure confidence in our sample.}
Not all sure confident borrowers are necessarily overconfident, but the overall default rate of 42\% for the sure confidents under the status quo suggests that many are overly optimistic about their prospect of recovery. Empirically, we find that the positive TUT effect is concentrated among the sure confident. As shown in the right panel of Figure \ref{tut_beh_partition}, the TUT estimate for sure confident borrowers (at right) is 216 pesos (95\% CI: 78 to 354 pesos). The estimate for other borrowers (at left) is 46 pesos (95\% CI: -133 to 227), although the difference between TUT effects is not quite significant ($p = 0.12$).
Figure \ref{counterfactual_TUT} in the Online Appendix shows that these differences in TUT effects are driven by differences in average untreated potential outcomes: $\mathbbm{E}[Y(0)|C=0]$.
Under structure, sure confident borrowers experience the same borrowing costs on average as other borrowers; under the status quo, however, their borrowing costs are much higher. This pattern is consistent with the idea that structured payments correct mistakes among overconfident borrowers.

For take-up, overconfidence delivers a clearer theoretical prediction than present bias: borrowers certain of recovery should show less demand for structure because it is a solution to a problem they do not believe they have. 
As seen from Online Appendix \ref{determinants_choose}, however, neither sure confidence nor present bias predicts take-up in the choice arm. This is consistent with recent results from \cite{CarreraEtAl2022} and suggests that self-targeting of structured payments may not be effective in our context.

\subsection{Modeling the Role of Overconfidence}
\label{sec:model}
These results raise a natural question: why would sure confident borrowers react to mandated structure in the first place? If they are certain that structure will not help them to avoid default, why don't they simply ignore the monthly payment deadlines? The answer lies in two unique features of pawn lending. First, late fees are state-contingent, because they can only be collected from borrowers who ultimately recover their pawn. Second, prepayments are risky because they are forfeited in default. 

Evidence for the importance of the fee comes from our Promise treatment arms, which provide the same monthly payment schedule as the structured payment contracts but without late fees. TUT effects are negative for all groups with wide confidence intervals that include zero; see Online Appendix figure \ref{tut_beh_partition_promise} for details. This suggests that the fee and its enforcement by the lender are critical. Structured payments generate sharp bunching at monthly deadlines (Figure \ref{porc_payment_over_time} (a)) and substantial cost savings (Table \ref{main_impact_table}). Promise arms with identical schedules but no enforcement show no bunching (Figure \ref{porc_payment_over_time} (b)) and null effects (Table  \ref{sec_impact_table}). The fee makes the payment structure both salient and credible. The payment schedule alone does not change behavior.

Motivated by these findings, Appendix \ref{append:model} presents a simple stylized model explaining how the state-contingent features of the structured payments contract interact with borrower beliefs to make sure confident borrowers \emph{more responsive} to mandated structure. There are three periods. In the first, borrowers are assigned the status quo or structured payments contract.\footnote{Because its goal is to explain a difference in TUT effects arising from overconfidence, our model is designed to apply only to borrowers who would not choose structure voluntarily, 89\% of our sample.} 
In the second, they decide whether or not to make a fixed prepayment toward recovery of their pawn. In the third, the default or recovery outcome is realized, with a probability that depends on prepayment. Borrowers are free to make a prepayment under either contract type, but structure changes the costs and benefits of prepayment by introducing a fee for missed prepayments.
\footnote{As discussed in Section \ref{sec:limitations}, the ``fee'' in our model can be interpreted more broadly as a reduced form for behavioral channels that are absent from the Promise arm but may be present in the Structured payments arm.}
To accommodate overconfidence, the model allows borrowers' subjective probability of recovery given prepayment to differ from the objective probability.

Sure confident borrowers are certain they will recover their pawn. For this reason they are also certain they will pay a fee if they miss prepayments under the structured payments contract. Borrowers who are less certain of recovery perceive weaker incentives from the same fee: they believe there is a chance that they will default and never owe the fee. Sure confident borrowers are also certain that they will not forfeit any prepayments made, whereas other borrowers view prepayment as risking money that they may lose in default. Under plausible conditions on liquidity costs and beliefs, described in Appendix \ref{append:sure-confident}, these two channels imply that sure confident borrowers are the \emph{most likely} to react to mandated structure by prepaying. In Figure \ref{SC_prepayment}, we show that this is indeed the case in our experiment, by comparing the causal effect of structure on prepayment for sure confident borrowers against the same effect for other borrowers (a difference-in-differences). Consistent with our model, the difference in causal effects of structure on cumulative prepayment (sure confident minus not) jumps around day 30 and remains elevated until day 90. Similarly, the difference in effects of structure on prepayment within a seven-day window jumps around days 30 and 60, the interim payment days stipulated by the loan contract.
In short: the sure confident react more strongly to structure, as predicted by the model.

In our model, prepayment improves the chance of recovery. As shown in the Appendix, unless sure confident borrowers have a \emph{much smaller} effect of prepayment on recovery, their greater responsiveness to structure implies that they will have the largest TUT effect. This is consistent with our findings from \ref{tut_beh_partition} above.

\subsection{Discussion of Behavioral Mechanisms}
\label{sec:limitations}

The results discussed above suggest that overconfidence about pawn recovery may provide a better explanation for our main empirical findings than present bias. However, this result should be interpreted cautiously for a number of reasons. First, our measure of present bias has several limitations that have been highlighted in the recent behavioral literature: it is dichotomous, unincentivized, and based on hypothetical monetary choices rather than consumption or real effort choices \citep{andreoni2012estimating,andreoni2015measuring,augenblick2015working,cohen2020measuring}. 
Given field constraints---surveying real pawnshop clients before commercial transactions---our simple measure was a necessary compromise, but the null result for present bias could reflect measurement error in our survey instrument.

Second, present bias and overconfidence are conceptually related. \cite{heidhues2010exploiting} and \cite{allcott2022high} show how na\"ive present bias can mechanically create overconfidence about future behavior.
Because our measure of present bias does not separately identify na\"ive from sophisticated types, our overconfidence measure may to some degree capture na\"ive present bias rather than pure overoptimism. Empirically, however, the correlation between our two measures is only -0.01, suggesting that they function as independent constructs in our sample, although this too could reflect error in our measurement of present bias. Following \cite{Manski2004MeasuringExpectations}, a large literature aims to elicit subjective beliefs, which are often considered easier to measure than time preference. Whether overconfidence is the correct primitive or a more tractable proxy, our results suggest that understanding beliefs about repayment has an important role to play in the design of financial services.

As emphasized throughout the paper, lower financial costs need not imply welfare gains. Our model provides one way to consider the welfare implications of mandated structure. Borrowers who prepay regardless of contract type---always-payers---experience no change in welfare from mandated structure. 
Borrowers who never prepay regardless of contract type---never-payers---are harmed in expectation by mandated structure: their recovery prospects are unchanged, but those who recover pay additional fees with no compensating benefit. For ``switchers''--- borrowers who only prepay under structure---welfare effects could be positive. On the one hand, switchers with correct beliefs are harmed: structure shifts them away from the optimal choice. 
On the other hand, switchers who overestimate their probability of recovery under the status quo, or who underestimate the efficacy of prepayment in avoiding default, can experience welfare gains. 
In our model, the aggregate welfare effect of mandated structure is more positive (or less negative) the greater the share of sure-confident borrowers, under conditions given in Online Appendix \ref{append:welfare}. This result does not, however, sign the aggregate welfare effect of mandated structure. Our results, documenting widespread overconfidence and financial cost savings from structure, are suggestive (not conclusive) of positive overall welfare effects. If the disutility from potentially violating the structure caused sufficient mental stress to outweigh the financial savings, however, welfare could fall. Similarly, if borrowers value the liquidity forgone by paying early far in excess of prevailing interest rates (the rate at which our model prices it), or if scheduled payments induce larger hassle costs than our estimates suggest, borrowers could be left worse off. Richer data than our survey provides would be needed to reach a definitive conclusion.   

An important question that our model leaves unanswered is why a small fee causes such a large change in behavior.
Our model posits overconfidence as an explanation for differences in borrowers' responses to structure, but the fact that they respond at all comes from the fee $f$.
We interpret $f$ more broadly as a reduced form for state-contingent behavioral channels that are absent from the Promise arm but may be present in the Structured payments arm, making it more effective than its relatively small monetary size would suggest.
Given the limitations of the data at our disposal, we can only speculate as to what these channels may be: borrowers may treat the familiar status-quo contract as a reference point, where kinked preferences could give the small fee disproportionately large behavioral effects.
\cite{Homonoff2018PlasticBags} proposes a similar interpretation for her finding that a \$0.05 fee reduced disposable shopping bag use by over 40 percentage points while an economically equivalent bonus had almost no effect.
Alternatively, the fee may capture the shame or regret of being formally penalized by an authority figure for violating the terms of a legal contract.
These channels are all contingent on repayment.
We cannot rule out channels that are not: the structured contract may work not only as a price schedule, but also as an instruction from the lender to pay monthly \citep{akerlof1991procrastination}, may cause borrowers to notice, remember, and track the payment schedule more closely, or may help them break a large future obligation into smaller intermediate targets, improving follow-through \citep{nickerson2010voting,milkman2011using,milkman2013planning,beshears2016beyond}.
Together, the essentially zero TUT for non-sure-confident borrowers (though imprecisely estimated) and the null result from our Promise arm suggest that the operative mechanism is state-contingent and reliant on external imposition by the lender.

\section{Conclusion} \label{conclusion}

In this paper we make three main contributions. First, we analyze pawn lending, an important but understudied industry serving hundreds of millions worldwide. We provide new stylized facts and show how contract structure drives repayment and borrower costs. In contrast with recent research documenting benefits of repayment flexibility in microfinance, we show that such flexibility can be costly in overcollateralized loan markets. Second, we introduce a ``Mandates versus Choice'' design that identifies treatment effects separately for choosers and non-choosers without requiring structural modeling assumptions. This allows us to test whether self-selection is efficient: whether those with the largest financial cost reductions choose treatment. This is a fundamental question for policy design and for evaluating paternalistic interventions. The design could prove valuable in other settings with low voluntary adoption, such as medical treatment adherence or commitment savings devices, where researchers need to distinguish whether low take-up reflects heterogeneous treatment effects or systematic mistakes. Third, we show that a simple change to contract terms generates substantial financial savings: mandatory structured payments lower financial cost by 183 pesos and reduce default by 7.7 percentage points, yet only 11\% voluntarily choose this contract. We find substantial cost savings for non-choosers and no evidence of selection on gains. To our knowledge, this is the first such result in the household finance literature.

At first it may seem puzzling that lenders do not offer structured payments contracts. Our results suggest a plausible explanation. Lender profits are higher when the borrower defaults. Paired with low demand for structured payments from borrowers---at least in the short run, when borrowers lack experience with the contract---there seems to be little competitive incentive to offer it. This creates a catch-22: borrowers will not learn the benefits of structured payments if lenders never introduce such contracts, but there may be little incentive for their introduction. Our context may thus be a case of ``veiled paternalism'' from \cite{Laibson2018} turned on its head: principals do not embed commitment into contracts in a manner that is not obvious to agents, benefiting the principal at their expense. 

While the Mandates versus Choice design enables us to estimate a range of important treatment effects, our experiment was not designed to identify preferences or pin down specific mechanisms. We make some progress nonetheless. We find that borrowers who are sure they will recover their pawn drive the positive TUT effect, despite having a 42\% default rate under the status quo. We show this pattern can be rationalized by a simple model in which the state-contingent features of the structured payments contract interact with overconfidence. Overconfident borrowers do not demand structure, but benefit most when it is mandated because their optimism is misplaced.

These results open promising directions for future research. First, it would be fruitful to test whether repeated experience with structured payments increases demand for them and reduces overconfidence.\footnote{Empirical research shows overconfidence can persist in investing \citep{BarberOdean2001}, management \citep{HuffmanRaymondShvets2022}, expert performance \citep{HeckBenjaminSimonsChabris2024}, and household forecasting \citep{StangoZinman2020}, and can be supported by motivated beliefs and selective/distorted memory \citep{BenabouTirole2002,HuffmanRaymondShvets2022}, asymmetric updating \citep{EilRao2011}, cognitive uncertainty \citep{EnkeGraeber2023}, and noisy feedback \citep{MooreHealy2008}.} Whether increased demand would be sufficient to overcome lenders' incentive to profit from default remains an open question.
Second, following a growing misperceptions literature, one could test whether providing information about true default probabilities increases demand for structure. Third, collecting in-depth measures of preferences could help separately identify the roles of present bias and overoptimism and permit more definitive welfare conclusions. Given the scale of pawn lending worldwide, our findings suggest that relatively simple changes to contract design could lower financial costs for millions of borrowers.

\ifdefined\NoDataAvailability\else
\input{data-availability}
\fi

  \begingroup
  \setstretch{1.10}
  \renewcommand{\refname}{References}
  \putbib
  \endgroup
\end{bibunit}

\clearpage
\newpage

\setcounter{table}{0}
\setcounter{figure}{0}
\setcounter{section}{0}
\pagenumbering{arabic}
\renewcommand\thefigure{OA-\arabic{figure}}
\renewcommand\thetable{OA-\arabic{table}}
\renewcommand*{\thepage}{OA - \arabic{page}}
\renewcommand\thesection{Appendix \Alph{section}.}
\renewcommand\thesubsection{\Alph{section}.\arabic{subsection}}

\begin{center}
	\Large ONLINE APPENDIX: Structured Payment in Pawnshop Borrowing: Mandates vs.\ Choice \\[0.5em]
	\large Francis J.\ DiTraglia, Craig McIntosh, Isaac Meza, Joyce Sadka and Enrique Seira
\end{center}

\begin{appendix}

\begin{bibunit}
    
\section{Additional materials}

\begin{figure}[H]
     \caption{Contract Terms Summary, and Promise Slip}
    \label{PaperSlip}
    \begin{center}
    \begin{subfigure}{0.75\textwidth}
        \centering
        \includegraphics[width=\textwidth]{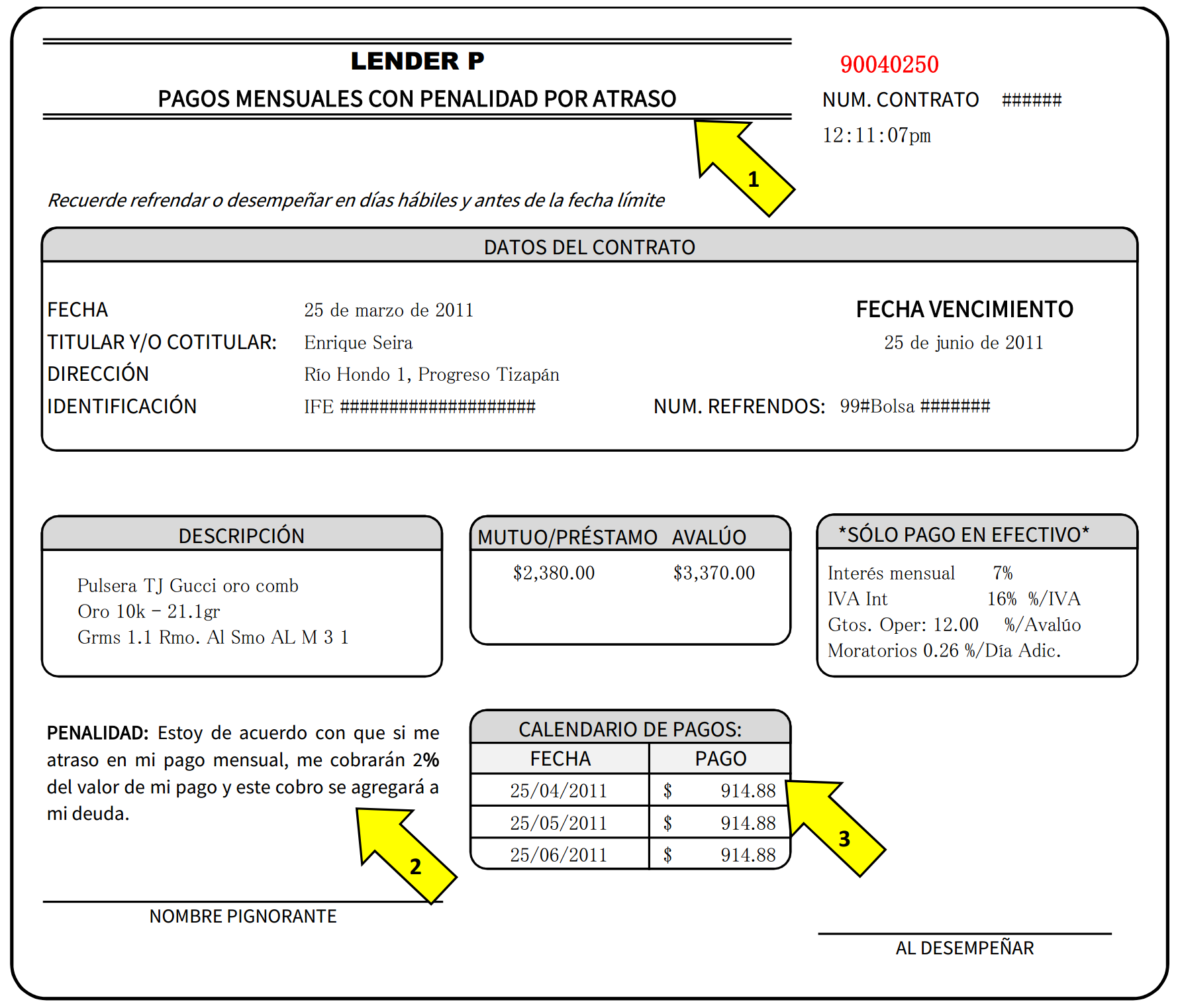}
    \end{subfigure}
    
    \vspace{3ex}
   
    \end{center}
    \scriptsize
        This figure is a sample receipt that was given to clients that got assigned to the structured-payments contract (the font and some aspects of format were changed to protect Lender's P identity). We want to highlight the salience of some items. First the title clearly indicated which contract the client has (arrow 1). Second, in the case of the structured payments contract it clearly indicates that there is a fee for paying late equivalent to 2\% of the value of the monthly payment (arrow 2). Third, there is a calendar for payments clearly specifying the dates and amounts to pay each month. 
\end{figure}

\begin{figure}[H]
    \caption{Weekly default rates experimental branches and all branches}
    \label{external_validity}
    \begin{center}
    \begin{subfigure}{0.65\textwidth}
        \centering
        \includegraphics[width=\textwidth]{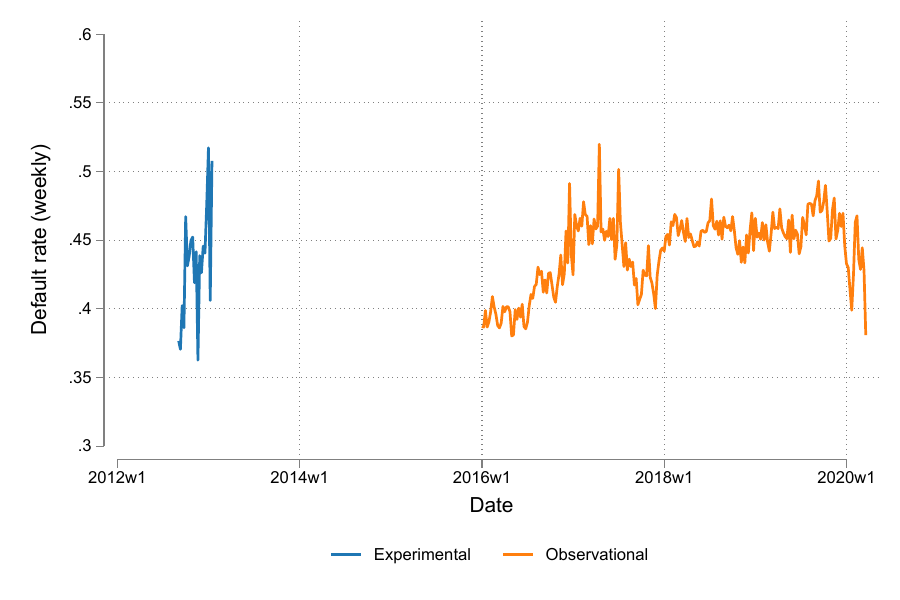}
    \end{subfigure}
    \end{center}
\footnotesize{The above figure shows the weekly default rates for our experimental sample, and for 4 years after our experiment. We show that the default rates in the experimental sample are not atypical. }
    \label{weekly_def_rates}
\end{figure}

\begin{figure}[H]
     \caption{Behavior of borrowers who lost their pawn.}
    \begin{center}
    \begin{subfigure}{0.35\textwidth}
        \caption{Elapsed days to first payment}
        \centering
        \includegraphics[width=\textwidth]{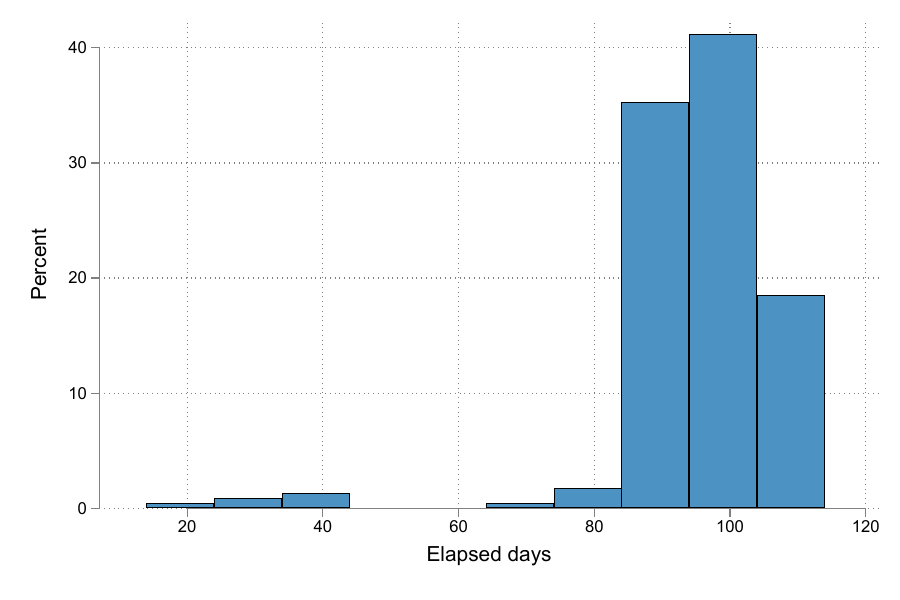}
    \end{subfigure}
    \begin{subfigure}{0.35\textwidth}
        \caption{Elapsed days to last payment}
        \centering
        \includegraphics[width=\textwidth]{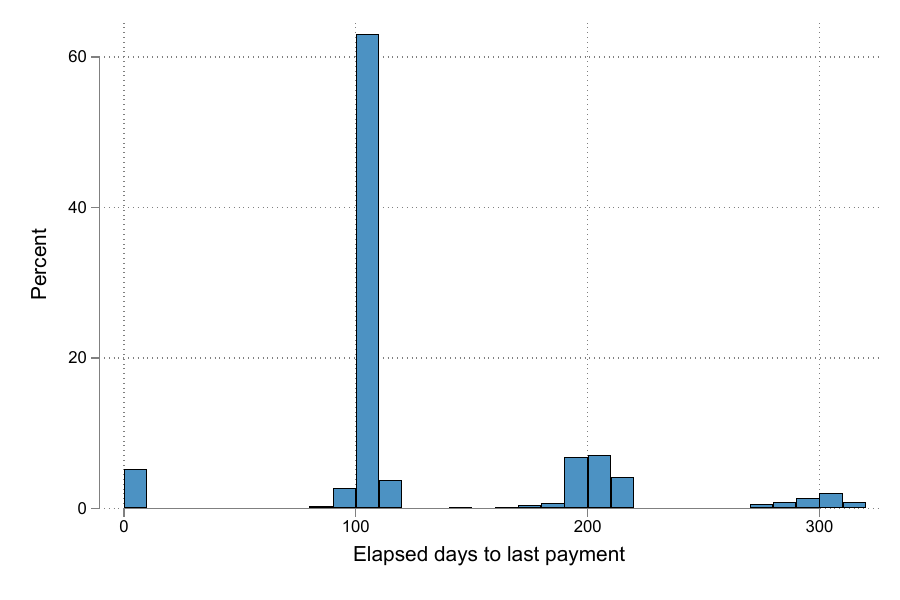}
    \end{subfigure}
        \begin{subfigure}{0.35\textwidth}
        \caption{Payments as \% of loan}
        \centering
        \includegraphics[width=\textwidth]{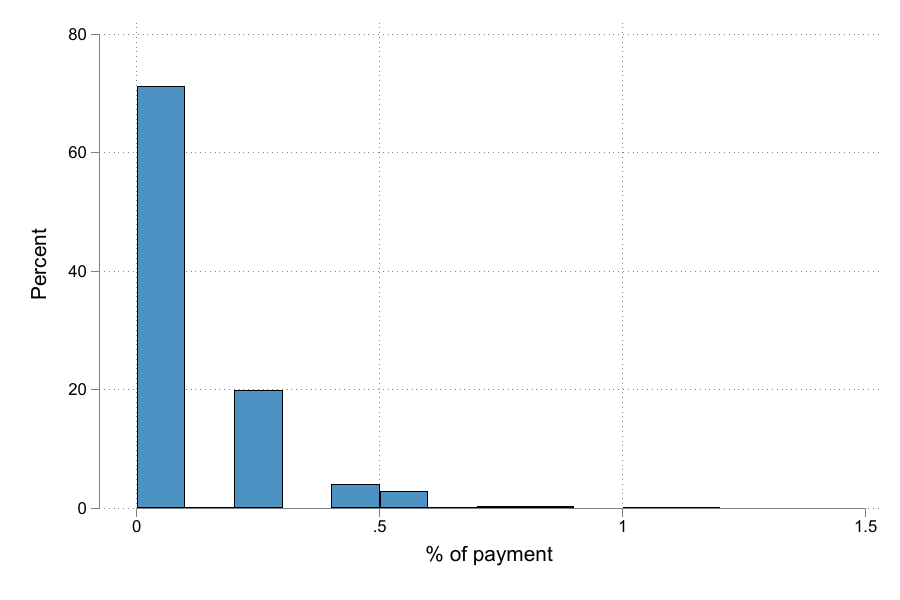}
    \end{subfigure}
    \begin{subfigure}{0.35\textwidth}
        \caption{Number of payments}
        \centering
        \includegraphics[width=\textwidth]{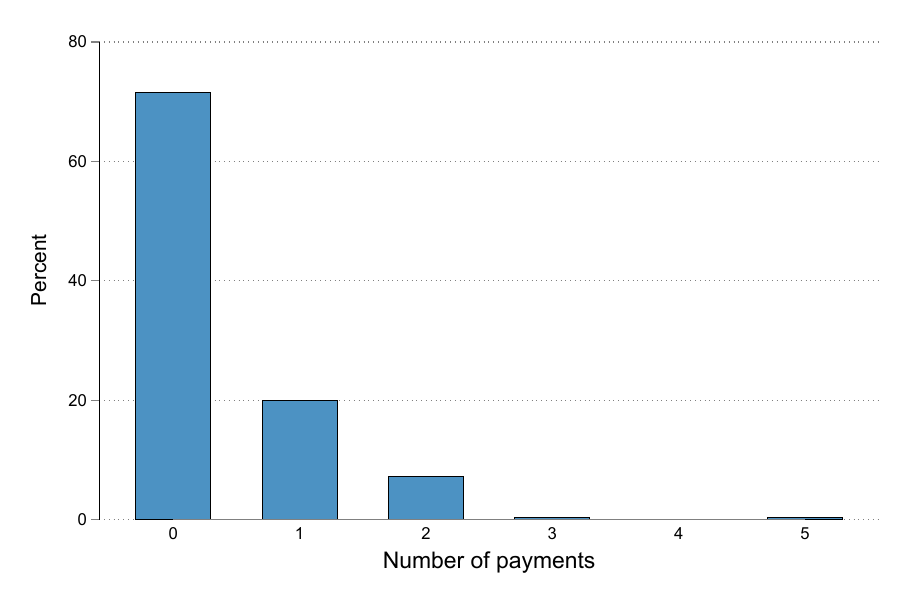}
    \end{subfigure}
    \end{center}
        \footnotesize{This figure provides more details on the behavior of clients who were assigned to the control group and did not recover their pawn. Panel (a) shows a histogram of days elapsed from the pawn to the first payment, while panel (b) displays a histogram of days elapsed until the last payment. Some borrowers make payments after day 105, the end of the grace period: if they pay all interest owed, they can ``restart'' the loan. This amounts to starting a new loan with the same conditions and same pawn. Panel (c) shows a histogram of the fraction of the loan paid, while panel (d) presents a barplot of the number of times that borrowers went to the branch to make payments.}
        \label{proxy_naive}
\end{figure}

\normalsize
\normalsize

\begin{figure}[H]
    \caption{Determinants of choice.}
    \begin{center}
    \begin{subfigure}{0.65\textwidth}
        \centering
        \includegraphics[width=\textwidth]{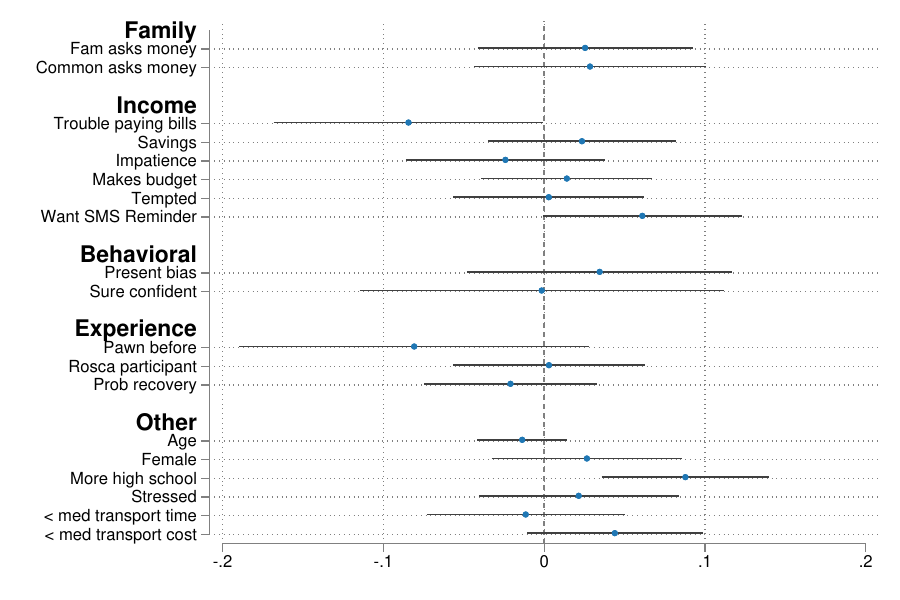}
    \end{subfigure}
    \end{center}
\footnotesize{This figure shows estimated coefficients and standard errors from a multivariate OLS regression that predicts an indicator for take-up of structured payments using the specified covariates for borrowers in the choice arm. Overall, our baseline survey variables are relatively weak predictors of contract choice.} 
    \label{determinants_choose}
\end{figure}

\subsection{No differential selection or imbalance: full details}
\label{app:integrity}

Here we provide full details of the no-differential-selection and balance evidence summarized in Section \ref{sec:integrity} of the body.

The validity of our experimental design depends critically on the absence of differential selection: borrowers must not selectively avoid the branch on days when certain contracts are offered. Three pieces of evidence establish that no such selection occurred.

First, Table \ref{attrition_table_final} uses administrative data to show there are no differences in the number of loans, borrowers, or amount borrowed across arms. Each row corresponds to a regression of the specified variable against indicators for treatment assignment (Status quo/control, Structured payment, Choice), with no constant. In the first row, we use the sample of borrowers who answered the baseline survey to create the binary indicator ``loan take-up'' that equals 1 if the surveyed borrower proceeded to take a loan. Because we surveyed potential borrowers before they learned the treatment assignment, their decision to take a loan should not differ across treatment arms. The overwhelming majority (96.1\%) of those who answered the baseline survey took a loan after learning the contract terms. This rate is not only high but also virtually identical across treatment arms: 96.7\% in the control arm versus 95.6\% and 96.2\% in the structured payments and choice arms, respectively. With standard errors close to 1\%, we cannot reject equality of means (p=0.86). 
These identical take-up rates across treatment arms demonstrate that borrowers did not differentially select into treatment conditions.

Second, rows 2-4 of Table \ref{attrition_table_final} use administrative data to test whether the average number of borrowers, pawns per borrower, and pawns per day are the same across experimental arms.  We cannot reject the null hypothesis that they are equal (p-values of 0.17, 0.43, and 0.24), indicating that borrowers did not selectively avoid pawning on days where the structured payments arm or choice arm was assigned.\footnote{Cluster-robust standard errors are on the order of 6 borrowers per day, even though our sample size is relatively large compared to similar studies. We have 80\% power to detect 7pp effects on default (a main outcome) and a difference of 20 loans per day (not an outcome studied).} 
We were particularly concerned that potential borrowers would leave the branch after hearing that the status quo contract was unavailable. However, we find little evidence of this: we observe 20.8 borrowers on status quo days (95\% CI: 14 to 27) and 22.2 on structured payment days (CI: 14 to 30).  The number of borrowers on choice days is 25.4, one standard error above the status quo arm. This appears to be due to sampling variability and a few outliers on choice days. The differences across arms are smaller when we examine medians and disappear when we winsorize at the 95th percentile. Rows 5-6 of Table \ref{attrition_table_final} show the same pattern for loan amounts: we cannot reject equality across arms.

\begin{table}[H]
\caption{No selection across arms}
\label{attrition_table_final}
\begin{center}
\resizebox{0.65\textwidth}{!}{
\footnotesize{
\begin{tabular}{lcccr}
\toprule
      & Control & Structure & Choice & \multicolumn{1}{c}{p-value} \\
\midrule
\midrule
Take-up & 0.967 & 0.956 & 0.962 & 0.86 \\
      & (0.01) & (0.01) & (0.01) &  \\
\midrule
Number of borrowers & 20.8  & 22.2  & 25.4  & 0.17 \\
      & (3.29) & (3.9) & (4.89) &  \\
\rowcolor[rgb]{ .949,  .949,  .949} \multicolumn{1}{r}{median} & 19    & 20    & 21    & 0.46 \\
\midrule
Number of pawns/borrower & 1.4   & 1.4   & 1.4   & 0.43 \\
      & (0.08) & (0.04) & (0.05) &  \\
\rowcolor[rgb]{ .949,  .949,  .949} \multicolumn{1}{r}{median} & 1.4   & 1.3   & 1.3   & 0.54 \\
\midrule
Number of pawns  & 31    & 31.3  & 37.2  & 0.24 \\
      & (5.8) & (5.6) & (7.9) &  \\
\rowcolor[rgb]{ .949,  .949,  .949} \multicolumn{1}{r}{median} & 27    & 28    & 30    & 0.43 \\
\midrule
Amt borrowed/borrower & 2266.8 & 2094  & 2115.2 & 0.18 \\
      & (101.8) & (83.7) & (99.9) &  \\
\rowcolor[rgb]{ .949,  .949,  .949} \multicolumn{1}{r}{median} & 2154.3 & 2041  & 2047.5 & 0.68 \\
\midrule
Total borrowed & 47877 & 47813 & 54780 & 0.4 \\
      & (8005) & (9436) & (12587) &  \\
\rowcolor[rgb]{ .949,  .949,  .949} \multicolumn{1}{r}{median} & 37520 & 39420 & 40850 & 0.73 \\
\midrule
Obs   & 85    & 81    & 94    &  \\
\bottomrule
\bottomrule
\end{tabular}%
}
}
\end{center}
\scriptsize{Each row in this table presents results for a regression at the branch-day level of the specified outcome variable on indicators of each experimental arm---control, mandatory structure, and choice---excluding an intercept. The table reports the coefficients on each of these indicators, standard errors clustered at the branch level, and p-values from an F-test of the null hypothesis of equality of the three coefficients, computed from a cluster-robust Wald statistic. We also report median regression results for all outcomes besides Take-up.
 }
\end{table}

Third, we cannot reject equality of \textit{borrowers'} characteristics across arms. Table \ref{SS_final} uses data from the baseline survey and estimates the same regression specification as above.  Again this is consistent with no differential selection across arms.
All p-values (except for the subjective probability of recovery) are above 0.25 and the estimated differences across arms have tight confidence intervals. For instance, the subjective probabilities of recovery are 91.9, 91.7, and 93.6 across the three arms, with standard errors of 0.7, 1, and 0.6, respectively.  

Taken together, these tables constitute strong evidence that there was no differential borrower selection across arms. While a differentially selected pool of borrowers would still provide useful information for the lender about the profitability of different contracts, it would not allow for the analysis we conduct comparing potential outcomes across arms for the same borrowers.

\begin{table}[H]
\caption{Borrowers' characteristics are balanced}
\label{SS_final}
\begin{center}
\resizebox{0.65\textwidth}{!}{
\footnotesize{
\begin{tabular}{lcccc}
\toprule
      & \multicolumn{1}{p{4.5em}}{Control} & \multicolumn{1}{p{4.93em}}{Structure} & \multicolumn{1}{p{3.43em}}{Choice} & \multicolumn{1}{p{3.43em}}{p-value} \\
\midrule
      & \multicolumn{4}{c}{Panel : Survey Data} \\
\midrule
\midrule
Subjective value & 4084  & 3877  & 4173  & 0.51 \\
      & (186) & (193) & (172) &  \\
Trouble paying bills & 0.19  & 0.21  & 0.18  & 0.67 \\
      & (0.024) & (0.023) & (0.02) &  \\
Present bias & 0.14  & 0.13  & 0.13  & 0.89 \\
      & (0.02) & (0.01) & (0.01) &  \\
Makes budget & 0.62  & 0.59  & 0.65  & 0.29 \\
      & (0.028) & (0.036) & (0.021) &  \\
Subj. pr. of recovery & 91.89 & 91.65 & 93.61 & 0.09 \\
      & (0.721) & (1.031) & (0.582) &  \\
Pawn before & 0.87  & 0.89  & 0.9   & 0.25 \\
      & (0.015) & (0.013) & (0.011) &  \\
Age   & 43.32 & 42.85 & 43.82 & 0.73 \\
      & (0.688) & (0.949) & (0.792) &  \\
Female & 0.73  & 0.72  & 0.71  & 0.88 \\
      & (0.023) & (0.019) & (0.02) &  \\
+ High-school & 0.66  & 0.67  & 0.65  & 0.84 \\
      & (0.027) & (0.022) & (0.018) &  \\
\midrule
Obs   & 1386  & 1469  & 1982  &  \\
\bottomrule
\bottomrule
\end{tabular}%
}
}
\end{center}
\scriptsize {The table uses survey data to show balance across arms. Survey data is not used for the main results of the paper. It is only used in Sections \ref{sec:RF} and \ref{why_paternalism}. Each row in this table presents regression results in which the unit of observation is a borrower who completed our baseline survey. Each regresses the specified survey covariate on indicator variables for each treatment arm---control, mandatory structure, and choice---without an intercept. The table reports estimates of these coefficients, standard errors clustered at the branch-day level, and p-values from an F-test of the null hypothesis of equality of the three coefficients, computed from a cluster-robust Wald statistic. ``Subjective value'' of the pawn refers to the client's minimum selling price for the pawn in pesos, ``Trouble paying bills'' is an indicator for the borrower reporting having problems paying utility bills in the last 6-months. Present bias is constructed as in \cite{Ashraf}; ``Makes budget'' as an indicator for whether the household makes an expense budget for the month ahead of time. The subjective probability of recovery was elicited \`a la Manski (from 0 to 100 what is the probability that you will recover your pawn); pawned before equals one if the client reports having pawned before (although not necessarily with Lender $P$); age is age of the borrower in years, +High-school is a dummy that indicates if the client has completed high school. 
}
\end{table}

\subsection{Censoring} \label{App_censoring}

Some loans in our sample are ``censored'' in that they continue beyond our observation period.  For these loans, we do not know whether the borrower ultimately defaulted or recovered her pawn.   We have also shown that one effect of the forcing arm is to accelerate repayment, meaning that it is less likely for loans in this arm to be censored.  This issue is illustrated in Figure \ref{survival_graph}, which shows the CDF of loan completion (either default or recovery in Panel (a)) and loan recovery (Panel (b)) by the number of days since first pawn.  Two features of these graphs are salient for our analysis.  The first is the extent to which loan outcomes are observed more quickly in the mandatory structure arm.  This is primarily due to the substantially higher rate of repayment of Mandatory structure loans at 120 days (15 pp higher than the other arms).  The second is the very low rate at which loans are recovered in any arm after 120 days.  In the 180-320 day window loans are largely dormant, suggesting that many of the censored loans will in fact end in default.

The confluence of censoring and a treatment effect on censoring is potentially problematic from an experimental point of view.  The approach taken in the headline results is a conservative one in that it inherently assumes that all of the loans outstanding at the end of the observation window will be repaid, making it so that the acceleration of payment observed in the Mandatory arm does not translate mechanically into the that treatment decreasing default.  Nonetheless, to be certain that this issue is not driving our results we conduct a bounding exercise to understand how large the effects of this problem can possibly be.  

One way of considering the effect that this issue could have on our results is to make extreme assumptions about the outcome of these loans in the treatment and control so as to bound the possible influence of censoring. In Table \ref{bounding_censoring} we compare the Mandatory and Control arms, bounding the censoring issue by reversing assumptions about the outcome of censored loans in the treatment versus the control.   Panel B provides the lower bound for the treatment effect (closest to zero) by assuming censored control loans are always repaid and treatment loans never are; even in this lower-bound case the treatment effect is cost-reducing and significant at the 1\% significance level.
Panel C estimates the upper bound by making the reverse assumption.  Comfortingly, even with these extreme assumptions the significance on the main treatment effects never flips and treatment effects on financial cost and interests payments remain negative and significant in all scenarios.  So there appears to be no scope for the censoring issue to overturn our main results. 

Finally, Panel E of this table conducts a logit prediction model that uses all of the available information on payment behavior for loans that were completed to predict the outcome of loans that were not.  This is a ``best guess'' of the outcome on censored loans.  Using this prediction, we replicate the main experimental results and find that the treatment effect on financial cost increases from -183 (main results) to -235 (censored loans predicted), and the CR from -6\% to -8\%.  Hence, while the censoring issue does have an effect on the magnitude of our estimated treatment effects, these  checks confirm that (a) the core results are fully robust to censoring, and (b) the headline approach that we take to the issue is conservative and likely understates the true magnitude of impacts.

\begin{table}[H]
\caption{Bounding censoring} 
\label{bounding_censoring}
\begin{center}
\resizebox{0.75\textwidth}{!}{
\footnotesize{
\begin{tabular}{lcccccc}
\toprule
      & FC    & Interest pymnt & Principal pymnt & Lost pawn value & Default & CR \\
\midrule
      & \multicolumn{6}{c}{Panel A : $\quad$ Control  = 0           $\quad\quad$                 Mandatory structured  = 0} \\
\midrule
\midrule
      & (1)   & (2)   & (3)   & (4)   & (5)   & (6) \\
\midrule
\midrule
Mandatory structured & -213.2*** & -154.3*** & -1.62 & -91.3*** & -0.077*** & -0.076*** \\
      & (53.8) & (44.8) & (3.15) & (34.5) & (0.025) & (0.014) \\
      &       &       &       &       &       &  \\
\midrule
Observations & 3724  & 3724  & 3724  & 3724  & 3724  & 3724 \\
R-sq  & 0.008 & 0.007 & 0.000 & 0.003 & 0.006 & 0.025 \\
Control Mean & 989.9 & 593.4 & 5.96  & 396.5 & 0.44  & 0.45 \\
\midrule
\midrule
      &       &       &       &       &       &  \\
\midrule
      & \multicolumn{6}{c}{Panel B : $\quad$ Control  = 0         $\quad\quad$                    Mandatory structured = 1} \\
\midrule
\midrule
      & (7)   & (8)   & (9)   & (10)  & (11)  & (12) \\
\midrule
\midrule
Mandatory structured & -166.8*** & -171.7*** & -0.055 & -27.6 & -0.0016 & -0.050*** \\
      & (55.6) & (44.2) & (3.43) & (35.6) & (0.027) & (0.016) \\
      &       &       &       &       &       &  \\
\midrule
Observations & 3724  & 3724  & 3724  & 3724  & 3724  & 3724 \\
R-sq  & 0.005 & 0.009 & 0.000 & 0.000 & 0.000 & 0.009 \\
Control Mean & 989.9 & 593.4 & 5.96  & 396.5 & 0.44  & 0.45 \\
\midrule
\midrule
      &       &       &       &       &       &  \\
\midrule
      & \multicolumn{6}{c}{Panel C : $\quad$ Control  = 1        $\quad\quad$                     Mandatory structured = 0} \\
\midrule
\midrule
      & (13)  & (14)  & (15)  & (16)  & (17)  & (18) \\
\midrule
\midrule
Mandatory structured & -292.5*** & -106.9*** & -3.35 & -218.1*** & -0.21*** & -0.11*** \\
      & (56.8) & (40.7) & (3.27) & (33.5) & (0.024) & (0.016) \\
      &       &       &       &       &       &  \\
\midrule
Observations & 3724  & 3724  & 3724  & 3724  & 3724  & 3724 \\
R-sq  & 0.014 & 0.004 & 0.000 & 0.017 & 0.044 & 0.040 \\
Control Mean & 1069.2 & 545.9 & 7.69  & 523.3 & 0.57  & 0.48 \\
\midrule
\midrule
      &       &       &       &       &       &  \\
\midrule
      & \multicolumn{6}{c}{Panel D : $\quad$ Control  = 1       $\quad\quad$                      Mandatory structured = 1} \\
\midrule
\midrule
      & (19)  & (20)  & (21)  & (22)  & (23)  & (24) \\
\midrule
\midrule
Mandatory structured & -246.1*** & -124.2*** & -1.79 & -154.4*** & -0.13*** & -0.085*** \\
      & (58.5) & (40.1) & (3.54) & (34.6) & (0.026) & (0.018) \\
      &       &       &       &       &       &  \\
\midrule
Observations & 3724  & 3724  & 3724  & 3724  & 3724  & 3724 \\
R-sq  & 0.009 & 0.006 & 0.000 & 0.008 & 0.018 & 0.020 \\
Control Mean & 1069.2 & 545.9 & 7.69  & 523.3 & 0.57  & 0.48 \\
\midrule
\midrule
      &       &       &       &       &       &  \\
\midrule
      & \multicolumn{6}{c}{Panel E : $\quad$ Prediction with lasso-logit model} \\
\midrule
\midrule
      & (25)  & (26)  & (27)  & (28)  & (29)  & (30) \\
\midrule
\midrule
Mandatory structured & -234.9*** & -135.2*** & -1.94 & -132.1*** & -0.12*** & -0.082*** \\
      & (56.7) & (42.5) & (3.46) & (34.3) & (0.025) & (0.016) \\
Choice  & -7.02 & 6.88  & -1.83 & -15.5 & -0.027 & 0.0059 \\
      & (67.4) & (52.2) & (3.42) & (35.5) & (0.023) & (0.019) \\
      &       &       &       &       &       &  \\
\midrule
Observations & 6304  & 6304  & 6304  & 6304  & 6304  & 6304 \\
R-sq  & 0.007 & 0.005 & 0.000 & 0.005 & 0.010 & 0.019 \\
Control Mean & 1034.5 & 563.4 & 7.69  & 471.2 & 0.52  & 0.47 \\
\bottomrule
\bottomrule
\end{tabular}%
}
}
\end{center}
 
\end{table}
 \footnotesize{ Given the censored loans, i.e. loans that have not finished by the end of the observation period, we estimate `a la Manski' bounds for these loans, meaning that we impute all loans to either \emph{default}$=1$ or \emph{recovery}$=0$ depending on the treatment arm. Different panels perform different imputations for the censored loans for all possible combinations for the imputation, and computes the ATE for the same outcomes of Table \ref{main_impact_table}. Panel A, for instance, assumes that all outstanding loans are fully payed. Panel B is the most conservative imputation since it assumes all outstanding loans in the control arm are repaid, while all the outstanding loans in the mandatory structured arm end in default. Panel C, on the other hand, is the most optimistic scenario opposite to that of Panel B. Panel D assumes all remaining loans default. The last panel makes the imputation to the censored loans according to the best prediction using a piecewise lasso logit model for default. In particular, we build two logit models with lasso regularization, depending on whether the loan duration is less than 220 days (two cycles) or more than 220 days. For prediction we use the former whenever the last recorded payment was done within 220 days, and the latter otherwise. Both models includes loan characteristics (loan size, branch), and payment behavior (loan duration so far, days to first payment, \% of first payment, \% of payments at 30, 60, 90, and 105 days, and \% of interest payed at 105 days), but the latter model also includes \% of payments at 150, 180, and 210 days. Together, these models achieve an apparent pooled in-sample accuracy rate of 87\%.
 Note that in all panels we maintain significant results for Financial Cost as dependent variable, and even in the most conservative scenario (Panel B).}

\begin{figure}[H]
\caption{Survival graph.}
    \begin{center}
   \begin{subfigure}{0.49\textwidth}
   \caption{Ended contract}
        \centering
        \includegraphics[width=\textwidth]{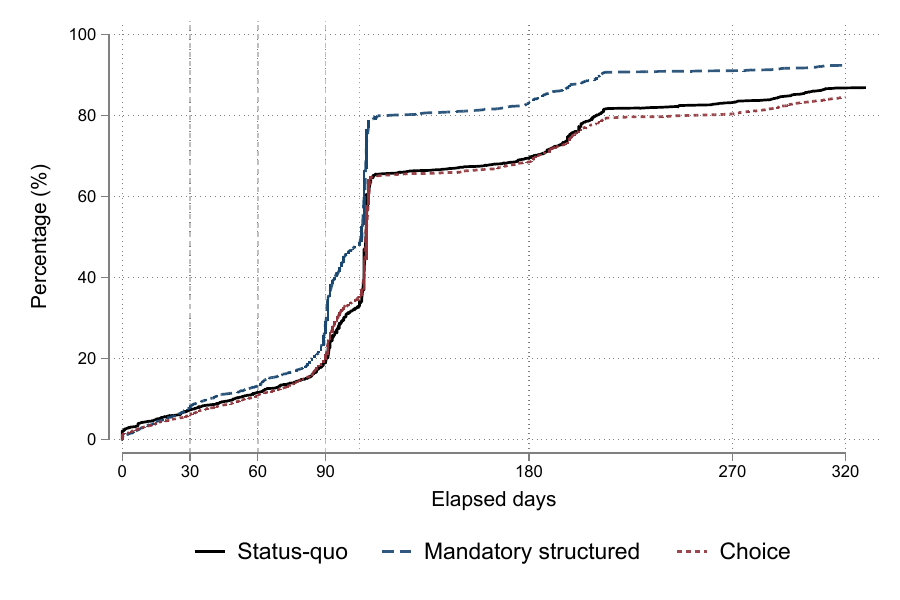}
    \end{subfigure} 
   \begin{subfigure}{0.49\textwidth}
   \caption{Recovery}
        \centering
        \includegraphics[width=\textwidth]{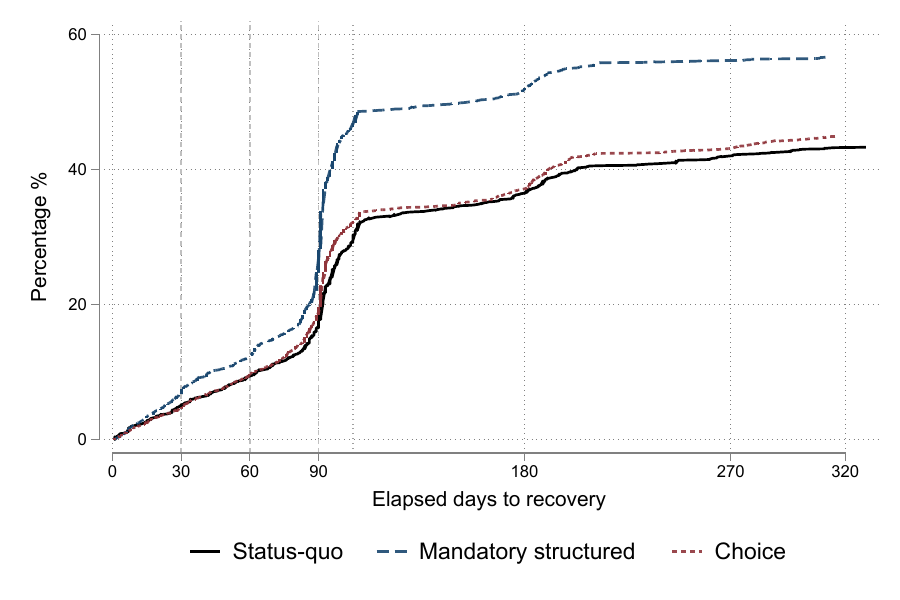}
    \end{subfigure}     
    \end{center}
      \footnotesize{This Figure shows the CDF of loan completion either default or recovery in Panel (a), or loan recovery in Panel (b), by the number of days since first pawn.}
      \label{survival_graph}
\end{figure}

\begin{figure}[H]
 \caption{Percentage of loan paid overtime (by experimental arm)}
    \begin{center}
   \begin{subfigure}{0.49\textwidth}
        \caption{Structured payment arms}
        \centering
        \includegraphics[width=\textwidth]{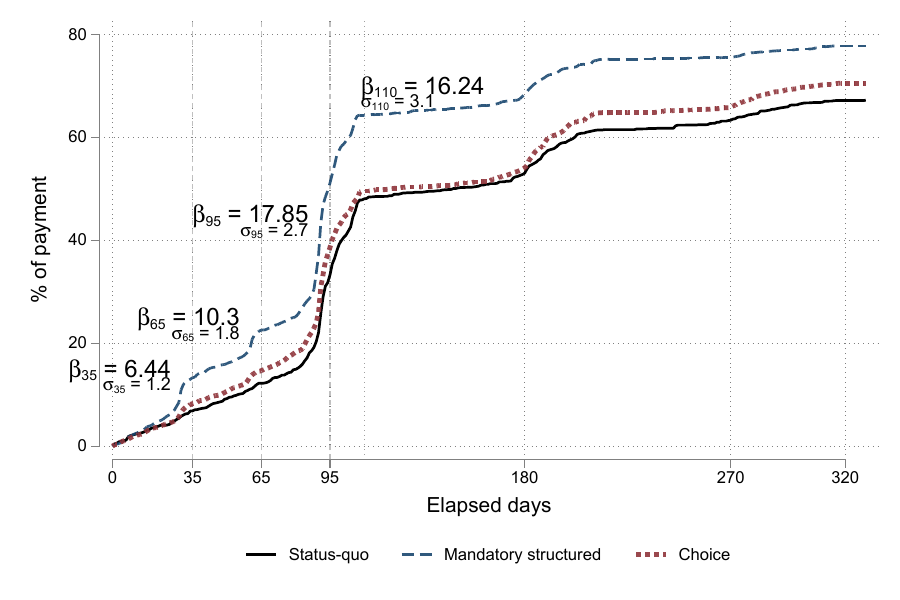}
    \end{subfigure} 
   \begin{subfigure}{0.49\textwidth}
        \caption{Promise arms}
        \centering
        \includegraphics[width=\textwidth]{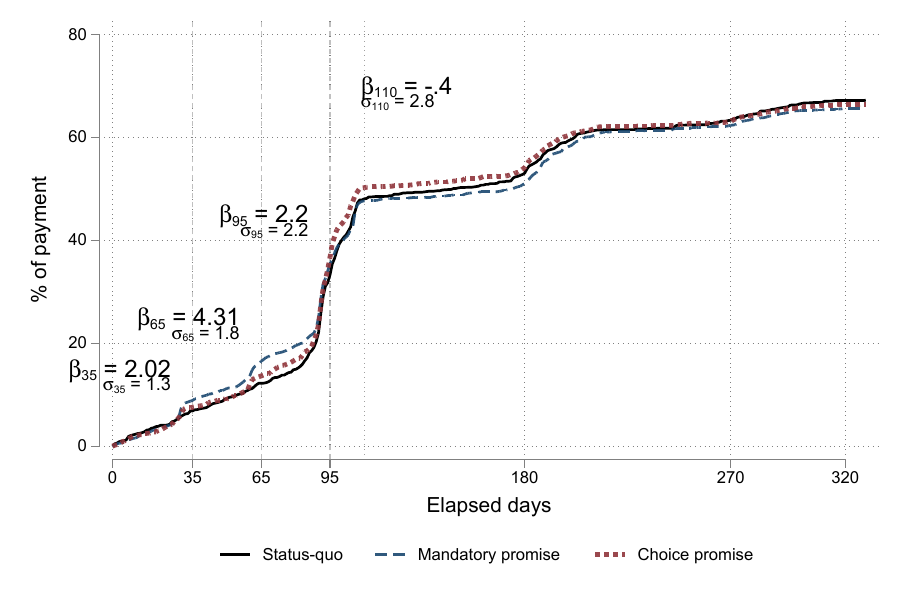}
    \end{subfigure}     
    \end{center}
   \footnotesize{This figure shows the accumulated percentage of recovery in time in the vertical axis against the days since the pawn loan was taken, separately by treatment arm. It also displays the coefficients and standard errors from four separate regressions comparing the means of treatment and control arms at 35 days, 65 days, 95 days and 110 days (using only observations from those days, regress \% of payment vs.\ a treatment dummy).}
    \label{porc_payment_over_time}
\end{figure}

\subsection{Robustness accounting for other costs}

\begin{table}[H]
\caption{Effects on more comprehensive cost measures}
\label{table_robustness_fc}
\begin{center}
\resizebox{0.95\textwidth}{!}{
\footnotesize{
\begin{tabular}{lccccc}
\toprule
      & FC    & FC (subj.value) & FC +  tc & FC - interest & FC (subj.value) + tc - int \\
\midrule
      & (1)   & (2)   & (3)   & (4)   & (5) \\
\midrule
\midrule
Mandatory structured & -183.0*** & -292.8*** & -182.8*** & -107.1*** & -168.4** \\
      & (50.9) & (89.9) & (52.6) & (38.5) & (81.1) \\
Choice  & -11.9 & -30.7 & -1.30 & -31.2 & -28.3 \\
      & (56.6) & (90.9) & (59.1) & (39.9) & (78.7) \\
      &       &       &       &       &  \\
\midrule
Observations & 6304  & 6304  & 6304  & 6304  & 6304 \\
R-squared & 0.005 & 0.004 & 0.005 & 0.002 & 0.002 \\
Control Mean & 942.4 & 1389.9 & 1026.1 & 480.7 & 927.7 \\
\midrule
\midrule
      &       &       &       &       &  \\
\midrule
      & CR    & CR (subj.value) & CR +  tc & CR - interest & CR (subj.value) + tc - int \\
\midrule
      & (6)   & (7)   & (8)   & (9)   & (10) \\
\midrule
\midrule
Mandatory structured & -0.062*** & -0.10*** & -0.058*** & -0.040*** & -0.053** \\
      & (0.012) & (0.023) & (0.015) & (0.013) & (0.023) \\
Choice  & -0.0011 & -0.020 & 0.0059 & -0.024** & -0.028 \\
      & (0.013) & (0.020) & (0.017) & (0.012) & (0.020) \\
      &       &       &       &       &  \\
\midrule
Observations & 6304  & 6304  & 6304  & 6304  & 6304 \\
R-squared & 0.016 & 0.009 & 0.011 & 0.003 & 0.002 \\
Control Mean & 0.43  & 0.65  & 0.50  & 0.23  & 0.49 \\
\bottomrule
\bottomrule
\end{tabular}%
}
}
\end{center}
 \footnotesize{ This table augments the measure of financial cost presented in Table \ref{main_impact_table} with measures of transaction costs, subjective costs, and adjustments for liquidity costs. Panel A reports financial cost in pesos, while Panel B shows CR. Columns (1) and (6) replicate our previous results for comparability. Columns (2) and (7) of Table \ref{table_robustness_fc} use the subjective value of the pawn reported by the borrower rather than its appraised value. Columns (3) and (8) adjust for self-reported transport costs per visit plus an entire day's wage, both multiplied by the number of visits that each individual made. For pawns with missing self-reported transport cost after recovering survey responses across all pawns belonging to the same borrower, we adjust using the overall mean self-reported transport cost among pawn observations with nonmissing transport cost. Columns (4) and (9) adjust to consider the liquidity cost. Finally, columns (5) and (10) include all three changes together. The main takeaway from the table is that results are quite robust to including a much expanded measure of costs.
 Standard errors are clustered at the branch-day level.}
\end{table}

\normalsize

\normalsize

\subsection{Lender Profit}
\label{append:lender-profit}

Here we present further details of the lender profit calculation from Section \ref{sec:baseline}. Figure \ref{profit_mandatory_control} presents results for the simple geometric model described in the body of the text. This exercise assumes that the per-loan profit generated by a given borrower is independent of her probability of returning to take out another loan.

\begin{figure}[H]
    \caption{Profit difference Mandatory vs.\ Control arm.}
    \begin{center}
    \begin{subfigure}{0.65\textwidth}
        \centering
        \includegraphics[width=\textwidth]{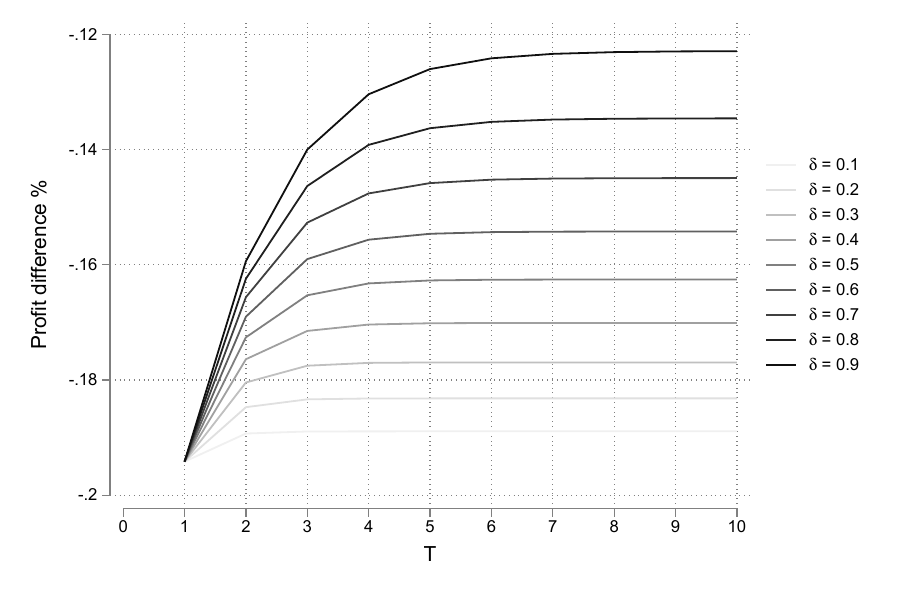}
    \end{subfigure}
    \end{center}
\footnotesize{The above figure shows the difference in profit $\text{Profit}_{j,z}  = \sum_{t=0}^T\delta^t\,\mathbb{E}\{\text{Financial Cost}_j | \text{Contract}=z\}\times\Pr(\text{Repeat} | \text{Contract}=z)^t$ for the Mandatory vs.\ Control contracts for different values of the discount rate $\delta$ and time horizon $T$.  }
    \label{profit_mandatory_control}
\end{figure}

We now present results from two exercises that relax the independence assumption. 
Let $X_i$ denote per-loan lender profit, assumed time-invariant,  and let $R_{i,t} \in\{0,1\}$ be an indicator for returning after period $t$. Because we cannot observe $R_{i,t}$ for $t\geq 2$, we assume a time-invariant return propensity $P_i$. For contract $Z_i=z$, the lender's expected profit is
\[
\Psi_z=\mathbb{E}\left[X_i \sum_{t=0}^T\left(\delta P_i\right)^t \mid Z_i=z\right].
\]
Note that this expression does not assume independence between $X_i$ and $P_i$.

Our first approach to allowing dependence between $X_i$ and $P_i$ is as follows.
Let $\eta_i = \text{logit}(P_i)$ and assume that
\begin{align*}
X_i&=\beta_0+\beta_Z Z_i+a L_i+\varepsilon_i^X \\
\eta_i&=\alpha_0+\alpha_Z Z_i+b L_i+\left(\gamma_0+\gamma_Z Z_i\right) X_i \\
L_i &\sim \mathcal{N}(0,1)\perp \varepsilon_i^X\sim \mathcal{N}(0,\sigma^2_X).
\end{align*}
In this model the common latent variable $L_i$ creates dependence between $X_i$ and $\eta_i \equiv \text{logit}(P_i)$.
We estimate this model by GSEM (linear $X_i$ equation and Bernoulli-logit for $R_{i, 1}$), fixing $\operatorname{Var}\left(L_i\right)=1$ for identification, and using adaptive likelihood integration with cluster-robust SEs.\footnote{We estimate the variance $\sigma^2_X$ from the first equation and the logit equation has a fixed error variance.} We then form $\hat{P}_i=\operatorname{logit}^{-1}\left(\hat{\eta}_i\right)$ and estimate $\Psi_z$ by the sample mean of $X_i \sum_{t=0}^T\left(\delta \hat{p}_i\right)^t$ within treatment arm $z$. 
To compute the difference across arms, we regress $X_i \sum_{t=0}^T\left(\delta \hat{P}_i\right)^t$ on arm indicators.

Our second approach bounds $\Psi_z$ by fixing the marginals of $\left(X_i, P_i\right)$ and maximizing /minimizing $\mathbb{E}\left[X_i \sum_{t=0}^T\left(\delta P_i\right)^t \mid Z_i=z\right]$ over all possible dependence structures. The comonotone coupling (sorting $X$ and $P$ in the same order) provides the sharp upper bound. The  counter-monotone coupling (sorting $X$ and $P$ in opposite orders) provides the sharp lower bound. We take the empirical distribution of $X_i$ in each arm and specify a parametric logit-normal margin for $P_i$ with dispersion $\sigma$ calibrated to our data. Specifically, we set  $\sigma=0.15$, which equals three times the implied standard deviation from our joint structural model from above.

\begin{figure}[H]
\caption{Bounding lender's profit.}
    \begin{center}
   \begin{subfigure}{0.49\textwidth}
        \centering
        \includegraphics[width=\textwidth]{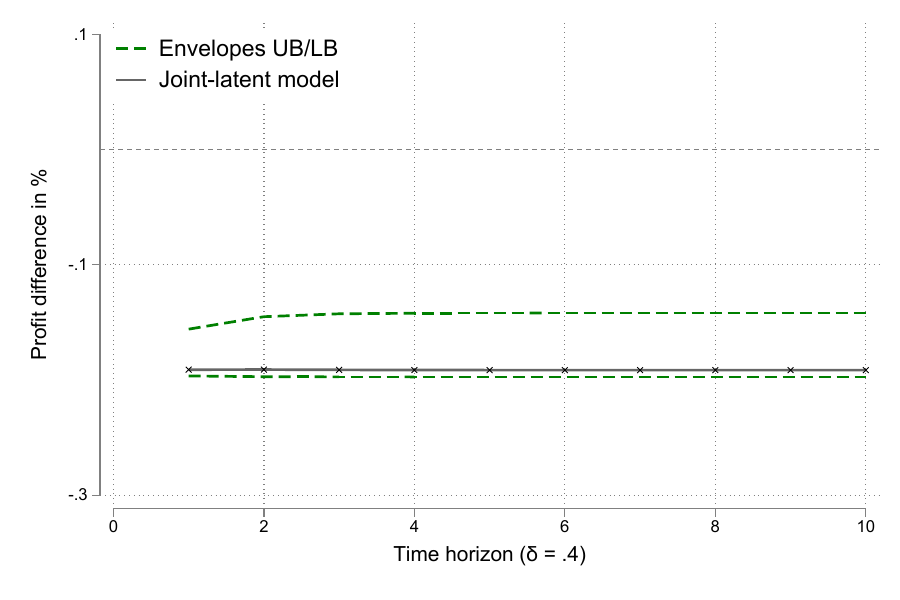}
    \end{subfigure} 
   \begin{subfigure}{0.49\textwidth}
        \centering
        \includegraphics[width=\textwidth]{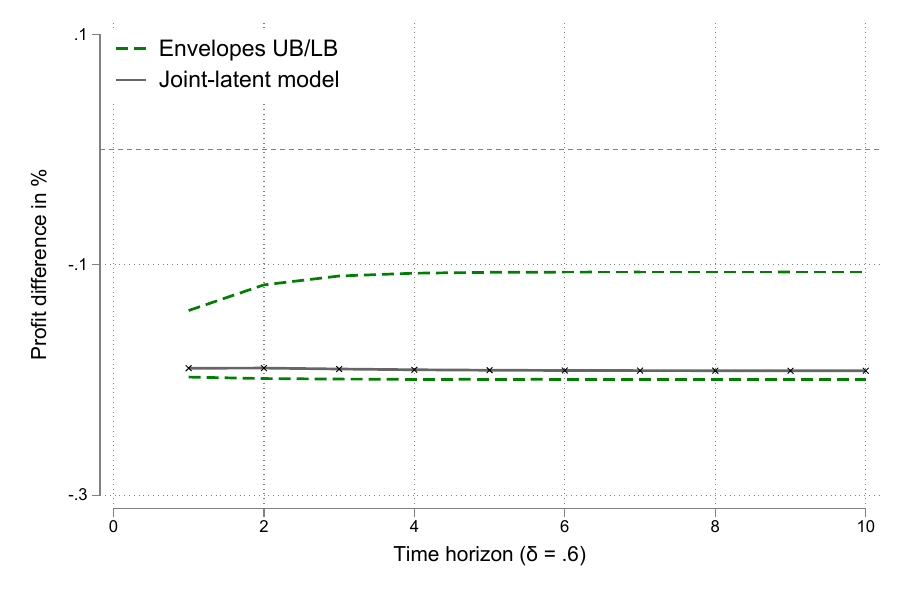}
    \end{subfigure} 
       \begin{subfigure}{0.49\textwidth}
        \centering
        \includegraphics[width=\textwidth]{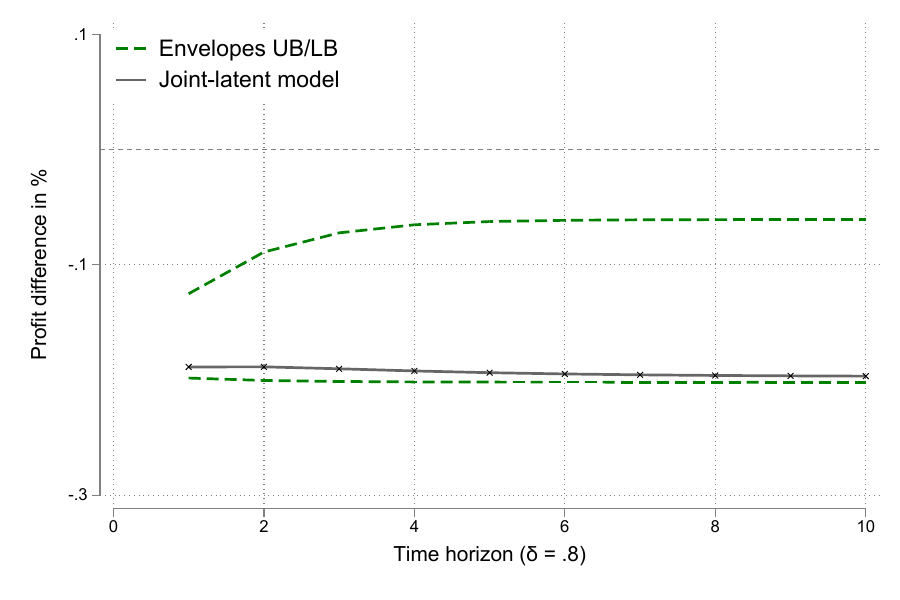}
    \end{subfigure} 
   \begin{subfigure}{0.49\textwidth}
        \centering
        \includegraphics[width=\textwidth]{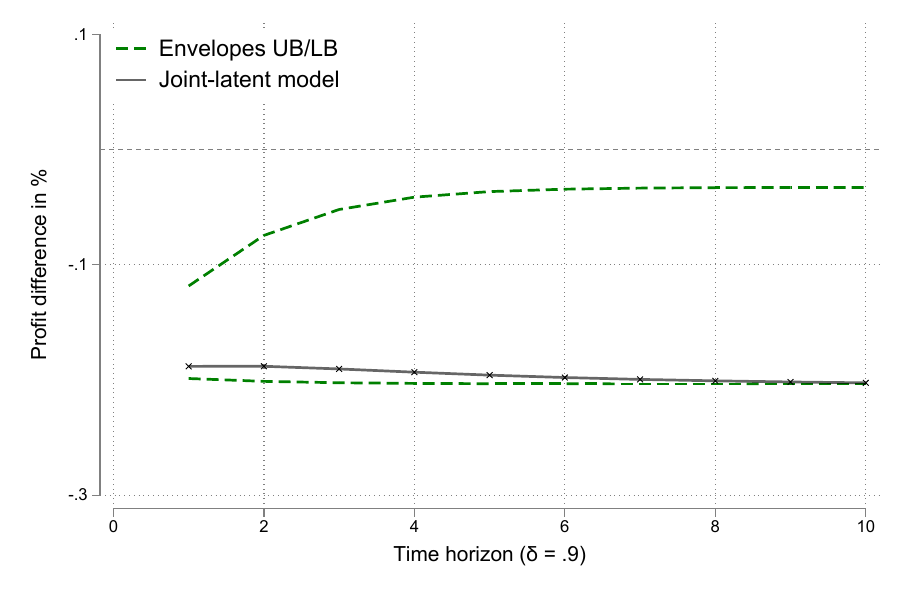}
    \end{subfigure}     
    \end{center}
      \footnotesize{}
      \label{bound_lender_profit}
\end{figure}

Figure \ref{bound_lender_profit} plots the structure versus status quo difference in $\Psi_z$ across discount factors $\delta$ and time horizons $T$ for both exercises described above: the joint latent model and the bounding approach. Regardless of which approach we use or which discount rate or time horizon we assume, lender profits remain lower under the structured contract.

\subsection{Personal promise results}

\begin{table}[H]
\caption{Effects of Promise on Financial Cost}
\label{sec_impact_table}
\begin{center}
\resizebox{1.0\textwidth}{!}{
\footnotesize{
\begin{tabular}{lcccccccc}
\toprule
      &       & \multicolumn{5}{c}{Components of FC}  &       &  \\
\cmidrule{3-7}      & FC    & Interest pymnt & Fee pymnt & Def$\times$Ppl pymnt & Lost pawn value & Default &       & CR \\
\cmidrule{2-2}\cmidrule{9-9}      &       & $\sum_t P^I_{it}$ & $\sum_t P^F_{it}$ & $\mathds{1}(\text{Def}_i)\times\sum_t P^C_{it}$ & $\mathds{1}(\text{Def}_i)\times \text{Value-Loan}_i$ & $\mathds{1}(\text{Def}_i)$ &       &  \\
\midrule
      & (1)   & (2)   & (3)   & (4)   & (5)   & (6)   &       & (7) \\
\midrule
\midrule
Promise Mandate & -18.0 & -20.3 & -     & 3.45  & 2.31  & 0.0086 &       & -0.0074 \\
      & (63.2) & (55.4) &       & (4.48) & (38.8) & (0.026) &       & (0.015) \\
Promise Choice & -87.2 & -69.9 & -     & 2.76  & -17.3 & 0.00033 &       & -0.010 \\
      & (54.1) & (45.5) &       & (3.71) & (35.2) & (0.024) &       & (0.012) \\
      &       &       &       &       &       &       &       &  \\
\midrule
Observations & 5144  & 5144  &       & 5144  & 5144  & 5144  &       & 5144 \\
R-squared & 0.001 & 0.001 &       & 0.000 & 0.000 & 0.000 &       & 0.000 \\
Control Mean & 942.4 & 545.9 &       & 5.96  & 396.5 & 0.44  &       & 0.43 \\
\bottomrule
\bottomrule
\end{tabular}%
}
}
\end{center}

\end{table}

\subsection{Repawning}

We now estimate equation \ref{basic_reg} with different dependent variables.  In column 1, the dependent variable indicates, for each borrower in the experiment, whether he/she pawned again after the first loan in the
experiment (up to the end period of our data set 338 days after the experiment began). The result is that the likelihood of repeat business increases by 6.7\%.  While this appears to be \emph{prima facie} evidence of greater satisfaction among borrowers in the mandatory arm, the interpretation is complicated by the fact that monthly payments may themselves trigger more borrowing to pay them. Note that from the lender's perspective, why repeat pawning happens may not be as important as the fact that it happens. Illiquidity does not seem to drive this result, given that the effect on re-pawning comes after 90 days (during the period of contract demanded payments) and not before  (see columns 2 and 3). Column 4 only considers new loans that use different collateral from that of the initial one. We do this to foreclose the explanation that those in the Mandatory arm, being more liquidity-constrained, return to pawn a second pawn to be able to pay the monthly payments of their first loan. We cannot reject a zero effect on pawning a different collateral. Column 5 focuses on the (endogenous) subsample of those recovering their pawn in both arms of the experiment. This means that \textit{both} arms have recovered their pawn and could re-pawn if they so wish, and also that the liquidity demands from the monthly contract are \textit{no longer} there as the contract has been closed.  We find that the difference between the mandatory structure contract and the status quo is even larger in this subsample, with the former having 11pp higher likelihood of being a repeat client during our sample period.

\begin{table}[H]
\caption{Effects on Repeat Pawning}
\label{repeat_loans}
\begin{center}
\footnotesize{
\begin{tabular}{lccccc}
\toprule
      & \multicolumn{5}{c}{Ever pawns again (ITT)} \\
\cmidrule{2-6}      &       & After 90 days & Within 90 days & Different collateral & Cond. on rec \\
\midrule
\midrule
      & (1)   & (2)   & (3)   & (4)   & (5) \\
\midrule
\midrule
Mandatory structured & 0.067* & 0.037*** & 0.032 & 0.044 & 0.11** \\
      & (0.035) & (0.013) & (0.027) & (0.030) & (0.047) \\
Choice  & 0.040 & 0.0098 & 0.030 & 0.038 & 0.057 \\
      & (0.031) & (0.0087) & (0.026) & (0.028) & (0.042) \\
      &       &       &       &       &  \\
\midrule
Observations & 4441  & 4441  & 4441  & 4441  & 2170 \\
R-squared & 0.003 & 0.006 & 0.001 & 0.002 & 0.008 \\
Control Mean & 0.32  & 0.020 & 0.30  & 0.30  & 0.35 \\
\bottomrule
\bottomrule
\end{tabular}%
}
\end{center}
 \footnotesize{This table estimates the specification of equation \ref{basic_reg} but at the level of the borrower (not the loan). Each column represents a regression with a different outcome variable. In column 1, the dependent variable indicates, for each borrower in the experiment, whether he/she pawned again after the first loan in the experiment (up to the end period of our data set 338 days after the experiment began). Column 2 is analogous, but only pawning after 90 days of the first loan is considered. Column 3 instead considers pawning before 90 days. Column 4 is analogous to column 1 but focuses on the pawning of a gold piece that is different from the one in the first experimental loan. Column 5 is analogous to column 1, but conditioning on the sample that recovered the first loan. 
Standard errors are clustered at the branch-day level.}
\end{table}

\begin{table}[H]
\caption{Baseline survey questions (translated to English)}
\label{baseline_survey}
\begin{center}
\scriptsize{
\begin{tabular}{cl}
\toprule
      & \textbf{Baseline Survey} \\
\midrule
\midrule
1     & \textbf{Your pawn was:} \\
      & (a) Inherited, (b) a gift, (c) bought by me, (d) lent to me, (e) other \_\_\_\_\_\_\_\_\_\_\_\_ \\
\rowcolor[rgb]{ .682,  .667,  .667} 2     & \textbf{Mark with an "X" on the line below how likely it is that you will recover your pawn. } \\
      & \textbf{Where 0 is impossible and 100 is completely certain} \\
\rowcolor[rgb]{ .682,  .667,  .667} 3     & \textbf{How much do you think the item you plan to pawn is worth?       \_\_\_\_\_\_\_\_\_\_\_\_\_\_ pesos} \\
\rowcolor[rgb]{ .682,  .667,  .667} 4     & \textbf{Gender      } \\
\rowcolor[rgb]{ .682,  .667,  .667} 5     & \textbf{Age} \\
6     & \textbf{Civil Status } \\
      & (a) married, (b) single, (c) divorced, (d) widowed \\
7     & \textbf{Work status} \\
      & (a) employed, (b) own business, (c) homemaker, (d) don't work, (e) retired, (f) study \\
\rowcolor[rgb]{ .682,  .667,  .667} 8     & \textbf{Education} \\
      & (a) no formal education, (b) primary, (c) middle school, (d) high school, (e) more than high school \\
\rowcolor[rgb]{ .682,  .667,  .667} 9     & \textbf{In the last month, did a friend or family member asked you for money?} \\
      & (a) yes  (b) no \\
\rowcolor[rgb]{ .682,  .667,  .667} 10    & \textbf{What would you like to have: 100 pesos tomorrow or 150 pesos in one month?} \\
\rowcolor[rgb]{ .682,  .667,  .667} 11    & \textbf{How often do you feel stressed by your economic situation?} \\
      & (a) always, (b) very often, (c) sometimes, (d) never \\
12    & \textbf{What is the main reason you want to pawn?} \\
      & (a) Need the money because somebody in my family lost his/her job \\
      & (b) Need the money to pay for a sickness in the family \\
      & (c) Need the money for an urgent expense \\
      & (d) Need the money for some non urgent expense. \\
\rowcolor[rgb]{ .682,  .667,  .667} 13    & \textbf{How stressed do you feel from the situation that led to pawn?} \\
      & (a) very stressed, (b) somewhat stressed, (c) a little stressed, (d) not stressed  \\
\rowcolor[rgb]{ .682,  .667,  .667} 14    & \textbf{In 3 months, I expect to have a  \_\_\_\_\_\_\_\_\_\_\_\_\_\_\_\_ situation} \\
      & (a) better, (b) similar,  (c) worse \\
\rowcolor[rgb]{ .682,  .667,  .667} 15    & \textbf{Have you pawned before?} \\
      & (a) yes  (b) no \\
\rowcolor[rgb]{ .682,  .667,  .667} 16    & \textbf{How many times have you pawned on a Lender P branch?} \\
      & (a) NO\_\_\_    (b)  1-2 times \_\_\_    (c) 3-5 times\_\_\_\_   (d) More than 5\_\_\_\_ \\
\rowcolor[rgb]{ .682,  .667,  .667} 17    & \multicolumn{1}{p{54.91em}}{\textbf{If you are saving money and a family member wants to use it for something }} \\
      & \multicolumn{1}{p{54.91em}}{(a) I would only give him the money for an urgent expense} \\
      & \multicolumn{1}{p{54.91em}}{(b) I would give him the money even if it was not an urgent expense} \\
      & \multicolumn{1}{p{54.91em}}{(c) I would not give him/her the money regardless} \\
      & \multicolumn{1}{p{54.91em}}{(d) No one would ask me for my money} \\
\rowcolor[rgb]{ .682,  .667,  .667} 18    & \textbf{Do you make an expenses budget for the month ahead of time?} \\
      & (a) always, (b) very often, (c) sometimes, (d) never \\
19    & \multicolumn{1}{p{54.91em}}{\textbf{Do you have other items you could pawn?}} \\
      & (a) yes  (b) no \\
\rowcolor[rgb]{ .682,  .667,  .667} 20    & \textbf{Do you have savings?} \\
      & (a) yes  (b) no \\
\rowcolor[rgb]{ .682,  .667,  .667} 21    & \textbf{Do you participate in a ROSCA?} \\
      & (a) yes  (b) no \\
\rowcolor[rgb]{ .682,  .667,  .667} 22    & \textbf{Is it common that family or friends ask for money?} \\
      & (a) yes  (b) no \\
\rowcolor[rgb]{ .682,  .667,  .667} 23    & \textbf{How much did you spend to come to the branch today?    \$\_\_\_\_\_\_\_\_\_\_\_\_\_\_ pesos} \\
\rowcolor[rgb]{ .682,  .667,  .667} 24    & \textbf{How much time does it usually take to come to this branch?    \_\_\_\_\_\_\_\_\_\_\_} \\
25    & \textbf{How much does your family spend in a normal week?   \$\_\_\_\_\_\_\_\_\_\_\_\_\_\_ pesos} \\
26    & \textbf{How much do you manage to save in a normal week?   \$\_\_\_\_\_\_\_\_\_\_\_\_\_\_ pesos} \\
\rowcolor[rgb]{ .682,  .667,  .667} 27    & \textbf{Does it happen to you that you spend more than you wanted because you fall into temptation?} \\
      & (a) never, (b) almost never, (c) sometimes, (d) very often \\
\rowcolor[rgb]{ .682,  .667,  .667} 28    & \textbf{In the last 6 months, has it happened that at some point you lacked money to pay} \\
      & (a) rent?    (b) food    (c)food   (d) medicine  (e) electricity   (f) heating   (g) telephone    (i) water \\
\rowcolor[rgb]{ .682,  .667,  .667} 29    & \textbf{What would you like to have: 100 pesos in 3 months or 150 pesos in four months?} \\
\rowcolor[rgb]{ .682,  .667,  .667} 30    & \textbf{Would you like to receive (free) reminders for upcoming payments?} \\
      & (a) yes  (b) no \\
\bottomrule
\end{tabular}%
}
\end{center}
\scriptsize{
Translation of the baseline questionnaire. Highlighted questions were used in our causal forest exercises.}
\end{table}

\section{Identification Arguments}
\label{sec:proofs}
This section gives a formal derivation of the identification results presented in Equations \eqref{eq:TOT}--\eqref{eq:ASL} of Section \ref{sec:randchoice}.
To simplify the presentation, we omit $i$ subscripts throughout and present derivations that do not explicitly condition on pre-treatment covariates. 
Because $Z_i$ is assigned independently of pre-treatment covariates, however, all of the derivations presented below also hold conditional on $X$ under Assumptions \ref{assump:design}--\ref{assump:exclusion} from the body of the paper.\footnote{Since our exclusion restriction is an individual-level assumption $Y_i(d,z) = Y_i(d)$ for all $i$, it automatically holds for any subgroup defined by $X$.}
For convenience we work with the following implication of these three assumptions. 

\begin{assumption}[Randomized Choice Design and Exclusion Restriction]\mbox{}
\label{assump:randchoice}
   \begin{enumerate}[(i)]
   \item $Z$ is independent of $(Y_{0}, Y_{1}, C)$
   \item $D = \mathbbm{1}(Z \neq 2)Z + \mathbbm{1}(Z = 2)C$
   \item $Y = \mathbbm{1}(Z = 0) Y_{0} + \mathbbm{1}(Z = 1) Y_{1} + \mathbbm{1}(Z = 2) [ (1 - C) Y_{0} + C Y_{1}]$
   \end{enumerate}
\end{assumption}

\begin{lem}
Under Assumption \ref{assump:randchoice},
\label{lem_randchoice}
   \begin{enumerate}[(i)]
       \item $\mathbbm{E}(D|Z=2) = \mathbb{P}(C=1)$
       \item $\mathbbm{E}(Y|Z=0) = \mathbbm{E}(Y_0)$
       \item $\mathbbm{E}(Y|Z=1) = \mathbbm{E}(Y_1)$
       \item $\mathbbm{E}(Y|D=0,Z=2) = \mathbbm{E}(Y_0|C=0)$
       \item $\mathbbm{E}(Y|D=1,Z=2) = \mathbbm{E}(Y_1|C=1)$.
   \end{enumerate} 
\end{lem}

\begin{proof}
Part (i) follows because $Z=2$ implies $D=C$ and $Z$ is independent of $C$.  
Parts (ii) and (iii) follow similarly: given $Z=0$ we have $Y = Y_0$, given $Z=1$ we have $Y = Y_1$, and $Z$ is independent of $(Y_0,Y_1)$.
For parts (iv) and (v), first note that Assumption \ref{assump:randchoice} (i) implies that $Z$ is conditionally independent of $(Y_0,Y_1)$ given $C$.
Now, $Z=2$ implies that $D=0$ if and only if $C=0$. Hence, $\mathbbm{E}(Y|D=0, Z=2) = \mathbbm{E}(Y_0|C=0)$ establishing part (iv).
For part (v) $Z=2$ implies that $D=1$ if and only if $C=1$ from which it follows that $\mathbbm{E}(Y|D=1, Z=2) = \mathbbm{E}(Y_1|C=1)$.
\end{proof}

\begin{prop} 
\label{pro:randchoice}
Under Assumption \ref{assump:randchoice},
    \begin{enumerate}[(i)]
        \item $\text{TOT} \equiv \mathbbm{E}(Y_1 - Y_0|C=1)  = \displaystyle \frac{\mathbbm{E}(Y|Z=2) - \mathbbm{E}(Y|Z=0)}{\mathbbm{E}(D|Z=2)}$
        \item $\text{TUT} \equiv \mathbbm{E}(Y_1 - Y_0|C=0) = \displaystyle \frac{\mathbbm{E}(Y|Z=1) - \mathbbm{E}(Y|Z=2)}{1 - \mathbbm{E}(D|Z=2)}$
        \item $\text{ASB} \equiv \mathbbm{E}(Y_0|C=1) - \mathbbm{E}(Y_0|C=0) = \displaystyle \frac{\mathbbm{E}(Y|Z=0) - \mathbbm{E}(Y|Z=2,D=0)}{\mathbbm{E}(D|Z=2)}$
        \item $\text{ASL} \equiv \mathbbm{E}(Y_1|C=1) - \mathbbm{E}(Y_1|C=0) = \displaystyle \frac{\mathbbm{E}(Y|Z=2,D=1) - \mathbbm{E}(Y|Z=1)}{1 - \mathbbm{E}(D|Z=2)}$.
    \end{enumerate}
\end{prop}

\begin{proof}
Parts (i) and (iii) we require an expression for $\mathbbm{E}(Y_0|C=1)$ in terms of $(Y, D, Z)$.
By Lemma \ref{lem_randchoice}(ii) and iterated expectations 
\[
\mathbbm{E}(Y|Z=0) = \mathbbm{E}(Y_0) = \mathbbm{E}(Y_0|C=0) \mathbbm{P}(C=0) + \mathbbm{E}(Y_0|C=1) \mathbbm{P}(C=1).
\]
Re-arranging and substituting Lemma \ref{lem_randchoice}(i) and (iv),
\begin{align}
\mathbbm{E}(Y_0|C=1)  
&=  \frac{\mathbbm{E}(Y|Z=0) - \mathbbm{E}(Y|Z=2,D=0) \mathbbm{E}(1 - D|Z=2)}{\mathbbm{E}(D|Z=2)}.
\label{eq:Y0C1}
\end{align}
Part (i) follows by combining \eqref{eq:Y0C1} with Lemma \ref{lem_randchoice}(v) and simplifying; part (iii) follows by combining \eqref{eq:Y0C1} with Lemma \ref{lem_randchoice}(iv) and simplifying.
Similarly, for parts (ii) and (iv) we require an expression for $\mathbbm{E}(Y_1|C=0)$ in terms of observables.
By Lemma \ref{lem_randchoice}(iii) and iterated expectations,
\[
\mathbbm{E}(Y|Z=1) = \mathbbm{E}(Y_1) = \mathbbm{E}(Y_1|C=0)\mathbbm{P}(C=0) + \mathbbm{E}(Y_1|C=1) \mathbbm{P}(C=1).
\]
Re-arranging and substituting Lemma \ref{lem_randchoice}(i) and (v),
\begin{align}
\mathbbm{E}(Y_1|C=0) 
&=\frac{\mathbbm{E}(Y|Z=1) - \mathbbm{E}(Y|Z=2,D=1)\mathbbm{E}(D|Z=2)}{\mathbb{E}(1-D|Z=2)}.
\label{eq:Y1C0}
\end{align}
Part (ii) follows by combining \eqref{eq:Y1C0} with Lemma \ref{lem_randchoice}(iv) and simplifying; part (iv) follows by combining \eqref{eq:Y1C0} with Lemma \ref{lem_randchoice}(v) and simplifying.
\end{proof}

\section{Further Details on Fan \& Park (2010) Bounds}
\label{append:bounds}

In this appendix we provide details on estimation and inference for the bounds on the distribution of treatment effects from Section \ref{sec:bounds}. 
\footnote{We provide a STATA package implementing these methods, which can be installed as follows: \texttt{net install fan\_park, from(https://raw.githubusercontent.com/isaacmeza/fan\_park/main) replace}}
As shown in \cite{fan2010sharp}, for any fixed $\delta$ the sharp bounds for $F_\Delta(\delta)$ are $[\underline{F}(\delta), \,\overline{F}(\delta)]$ where 
\begin{equation*}
\underline{F}(\delta) \equiv \max \{0, \sup_y F_1(y) - F_0(y - \delta)  \}, \, \,
\overline{F}(\delta) \equiv 1 + \min \{0, \inf_y F_1(y) - F_0(y-\delta) \}.
\end{equation*}
These bounds can be improved by accounting for covariates. 
Define $\underline{F}(\delta|X=x)$ and $\overline{F}(\delta|X=x)$ by replacing $F_0$ and $F_1$ with $F_0(\cdot|X=x)$ and $F_1(\cdot|X=x)$ in the definitions from above.
Then the sharp bounds for $F_\Delta(\delta)$ become $\left[ \mathbbm{E}\left\{\underline{F}(\delta|X)\right\},\, \mathbbm{E}\left\{\overline{F}(\delta|X)\right\}\right]$. 

In Section \ref{sec:bounds} and Figure \ref{fan_park_bounds} we adjust for day of week and branch dummies to tighten the bounds. Because these covariates are discrete, inference is relatively straightforward.
Let $\left\{x_k\right\}_{k=1}^K$ be the support set of the discrete covariate vector $X$, where $K$ is finite.  Define $p_k=\mathbbm{P}\left(X=x_k\right)$.
Write $\pi_{d k}=\mathbbm{P}\left(D=d \mid X=x_k\right)$. 
As above, let $F_d\left(\cdot \mid x_k\right)$ be the conditional CDFs for the potential outcomes. 
For a fixed $\delta$, define
\begin{align*}
\underline{F}\left(\delta \mid X=x_k\right)&\equiv \sup _y\left\{F_1\left(y \mid x_k\right)-F_0\left(y-\delta \mid x_k\right)\right\}_{+}\\
\quad y_{L k}(\delta) &\equiv \arg \max _y\left\{F_1\left(y \mid x_k\right)-F_0\left(y-\delta \mid x_k\right)\right\}.
\end{align*}
The unconditional lower bound with covariates is $\underline{F}(\delta)=\mathbb{E}[\underline{F}(\delta \mid X)]=\sum_{k=1}^K p_k \underline{F}\left(\delta \mid x_k\right)$.
Under standard regularity conditions, we have $\sqrt{n}(\underline{\widehat{F}}(\delta)-\underline{F}(\delta)) \stackrel{d}{\Longrightarrow} \mathcal{N}\left(0, V_L(\delta)\right)$ for each fixed $\delta$, where the asymptotic variance is given by
\begin{align*}
V_L(\delta)= &\operatorname{Var}(\underline{F}(\delta \mid X)) 
 +\mathbbm{E}\left[\frac{F_1\left(y_L(X, \delta) \mid X\right)\left(1-F_1\left(y_L(X, \delta) \mid X\right)\right)}{\pi_1(X)}\right]\\&+\mathbbm{E}\left[\frac{F_0\left(y_L(X, \delta)-\delta \mid X\right)\left(1-F_0\left(\left.y_L\left(X, \delta\right)-\delta \right\rvert\, X\right)\right)}{\pi_0(X)}\right] 
\end{align*}
and corresponding plug-in estimator
\begin{align*}
  \widehat{V}_L(\delta)&=\underbrace{\sum_k \hat{p}_k \underline{\widehat{F}}\left(\delta \mid x_k\right)^2-\left(\sum_p \hat{p}_k \underline{\widehat{F}}\left(\delta \mid x_k\right)\right)^2}_{\widehat{\operatorname{Var}}[{\underline{F}}(\delta \mid X)]} \\
   &+\sum_k \hat{p}_k\left[\frac{\widehat{F}_1\left(\widehat{y}_{L k} \mid x_k\right)\left(1-\widehat{F}_1\left(\widehat{y}_{L k} \mid x_k\right)\right)}{\widehat{\pi}_{1 k}}+\frac{\widehat{F}_0\left(\widehat{y}_{L k}-\delta \mid x_k\right)\left(1-\widehat{F}_0\left(\widehat{y}_{L k}-\delta \mid x_k\right)\right)}{\widehat{\pi}_{0 k}}\right]  
\end{align*}
For the upper bound replace $y_L$ with $y_U(x, \delta) \equiv \arg \min \left\{F_1(y \mid x)-F_0(y-\delta \mid x)\right\}$.

\begin{figure}[H]
   \caption{Fan \& Park bounds for benefit in CR\%.}
    \begin{center}
        \centering
        \includegraphics[width=0.65\textwidth]{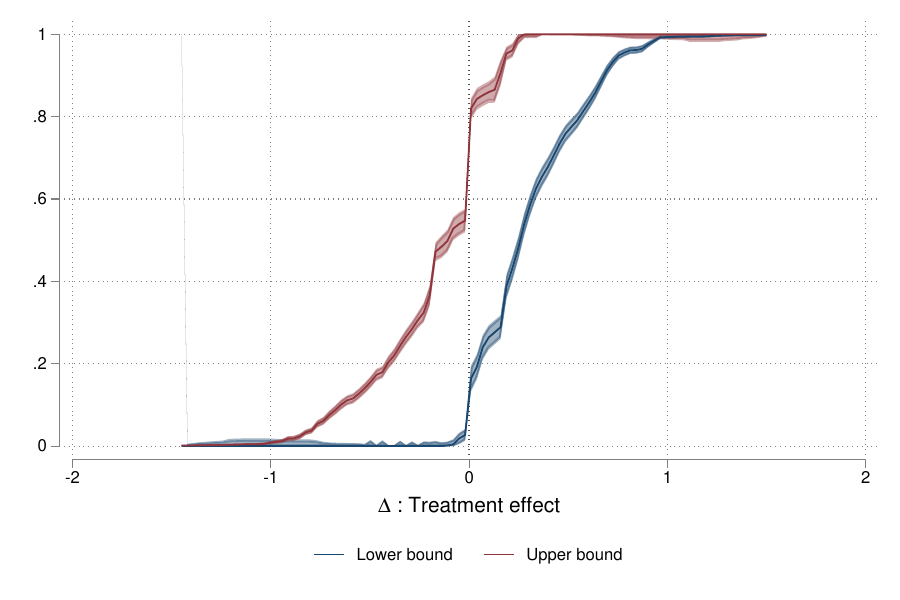}
   \end{center}
          \footnotesize{This figure depicts the \cite{fan2010sharp} bounds on the distribution $F_\Delta$ of individual treatment effects $\Delta \equiv (Y_1 - Y_0)$, described in Section \ref{sec:bounds}, for the CR outcome.
    The dark red curve and light red shaded region give the estimated upper bound function $\overline{F}$ for $F_\Delta$ and associated (pointwise) 95\% confidence interval. 
    The dark blue curve and light blue shaded region give the estimated lower bound function $\underline{F}$ for $F_\Delta$ and associated (pointwise) 95\% confidence interval.
    Confidence intervals are computed using the asymptotic distribution for the bounds.  Evaluating the bounds at $\delta = 0$, we see that between 28\% and 96\% of borrowers have a positive individual treatment effect.}
     \label{fan_park_bounds}
\end{figure}

\begin{figure}[H]
\caption{Distribution of treatment effects for CR benefit under rank invariance.}         
\begin{center}
    \centering
    \includegraphics[width=0.6\textwidth]{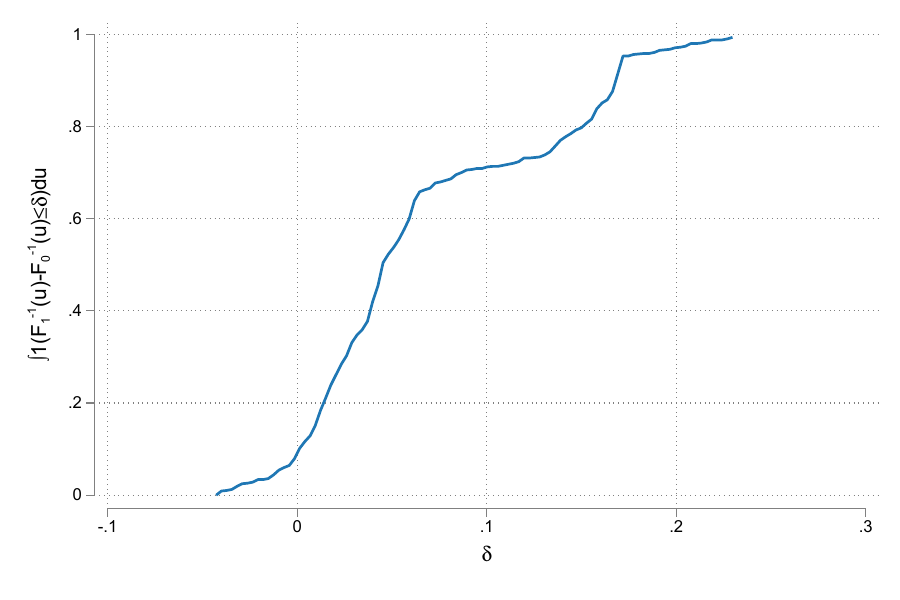}
 \end{center}
    \footnotesize{This figure shows the CDF of individual treatment effects under the assumption of rank invariance, computed from $F_\Delta(\delta) = \int_0^1 \mathbbm{1}\{ F_1^{-1}(u) - F_0^{-1}(u)\leq \delta\}\,\mathrm{d}u$ where $F_1^{-1}$ and $F_0^{-1}$ are the quantile functions of $Y_1$ and $Y_0$.}
\label{te_rankinvariance}
\end{figure}

\clearpage

\subsection{Baseline Results for Survey Respondents}
\label{append:respondents}

\begin{table}[H]
\caption{Effects on Financial Cost by selected subsample of survey respondents}
\label{tab:main_effects_respondents}
\begin{center}
\resizebox{1.0\textwidth}{!}{
\footnotesize{
\begin{tabular}{lcccccccc}
\toprule
      &       & \multicolumn{5}{c}{Components of FC}  &       &  \\
\cmidrule{3-7}      & FC    & Interest pymnt & Fee pymnt & Def$\times$Ppl pymnt & Lost pawn value & Default &       & CR \\
\cmidrule{2-2}\cmidrule{9-9}      &       & $\sum_t P^I_{it}$ & $\sum_t P^F_{it}$ & $\mathds{1}(\text{Def}_i)\times\sum_t P^C_{it}$ & $\mathds{1}(\text{Def}_i)\times \text{Value-Loan}_i$ & $\mathds{1}(\text{Def}_i)$ &       &  \\
\midrule
      & (1)   & (2)   & (3)   & (4)   & (5)   & (6)   &       & (7) \\
\midrule
\midrule
Mandatory structured & -163.3*** & -119.6** & 30.6*** & -0.38 & -74.3** & -0.078*** &       & -0.067*** \\
      & (59.4) & (46.7) & (1.74) & (3.58) & (36.6) & (0.027) &       & (0.012) \\
Choice  & -4.88 & 10.7  & 1.67*** & 0.25  & -17.2 & -0.047* &       & -0.0046 \\
      & (64.7) & (56.6) & (0.31) & (3.32) & (35.6) & (0.025) &       & (0.014) \\
      &       &       &       &       &       &       &       &  \\
\midrule
Observations & 5037  & 5037  & 5037  & 5037  & 5037  & 5037  &       & 5037 \\
R-squared & 0.004 & 0.004 & 0.134 & 0.000 & 0.002 & 0.004 &       & 0.018 \\
Control Mean & 916.1 & 526.4 & 0     & 4.75  & 389.7 & 0.44  &       & 0.42 \\
CATE subsample & \checkmark & \checkmark & \checkmark & \checkmark & \checkmark & \checkmark &       & \checkmark \\
\midrule
\midrule
      &       &       &       &       &       &       &       &  \\
\midrule
      & (8)   & (9)   & (10)  & (11)  & (12)  & (13)  &       & (14) \\
\midrule
\midrule
Mandatory structured & -147.2*** & -107.3** & 31.2*** & -0.33 & -71.1* & -0.076*** &       & -0.064*** \\
      & (56.5) & (44.4) & (1.80) & (3.71) & (37.3) & (0.027) &       & (0.013) \\
Choice  & -7.37 & 9.89  & 1.74*** & 0.14  & -19.0 & -0.051** &       & -0.0069 \\
      & (60.7) & (53.6) & (0.32) & (3.43) & (35.8) & (0.025) &       & (0.013) \\
      &       &       &       &       &       &       &       &  \\
\midrule
Observations & 4837  & 4837  & 4837  & 4837  & 4837  & 4837  &       & 4837 \\
R-squared & 0.003 & 0.004 & 0.136 & 0.000 & 0.001 & 0.004 &       & 0.016 \\
Control Mean & 911.3 & 519.8 & 0     & 4.90  & 391.6 & 0.45  &       & 0.42 \\
Survey subsample & \checkmark & \checkmark & \checkmark & \checkmark & \checkmark & \checkmark &       & \checkmark \\
\bottomrule
\bottomrule
\end{tabular}%
}
}
\end{center}
 \footnotesize{Baseline experimental results for the subset of borrowers with survey variables used in the estimation of CATEs (Top panel) and respondents of any survey question (Bottom panel).}
\end{table}

\section{Causal Forest Details} \label{app:cate}

\subsection{Estimation}

To estimate the conditional average treatment effects shown in Figure \ref{heterogeneous_effects} from our survey and administrative data, we use the \texttt{grf} R package. 
This package implements special cases of the ``generalized random forest'' approach of \cite{atheygrf}.
In broad strokes, they combine a large number of regression trees that recursively partition the covariate space to estimate conditional average effects.
The trees are ``honest'' in that observations used to determine the optimal partition are not used to estimate effects, and vice-versa.
While closely related to more familiar ``regression-tree'' random forests, the generalized random forest approach explicitly targets the parameter of interest---a conditional ATE or TUT/TOT estimand---when choosing the optimal covariate partition.

We use the \texttt{causal\_forest()} function to estimate CATE effects and the related \texttt{instrumental\_forest()} function to estimate CTUT and CTOT effects. 
In each case, we use the default parameter values from the \texttt{grf} package with one exception: we increase the number of trees from the default value of 2000 to 5000.\footnote{For more detail on implementation, see \cite{atheygrf} and the \texttt{grf} documentation: \url{https://grf-labs.github.io/grf/}.}
When estimating conditional average treatment effects based on our experimental data, we use observations for all borrowers who answered at least \emph{part} of the intake survey.
The survey measures used in this analysis are indicated in Table \ref{baseline_survey}.
We impute the median response for the missing values, while also including an indicator whether the variable originally had a missing value. Results are similar if we manually include interactions between the original/imputed variable and an indicator for missingness. This is as expected, as that tree-based methods automatically consider interactions of arbitrary orders.

\subsection{Inference}

We use a Bayesian bootstrap approach to carry out inference summary measures constructed from our heterogeneous treatment effect estimates in Figure \ref{choose_wrong}.\footnote{See \cite{hardle2025uniform} for uniform inference results in this setting.}
For simplicity, we describe our inference procedure in terms of the CATE, rather than separately for all of the conditional average effects we estimate, but the logic is identical in each case. 
Similarly, we focus on inference for the share with a negative conditional average effect, both because this is our main question of interest and because the procedure is identical for other thresholds: simply replace $0$ with $\delta$ in the expressions below.

As above, let $\tau(x) \equiv \mathbbm{E}[Y_1 - Y_0 \mid X = x]$ be the CATE when $X = x$. Since $X$ is a random variable in the population, $\tau(X)$ likewise has a population distribution. Denote the CDF of this distribution by $F_\tau$ and the density by $f_\tau$.
Our goal is to carry out inference for the fraction of individuals in the population who have a negative CATE: 
\[
  p \equiv F_\tau(0) = \mathbbm{P}\big( \tau(X) \leq 0 \big).
\]
Using our causal forest estimator $\widehat{\tau}(\cdot)$ of $\tau(\cdot)$ described above, we construct the following plug-in estimator of $p$
\[
  \widehat p \;=\; \frac{1}{n} \sum_{i=1}^n \mathbbm{1}\!\big\{\widehat\tau_{-i}(X_i) \le 0\big\},
\]
where the subscript $-i$ indicates that, when evaluating the estimated CATE function at $X_i$, we average only over those trees for which $X_i$ was not used in the estimation sample (out-of-bag). Our full bootstrap procedure is as follows:
\begin{enumerate}
\item \textbf{Draw cluster weights (Bayesian Bootstrap).} Let $g=1,\dots,G$ index clusters and $n_g$ be cluster sizes. Draw $(V_1,\ldots, V_G)\sim \operatorname{Dirichlet}(\alpha_1,\ldots, \alpha_G)$ with $\alpha_g=n_g$ for an \emph{individual-weighted} estimand. Set unit weights $w_i = \frac{V_{g(i)}}{n_{g(i)}}$ and normalize so that $\sum_i w_i= 1$.
\item \textbf{Estimation (keep honesty and clustering).} Fit the GRF on the original sample using the weights from Step 1 to obtain :
\[\hat{p}^{(b)}=\sum_{i=1}^n w_i 1\left\{\hat{\tau}_{(-g(i))}\left(X_i\right) \leq 0\right\}\]
\item \textbf{Repeat.} Repeat steps 1--2 for $b=1,\dots,B$ to obtain $\{\hat{p}^{(b)}\}_{b=1}^B$.
\item \textbf{Inference.} Form a confidence interval for $p$ from the empirical quantiles of $\{\hat{p}^{(b)}\}$ (\textbf{percentile Bayesian Bootstrap} CI).
\end{enumerate}

The Bayesian Bootstrap includes \emph{all} clusters every draw with randomly drawn weights, giving smoothing estimates of the target parameter across bootstrap samples, especially when $G$ is modest. These weights can be directly plugged into the \texttt{grf} estimation function via the option \texttt{sample.weights}. Our approach respects the clustered design while avoiding the instability that standard cluster resampling can introduce in tree splits and OOB predictions. While we use this approach to produce pointwise confidence intervals with valid Frequentist coverage, our inferences may also be given a Bayesian interpretation if desired: they are posterior probabilities under a nonparametric prior over the empirical support.

\section{Neoclassical and Behavioral Mechanisms} 

\subsection{Learning}
\label{append:learning}
Table \ref{learning_table} presents information about borrowers' \emph{future} pawning behavior as a function of treatment assignment. Column (1) considers the 228 clients who returned only a second time to pawn again at a day/branch that was randomly assigned to the choice arm. Each of the two rows in this column presents a difference of mean structure take-up rates, and associated standard error. The first row compares those who were \emph{initially} assigned to mandatory structure against those who where assigned to control; the second row compares those who were initially assigned to the choice arm to those who were assigned to the other two arms. In each case, there is no statistically discernible difference in the rates of structured payment take-up. Granted, this is a selected sample because the decision to pawn again is potentially endogenous to the initial treatment allocation. For this reason, Column (2) considers the full sample of 4441 borrowers by re-defining the outcome variable to be an indicator for returning to pawn again at a branch/day when structure was offered \emph{and} choosing structure. This composite outcome variable is not subject to the sample selection problem (although it is directly driven by the decision to borrow again). The comparison in the two rows remains the same: mandatory structure versus control in row one and choice versus the other two arms in row two. Again, there is no statistically discernible difference in structure take-up rates in either row. While these exercises cannot completely exclude the possibility that learning plays a role, they provide no indication that the lack of voluntary compliance is simply a matter of inexperience with structure.

\begin{table}[H]
        \caption{Effect of Prior Assignment on Subsequent Choice}
    \label{learning_table}
\begin{center}
\footnotesize{
\begin{tabular}{lcc}
\toprule
      & Choose structure in $t+1$ & Ever choose structure in $t+1$ \\
\cmidrule{2-3}$t$   & (1)   & (2) \\
\midrule
\midrule
Mandatory structured (ATE) & -0.0047 & 0.00014 \\
      & (0.048) & (0.0027) \\
Choice  (ITT) & 0.034 & 0.0015 \\
      & (0.057) & (0.0030) \\
      &       &  \\
\midrule
Observations & 228   & 4441 \\
R-sq  & 0.004 & 0.000 \\
DepVarMean & 0.092 & 0.0047 \\
\bottomrule
\bottomrule
\end{tabular}%
}
\end{center}
 \footnotesize{Column (1) reports results for the 228 borrowers who returned to pawn again at a day/branch that was randomly assigned to the choice arm, enabling us to observe whether they chose structure or the status quo contract.
Each row presents a difference in mean structure take-up rates and associated standard errors.
The first row (ATE) compares borrowers who were initially assigned to mandatory structure against those who were assigned to the control condition.
The second row (ITT) compares borrowers who were initially assigned to the choice  condition to those who were not.
Whereas column (1) conditions on the (endogenously) selected sample of borrowers who return to pawn again, column (2) considers the full sample by re-defining the ``outcome'' to be an indicator for whether a borrower pawned again on a day when choice was offered \emph{and} chose structure.}

\end{table}
\normalsize
\normalsize

\subsection{Discount rates}

\begin{figure}[H]
  \caption{Financial benefit TUT effect for different discount rates.}
    \begin{center}
        \centering
        \includegraphics[width=0.35\textwidth]{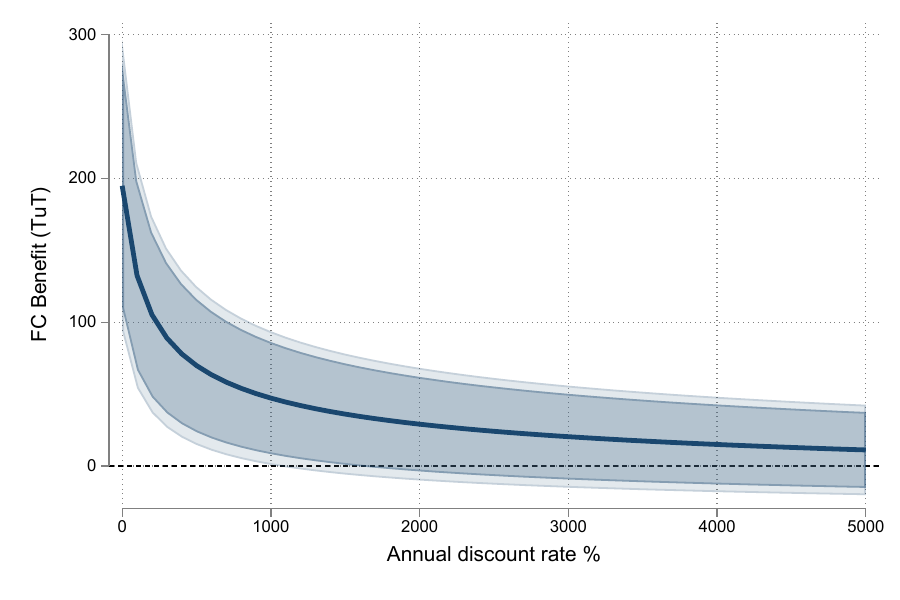}
    \end{center}
    \footnotesize{This figure re-estimates the treatment on the untreated (TUT) effect from Table \ref{tot_tut}, introducing a daily discount factor in the definition of financial benefit. At a given annual discount rate in percentage points (x-axis), the solid line gives the TUT effect with all costs and benefits converted to present value using that discount rate, and the shaded regions show 90\% and 95\% confidence bands. At a discount rate of zero, the estimate matches Table \ref{tot_tut}. As seen from the figure, borrowers would need to face implausibly high discount rates (above 1000\% annually) to reverse our headline result.} 
        \label{fc_discount_rates}
\end{figure}

\normalsize
\normalsize

\subsection{Sure Confidence} \label{App_sureconfidence}

\begin{figure}[H]
\caption{Determinants sure confidence.}
    \begin{center}
    \begin{subfigure}{0.6\textwidth}
        \centering
        \includegraphics[width=\textwidth]{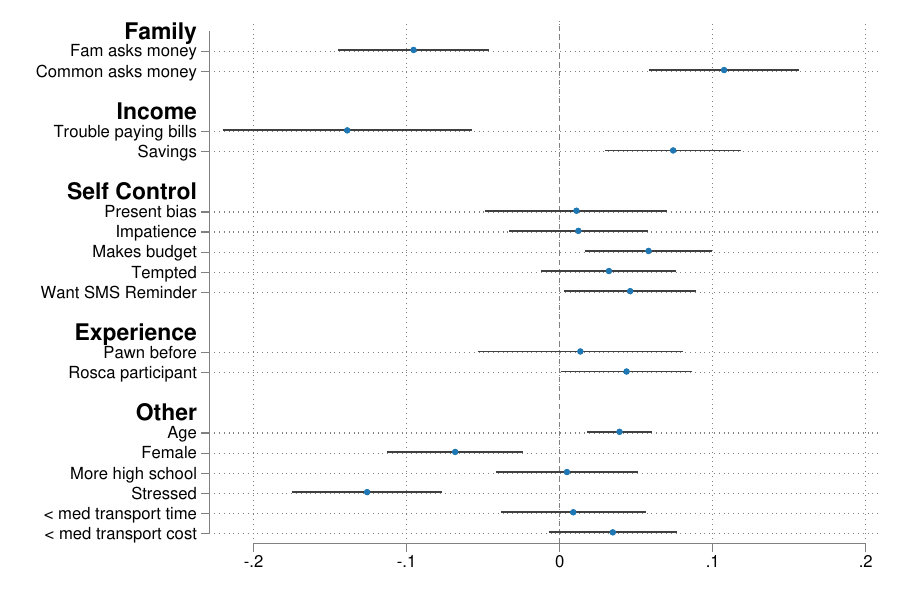}
    \end{subfigure}
    \end{center}
     \footnotesize{The above figure shows the coefficients in a multivariate OLS regression of sure confidence among borrowers in the status quo, mandatory structured, and choice arms. Sure confidence is a binary variable defined to be one when people report a 100\% probability of recovery.}
    \label{determinants_sure}
\end{figure}

\subsection{Counterfactual Outcomes for Non-Choosers}

\begin{figure}[H]
\caption{Estimating counterfactual outcomes for non-choosers, by overconfidence} 
    \begin{center}
       \begin{subfigure}{0.49\textwidth}
        \centering
        \includegraphics[width=\textwidth]{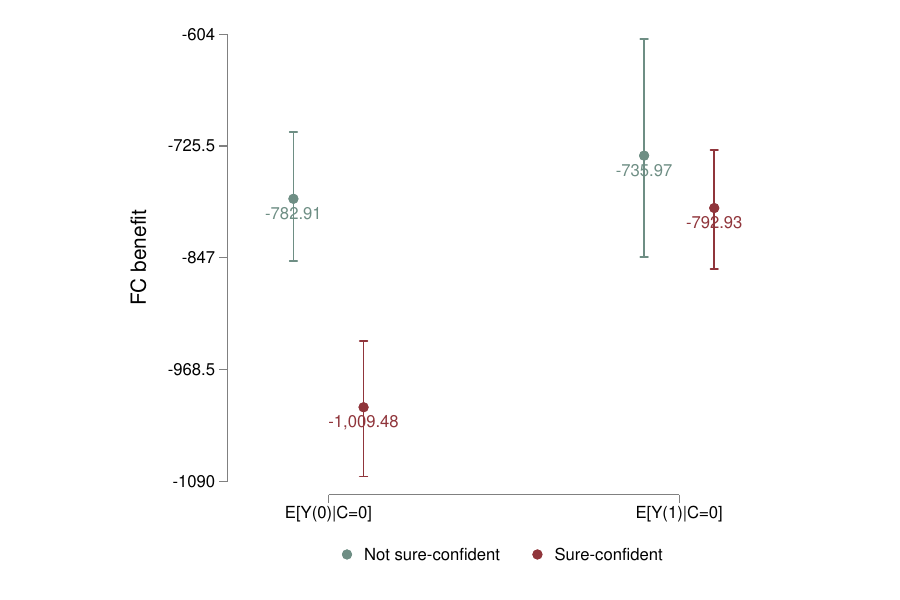}
        \caption{Mandatory structured arm}
    \end{subfigure} 
   \begin{subfigure}{0.49\textwidth}
        \centering
        \includegraphics[width=\textwidth]{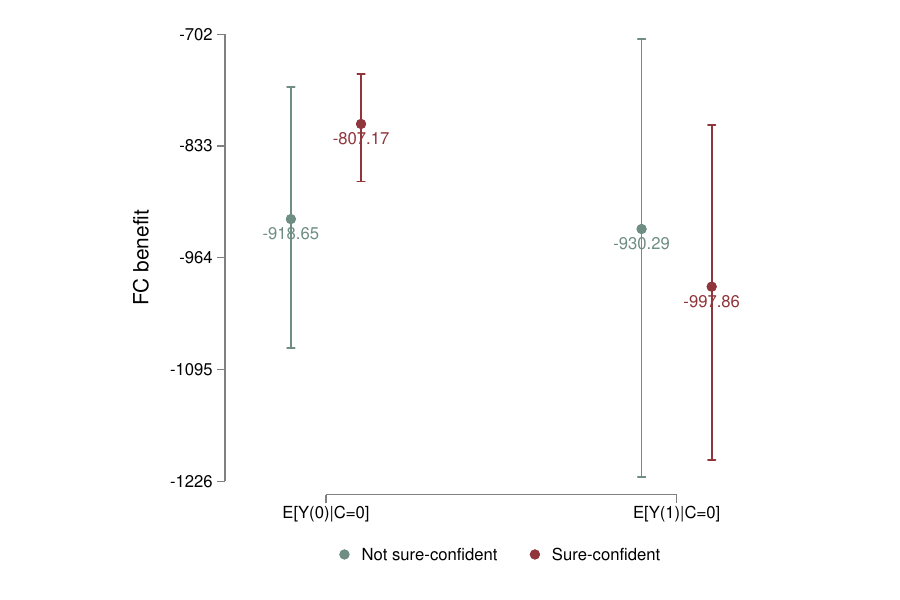}
       \caption{Promise arm}
    \end{subfigure}     
    \end{center}
      \footnotesize{This figure uses our mandates-choice design to back out counterfactual outcomes (Financial Cost) for the overconfident and the non-overconfident. Panel (a) shows the counterfactuals in the structured contract arm. Panel (b) show this counterfactuals in the promise arm.}
      \label{counterfactual_TUT}
\end{figure}

\subsection{TuT partition by behavioral variables}

\begin{figure}[H]
\caption{Differences in TUT Estimated by behavioral variables: Promise Arms} 
    \begin{center}
        \centering
\includegraphics[width=0.65\textwidth]{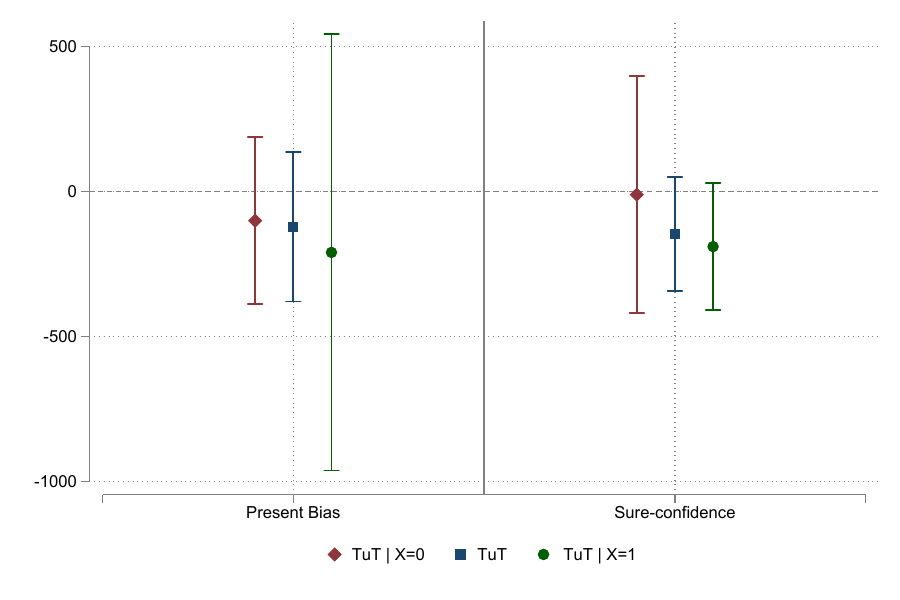} 
    \end{center}  
    \footnotesize{This figure replicates Figure \ref{tut_beh_partition} for financial cost using the Promise Arms rather than the Structured arm. Error bars show 95\% confidence intervals with standard errors clustered at the branch-day level.}
    \label{tut_beh_partition_promise}
\end{figure}

\section{The Effect of Structure on Prepayment}

\begin{figure}[H]
\caption{Comparing the Effect of Structure on Prepayment} 
    \begin{center}
       \begin{subfigure}{0.49\textwidth}
        \centering
        \includegraphics[width=\textwidth]{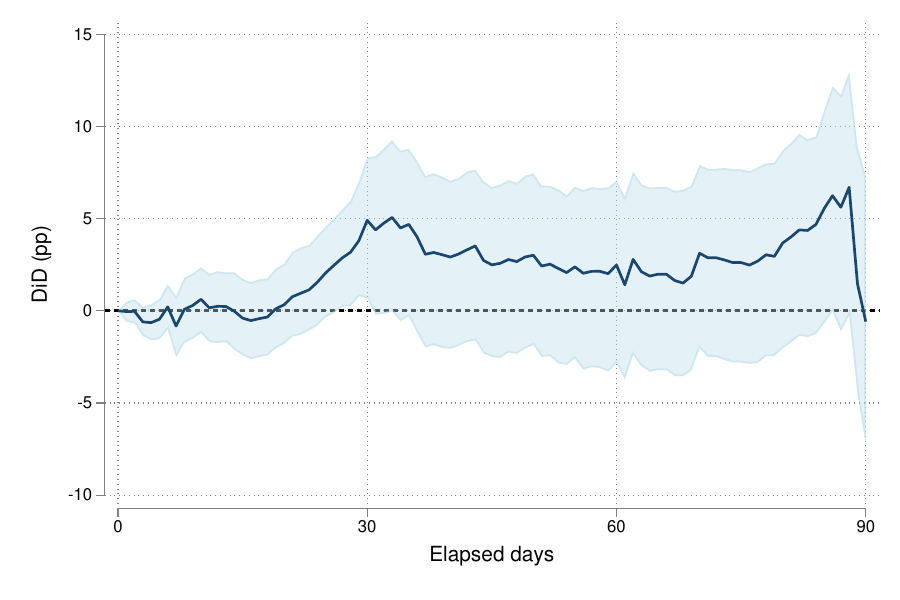}
        \caption{Cumulative prepayments}
    \end{subfigure} 
   \begin{subfigure}{0.49\textwidth}
        \centering
        \includegraphics[width=\textwidth]{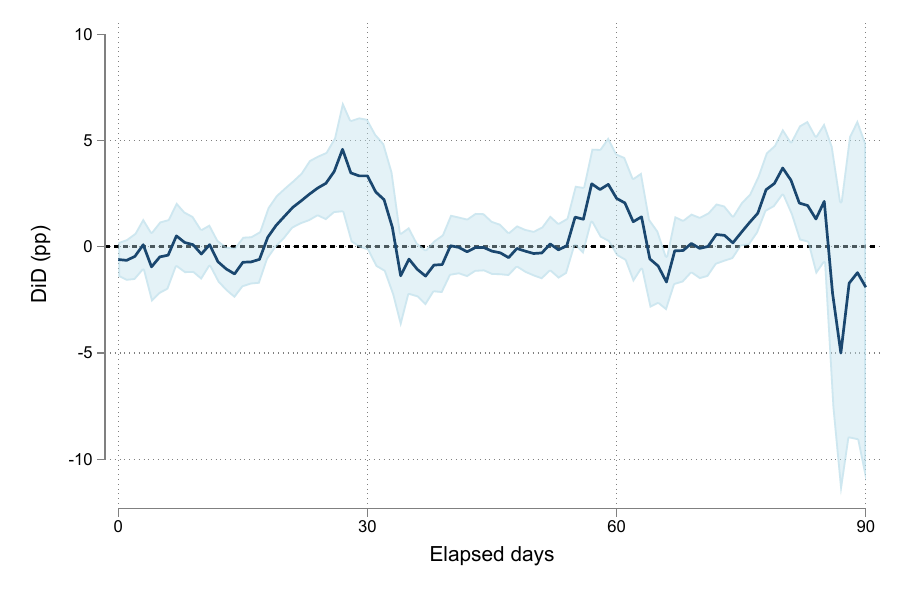}
       \caption{Prepayments within a 7-day window}
    \end{subfigure}     
    \end{center}
      \footnotesize{This figure shows that sure confident borrowers are more likely to react to mandated structure by repaying than other borrowers. Each panel presents results from a regression of a prepayment indicator on an indicator for being sure confident, an indicator for being assigned to the structured payments arm, and their interaction. Panel (a) uses ``made a prepayment by day x'' as the outcome variable, while panel (b) uses ``made a payment exceeding 33\% of the loan on at least one day within a 7-day window centered at day x''. The solid line gives the estimated coefficient for the interaction effect---the difference in causal effects of mandated structure across borrower types---while the shaded regions indicate associated pointwise 90\% confidence intervals.}
      \label{SC_prepayment}
\end{figure}

\section{A Model of Overconfidence in Pawnshop Loans} 
\label{append:model}

This appendix presents a stylized model of how borrower beliefs interact with structured repayment contracts. 
The model shows why sure confident borrowers may respond more strongly to mandated structured payments, and why they may have a correspondingly larger TUT effect.

There are three periods.
In the first period each borrower is assigned a contract: structured repayment or the status quo.
In the second period, each borrower chooses whether to make a prepayment ($P = 1$) or not ($P = 0$).
The amount of the prepayment is fixed at $k$ for all borrowers.
In the third period, the outcome is realized: the borrower either recovers her pawn ($R = 1$) or defaults ($R = 0$). 

If the borrower makes a prepayment, her objective probability of recovery is $\Pi_1$; if she does not make a prepayment, it is $\Pi_0$.
The borrower's \emph{subjective} probabilities of recovery are $\tilde{\Pi}_0$ if she does not prepay and $\tilde{\Pi}_1$ if she does.
If the borrower recovers, she experiences a net continuation utility of $\nu$. We normalize her continuation utility to zero in default. 

Prepaying $k$ in period two lowers the amount due in period three by $(1 + i)k$, where $i$ is the one period interest rate.
Prepayment is risky because all prepayments are lost in default.
Prepayment is also costly: in addition to the prepayment itself, the borrower also incurs a liquidity/hassle cost of $\mu$ for making a prepayment in period two.

The structured contract introduces a fee $f$ that the borrower must pay to recover her pawn if she \emph{does not} prepay.\footnote{As discussed in Section \ref{sec:limitations}, the ``fee'' in this model may be interpreted more broadly as a reduced form for behavioral channels that are absent from the Promise arm but present in the Structured payments arm.}
Crucially, the fee is \emph{state-contingent}: since borrowers who default cannot be made to pay, $f$ only affects those who \emph{recover} without prepaying.
Thus, a borrower who is confident that she will recover (high $\tilde{\Pi}_0$) perceives the fee as nearly certain if she forgoes prepayment, while seeing little risk of losing a prepayment in default.
This implies that borrowers most confident about recovery may be most responsive to structure.

Utility is linear with no discounting.
Let $\tilde{U}(1)$ be the borrower's subjective expected utility from prepayment ($P = 1$), which does not depend on contract type:
\[
\tilde{U}(1) = \underbrace{\tilde{\Pi}_1\left[\nu + (1 + i)k \right]}_\text{Expected Prepayment Benefit} - \underbrace{(\mu + k).}_\text{Prepayment Cost}
\]
Without prepayment ($P=0$), subjective expected utility depends on contract type: $\tilde{U}_\text{SQ}(0) = \tilde{\Pi}_0\nu$ under the status quo versus $\tilde{U}_\text{Struct}(0) = \tilde{\Pi}_0(\nu - f)$ under structure.

To analyze prepayment decisions, define the \textbf{subjective net benefit} $\tilde{B}$ as the perceived gain from prepaying under status quo, excluding the liquidity cost $\mu$:
\[
\tilde{B} \equiv \tilde{\Pi}_1[\nu + (1 + i)k] - k - \tilde{\Pi}_0 \nu.
\]
Define the \textbf{subjective penalty} $\tilde{\Delta} \equiv \tilde{U}_\text{SQ}(0) - \tilde{U}_\text{Struct}(0) = \tilde{\Pi}_0 f$ as the utility loss from not prepaying under structure.
We assume that $\tilde{\Pi}_0>0$, implying $\tilde{\Delta} > 0$.
In words: borrowers believe that recovery is possible even without prepayment. 

A borrower prepays whenever her perceived net benefit exceeds the liquidity cost $\mu$.
Under the status quo contract, this net benefit is $\tilde{B}$, so she prepays when $\tilde{B} > \mu$.
The structured contract makes prepayment more attractive: by prepaying, a borrower gains the baseline benefit $\tilde{B}$ and averts the penalty $\tilde{\Delta}$.
Thus, under the structured contract she will prepay if $\tilde{B} + \tilde{\Delta} > \mu$.
Now define the prepayment potential outcomes:
\[
P_\text{SQ} = \mathbbm{1}\{\tilde{B} > \mu\} \quad \text{and} \quad
P_\text{Struct} = \mathbbm{1}\{\tilde{B} + \tilde{\Delta} > \mu\}.
\]
Since $\tilde{\Delta}>0$, we have $P_\text{Struct} \geq P_\text{SQ}$.
In words: there are no borrowers who prepay under the status quo but not under structure.

This monotonicity result leaves us with three borrower types: \emph{never-payers} have $P_\text{SQ} = P_\text{Struct} = 0$, \emph{always-payers} $P_\text{SQ} = P_\text{Struct} = 1$ have, and \emph{switchers} have $P_\text{Struct} > P_\text{SQ}$. 
Only switchers respond to mandated structure.
By definition, switchers have liquidity costs in the interval $\mu \in [\tilde{B}, \tilde{B} + \tilde{\Delta})$.
Call this the \textbf{switcher interval}.
Since $\tilde{\Delta} = \tilde{\Pi}_0 f > 0$, the width of the switcher interval is positive; even borrowers who believe prepayment has no causal effect on recovery ($\tilde{\Pi}_1 = \tilde{\Pi}_0$) can be induced to prepay, depending on their liquidity cost.

Note that our model implies no borrower would choose structure in period one: since $\tilde{U}_\text{SQ}(0) > \tilde{U}_\text{Struct}(0)$, the status quo weakly dominates from the perspective of subjective expected utility. 
This is by design: our model applies to the 89\% of borrowers who did not choose structure in our experiment, as its purpose is to elucidate TUT effects.\footnote{The 11\% who chose structure may have considerations outside our framework, such as sophisticated time-inconsistency or a taste for novelty.}
Although these borrowers do not choose structure, this does not imply that mandating it would reduce their welfare.
Choices are based on subjective beliefs $(\tilde{\Pi}_0, \tilde{\Pi}_1)$ while welfare depends on objective probabilities $(\Pi_0, \Pi_1)$. 
When beliefs are incorrect, mandated structure can improve outcomes even for those who would never select it.

To match our comparison of TUT effects using the baseline survey, call a borrower ``sure confident'' if she is subjectively certain she will recover her pawn regardless of prepayment: $\tilde{\Pi}_1 = \tilde{\Pi}_0 = 1$.
Despite believing prepayment has no effect on recovery, sure confident borrowers may still respond to structure.
For the sure confident we have $\tilde{B}_\text{Sure} = ik, \quad \tilde{\Delta}_\text{Sure} = f$.
Thus, a sure confident borrower is a switcher if $\mu \in [ik, ik + f)$.
The switcher interval is in fact \emph{widest} for the sure confident, since $\tilde{\Delta} = \tilde{\Pi}_0 f$ is maximized when $\tilde{\Pi}_0 = 1$.
In words: sure confident borrowers expect the fee to hurt them \emph{more} because they are subjectively certain that they will recover without prepayment.
Moreover, sure confident borrowers do not perceive prepayment to be risky: since they are sure they will recover their pawn, they do not fear losing a prepayment in default, simplifying their perceived net benefit of prepayment to $ik$.
Under conditions given in Appendix \ref{append:sure-confident}, this is lower than the perceived net benefit for other borrowers.
Under plausible assumptions on the distribution of $\mu$, these two effects imply that the mass of switchers is \emph{highest} among sure confident borrowers. 
Not only do they respond to structure, they may in fact respond \emph{more} than other borrowers.

Prepayment choice depends on subjective beliefs $(\tilde{\Pi}_1, \tilde{\Pi}_0)$, but treatment effects depend on objective probabilities $(\Pi_1, \Pi_0)$. 
We assume that structure affects neither probabilities nor beliefs---it merely introduces the fee $f$. 
Thus, prepayment fully mediates the causal effect of structure on recovery and the conditional average treatment effect for a borrower with beliefs $(\tilde{\Pi}_0, \tilde{\Pi}_1)$ is
\begin{equation}
\text{CATE}(\tilde{\Pi}_0, \tilde{\Pi}_1) = \mathbbm{P}(\text{Switcher}| \tilde{\Pi}_0, \tilde{\Pi}_1) \times \mathbbm{E}(\Pi_1 - \Pi_0 | \tilde{\Pi}_0, \tilde{\Pi}_1, \text{Switcher}),
\label{eq:ate-beliefs}
\end{equation}
this is the product of the share of switchers among borrowers with these beliefs, and the average causal effect of prepayment on recovery for those switchers.\footnote{Since no one chooses structure in our model, the ATE coincides with the TUT. Only switchers contribute to treatment effects; always-payers and never-payers have zero effect.}

For sure confident borrowers to experience positive treatment effects, two conditions must hold.
First, there must be a positive mass of switchers among the sure confident.
This was established above.
Second, prepayment must improve recovery ($\Pi_1 > \Pi_0$) among the sure confident despite their belief to the contrary. 
Given that 42\% of sure confident borrowers in our experiment defaulted under the status quo, many were clearly overconfident. 
Combined with their higher share of switchers, this provides a potential explanation for why sure confident borrowers appear to have a larger TUT effect. 

Under our model, mandated structure has heterogeneous welfare effects.
Always-payers are unaffected because they prepay regardless of contract type and hence never incur the fee $f$. 
Never-payers who recover are harmed: they incur the fee $f$ because structure cannot change their prepayment decision.
For switchers, welfare effects depend on belief accuracy.
Under correct beliefs ($\tilde{\Pi}_0 = \Pi_0, \, \tilde{\Pi}_1 = \Pi_1$) mandated structure is unambiguously harmful, inducing suboptimal prepayment.
Let $B$ denote the \textbf{objective net benefit}, namely
\[
B \equiv \Pi_1[\nu + (1 + i)k] - k - \Pi_0 \nu.
\]
Under incorrect beliefs, structure is welfare-improving when $\tilde{B} < \mu < B$, that is when prepayment is optimal, but the borrower would not choose to prepay under the status quo, given her beliefs. 
The gap between subjective and objective net benefits ($\tilde{B} < B$) that is required for mandated structure to yield welfare gains can arise from overconfidence about baseline recovery ($\tilde{\Pi}_0 > \Pi_0$) or underestimation of prepayment efficacy ($\Pi_1 > \tilde{\Pi}_1$).

\section{When do the sure confident respond most?}
\label{append:sure-confident}
In this addendum to our structural model from Appendix \ref{sec:model} above, we present sufficient conditions under which the mass of ``switchers'' is \emph{highest} among sure confident borrowers. 
Suppose that the interest rate $i$, the prepayment amount $k$ and the continuation utility $\nu$ are fixed.
Let liquidity costs $\mu$ be independent of beliefs and have density $g$ on $\mathbb{R}_+$. 
For a belief type $(\tilde\Pi_0,\tilde\Pi_1)$ define the subjective net benefit $\tilde{B}$ and penalty $\tilde{\Delta}$ as above:
\[
\tilde B(\tilde\Pi_0,\tilde\Pi_1) \equiv \tilde{\Pi}_1[\nu+(1+i)k]-k-\tilde{\Pi}_0\nu,
\qquad
\tilde\Delta(\tilde\Pi_0) \equiv \tilde{\Pi}_0 f.
\]
The mass of switchers among borrowers with beliefs $(\tilde\Pi_0,\tilde\Pi_1)$ is given by
\[
s(\tilde\Pi_0,\tilde\Pi_1)
\equiv \mathbbm{P}\big(\tilde B(\tilde\Pi_0,\tilde\Pi_1)\le \mu <
\tilde B(\tilde\Pi_0,\tilde\Pi_1)+\tilde\Delta(\tilde\Pi_0)\big)
= \int_{b}^{b+\delta} g(\mu)\,d\mu,
\]
adopting the shorthand $b \equiv \tilde B(\tilde\Pi_0,\tilde\Pi_1)$ and $\delta \equiv \tilde\Delta(\tilde\Pi_0)$. 
For sure-confident borrowers $(\tilde\Pi_0,\tilde\Pi_1)=(1,1)$, so we have $b = ik$ and $\delta = f$.

\paragraph{Uniform Liquidity Costs.}
Suppose that $\mu \sim \text{Uniform}(0, M)$ where $M$ is sufficiently large that $b + \delta \leq M$ for all belief types. 
Then $s(\tilde{\Pi}_0, \tilde{\Pi}_1) = \delta / M$.
Since $\delta = \tilde{\Pi}_0 f$ is maximized when $\tilde{\Pi}_0 = 1$, it follows that $s(1, 1) \geq s(\tilde{\Pi}_0, \tilde{\Pi}_1)$ for all other belief types, with strict inequality for $\tilde{\Pi}_0 \neq 1$.

\paragraph{Decreasing Density and Belief Restriction.}
Alternatively, suppose that $g$ is strictly decreasing on $\mathbb{R}_+$ and that beliefs satisfy 
\[
(1 - \tilde{\Pi}_1)k \leq (\tilde{\Pi}_1 - \tilde{\Pi}_0) \left( \frac{1}{1 + i}\right) \nu 
\]
for all borrowers.
The left-hand side is the borrower's expected period-two loss from prepayment; the right-hand side is her expected period-two gain from prepayment.
Setting liquidity concerns aside, the harm a borrower expects from possibly losing a prepayment is no greater than the benefit she expects from making it.
Under this condition, we have $\tilde{B}(\tilde{\Pi}_0, \tilde{\Pi}_1) \geq ik$ implying there is no belief type with a lower subjective net benefit of prepayment than the sure confident.
Thus, when evaluating $\int_b^{b + \delta} g(\mu) \, d\mu$ for different belief types, $b$ is smallest and $\delta$ is largest for sure confident borrowers.
(Recall that $\delta = \tilde{\Pi}_0 f$.) 
Since $g$ is strictly decreasing, it follows that $s(1, 1) \geq s(\tilde{\Pi}_0, \tilde{\Pi}_1)$ for all other belief types. 

\section{Sure Confidence and Welfare Effects} 
\label{append:welfare}
In this section, we build on the analysis from Online Appendix \ref{append:sure-confident} to show that the aggregate welfare effect of mandated structure is more positive (or less negative) the greater the share of sure-confident borrowers in the population. Let $p$ be the share of sure-confident borrowers and $1 - p$ be the share of borrowers with correct beliefs.
Let $\Delta W_\text{Sure}$ be the expected change in welfare that a sure confident borrower experiences when structure is mandated. Define $\Delta W_\text{Correct}$ analogously for a borrower with correct beliefs. Then,  the aggregate change in expected welfare is given by $\Delta W(p) = p\Delta W_\text{Sure} + (1 - p) \Delta W_\text{Correct}$. 
Therefore $\Delta W(p)$ is strictly increasing in $p$ whenever $\Delta W_\text{Sure} > \Delta W_\text{Correct}$.

Mandating structure changes the expected utility of switchers by $B - \mu$ and that of never-payers by $-\Pi_0 f$.
Thus, under the uniform liquidity costs assumption from Online Appendix \ref{append:sure-confident}, 
\[
\Delta W_{\text{Sure}} - \Delta W_{\text{Correct}} = \frac{(1-\Pi_0)f}{M}\left[(B - ik) - \frac{(1-\Pi_0)f}{2}\right].
\]
It follows that, under uniform liquidity costs, $\Delta W(p)$ is increasing in $p$ if and only if $(B - ik) > (1 - \Pi_0) f/2$.
Recall that $B - ik$ is the wedge between the objective net benefit of prepayment, $B$, and a sure-confident borrower's \emph{perceived} benefit of prepayment, $ik$. 
Thus, the inequality requires that the fee is not too large relative to sure confident borrowers' belief distortion. 
This is plausible in our setting, given the small size of the fee and the decrease in default from mandated structure.
It can be shown that the same inequality is sufficient, although no longer necessary, if the density of liquidity costs is weakly decreasing rather than uniform.

\begingroup
  \renewcommand{\refname}{Online Appendix References}
  \putbib
\endgroup
\end{bibunit}

\clearpage

\setcounter{table}{0}
\setcounter{figure}{0}
\setcounter{section}{0}

\pagenumbering{arabic}
\renewcommand\thefigure{SA-\arabic{figure}}
\renewcommand\thetable{SA-\arabic{table}}
\renewcommand*{\thepage}{SA-\arabic{page}}
\renewcommand\thesection{Appendix \Alph{section}.}
\renewcommand\theHsection{SA.\Alph{section}}
\renewcommand\thesubsection{\Alph{section}.\arabic{subsection}}

\begin{center}
	\Large SUPPLEMENTARY APPENDIX: Structured Payment in Pawnshop Borrowing: Mandates vs.\ Choice \\[0.5em]
	\large Francis J.\ DiTraglia, Craig McIntosh, Isaac Meza, Joyce Sadka and Enrique Seira
\end{center}

\begin{bibunit}

\section{Testable Implications of the Exclusion Restriction} \label{sec:testing_exclusion}

Let $Y_0 \equiv Y(0,0)$ and $Y_1 \equiv Y(1,1)$ denote the potential outcomes under \emph{mandatory treatment}: $Y_0$ is the potential outcome when assigned to the status quo contract and $Y_1$ when mandated to the structured contract. 
Further let $Y_{0,2} \equiv Y(0,2)$ and $Y_{1,2} \equiv Y(1,2)$ denote the potential outcomes under \emph{free choice of treatment}: $Y_{0,2}$ is the potential outcome when choosing the status quo contract and $Y_{1,2}$ when choosing the structured contract. 
Using this notation, the exclusion restrictions upon which our TUT and TOT results rely are $Y_0 = Y_{0,2}$ and $Y_1 = Y_{1,2}$.\footnote{While we assume $Y_i(d,z) = Y_i(d)$ in the body of the paper for simplicity, our derivations only rely on these two equalities.}
Without imposing these, Assumption \ref{assump:randchoice}(iii) becomes 
\[
Y = \mathbbm{1}(Z =0) Y_{0} + \mathbbm{1}(Z = 1)  Y_{1}  + \mathbbm{1}(Z = 2) \left[(1 - C) Y_{0,2} + C Y_{1,2} \right]
\]
but parts (i) and (ii) continue to hold.
Accordingly, parts (i)--(iii) of Lemma \ref{lem_randchoice} are unchanged, while parts (iv) and (v) become
\[
\mathbbm{E}(Y|D=0,Z=2) = \mathbbm{E}(Y_{0,2}|C=0), \quad
\mathbbm{E}(Y|D=1,Z=2) = \mathbbm{E}(Y_{1,2}|C=1).
\]
Using these expressions, the testable restrictions we consider here are as follows:
\begin{align}
\label{eq:exclusion_append0}
\mathbbm{E}(Y_0|C=0) &= \mathbbm{E}(Y_{0,2}|C=0) \\
\label{eq:exclusion_append1}
\mathbbm{E}(Y_1|C=1) &= \mathbbm{E}(Y_{1,2}|C=1).
\end{align}
Because they refer to different groups of people--choosers versus non-choosers--either of \eqref{eq:exclusion_append0}  and \eqref{eq:exclusion_append1} could hold when the other is violated.
For this reason we consider each in turn.
Our approach is closely related to arguments from \cite{huber_mellace} and \cite{BinaryRegressor}, among others.

Consider first \eqref{eq:exclusion_append0}.
Let $p \equiv \mathbbm{P}(C=1) = \mathbbm{P}(D=1|Z=2)$ denote the share of choosers in the population.
This value is point identified regardless of whether the exclusion restriction holds.
Because $Z$ was randomly assigned, a fraction $p$ of borrowers with $Z=0$ are choosers while the remaining $(1-p)$ are non-choosers.
It follows that, regardless of whether the exclusion restriction holds, the observed distribution of $Y|Z=0$ is a mixture of $Y_0|C=0$ and $Y_0|C=1$ with mixing weights $(1-p)$ and $p$.
This allows us to construct a pair of bounds for $\mathbb{E}(Y_0|C=0)$ as follows.
The non-choosers must lie \emph{somewhere} in the distribution of $Y|Z=0$.
Consider the two most extreme possibilities: they could occupy the bottom $(1-p)\times 100\%$ of the distribution or the top $(1-p)\times 100\%$ of the distribution.
For this reason, computing the average of the \emph{truncated} distribution of $Y|Z=0$, cutting out the top $p\times 100\%$, provides a lower bound for the average of $Y_0$ among non-choosers.
Similarly, cutting out the bottom $p \times 100\%$ provides an upper bound.
Let $y^0_{1-p}$ denote the $(1-p)$ quantile of $Y|Z=0$ and $y^0_{p}$ denote the $p$ quantile of the same distribution.
Using this notation, the bounds are given by
\[
\mathbb{E}\left(Y|Z=0, Y\leq y^0_{1-p}\right)\leq \mathbb{E}(Y_0|C=0) \leq \mathbb{E}\left(Y|Z=0, Y \geq y^0_p\right) 
\]
These bounds do not rely on the exclusion restriction.
Under Equation \ref{eq:exclusion_append0}, however, we know that $\mathbb{E}(Y_0|C=0)=\mathbb{E}(Y|D=0,Z=2)$.
Therefore, if the exclusion restriction for non-choosers holds, we must have
\begin{equation}
\mathbb{E}\left(Y|Z=0, Y\leq y^0_{1-p}\right)\leq \mathbb{E}(Y|D=0,Z=2) \leq \mathbb{E}\left(Y|Z=0, Y \geq y^0_p\right).
\label{eq:testable0}
\end{equation}
Equation \ref{eq:testable0} provides a pair of testable implications of \eqref{eq:exclusion_append0}.
If either inequality is violated, then the exclusion restriction for non-choosers fails.
In our experiment, $\widehat{p} = \widehat{\mathbbm{P}}(D=1|Z=2) = 0.11$. For the CR outcome we estimate
\[
\widehat{\mathbbm{E}}(Y_\text{CR}|Z=0,Y_\text{CR}\leq y^0_{0.89}) = 0.38, \quad
\widehat{\mathbbm{E}}(Y_\text{CR}|Z=0, Y_\text{CR}\geq y^0_{0.11}) = 0.47.
\]
Since $\widehat{\mathbbm{E}}(Y_\text{CR}|D=0,Z=2) = 0.43$ falls between these bounds, we find no evidence against the exclusion restriction for non-choosers. 

We can use an analogous approach to construct testable implications for \ref{eq:exclusion_append1}, yielding 
\begin{equation}
\mathbbm{E}\left(Y|Z=1,Y\leq y^1_{p}\right) \leq \mathbbm{E}(Y|D=1,Z=2) \leq \mathbbm{E}\left(Y|Z=1,Y \geq y^1_{1-p}\right).
\label{eq:testable1}
\end{equation}
where $y^1_p$ and $y^1_{1-p}$ are the $p$ and $1 - p$ quantiles of the distribution of $Y|Z=1$.
If either inequality is violated, then the exclusion restriction from Equation \ref{eq:exclusion_append1} fails.
Again, in our experiment $\widehat{p}=0.11$. For the CR outcome we estimate
\[
\widehat{\mathbbm{E}}(Y|Z=1,Y\leq y^1_{0.11}) = 0.06, \quad
\widehat{\mathbbm{E}}(Y|Z=1,Y\geq y^1_{0.89}) = 0.82
\]
Since $\widehat{\mathbbm{E}}(Y_\text{CR}|D=1,Z=2) = 0.34$ falls between these bounds, we find no evidence against the exclusion restriction for the choosers. 
The same results hold for the financial cost outcome: results available upon request.

\section{Estimation and Inference} 
\label{sec:estimation_inference}

\subsection{Regression-based Estimation} 

In this appendix we derive simple, regression-based estimators of the TOT, TUT, ASG, ASL, and ASB effects.
Let $Z_0 \equiv \mathbbm{1}\{Z=0\}$, $Z_1 \equiv \mathbbm{1}\{Z=1\}$, and $Z_2 \equiv \mathbbm{1}\{Z=2\}$.
Under standard regularity conditions, the following proposition shows that an IV regression of $Y$ on an intercept, $Z_1$ and $Z_2 D$ with instruments $(1, Z_0, Z_1)$ provides consistent estimates the ATE and TOT, while an IV regression of $Y$ on an intercept, $-Z_0$ and $-Z_2(1-D)$ with the same instrument set consistently estimates the ATE and TUT effects.

\begin{prop} 
\label{prop:TOTandTUTregs}
Under Assumption \ref{assump:randchoice}, 
\begin{enumerate}[(i)] 
\item $Y = \mathbbm{E}(Y_0) + \text{ATE}\times Z_1 + \text{TOT} \times Z_2 D + U$
\item $Y = \mathbbm{E}(Y_1) + \text{ATE}\times -Z_0 + \text{TUT} \times -Z_2 (1 - D) + V$
\end{enumerate}
where $\mathbbm{E}(U|Z) = \mathbbm{E}(V|Z) = 0$.
\end{prop}

\begin{proof}[Proof of Proposition \ref{prop:TOTandTUTregs}]
For part (i), since $Z_2 D = Z_2 C$ and $(Z_0 + Z_1 + Z_2) = 1$, Assumption \ref{assump:randchoice} (iii) implies $Y= Y_0 + Z_1 (Y_1 - Y_0) + Z_2D(Y_1 - Y_0)$.
Now define 
\[
U \equiv [Y_0 - \mathbbm{E}(Y_0)] + Z_1[(Y_1 - Y_0) - \text{ATE}] + Z_2 D[(Y_1 - Y_0) - \text{TOT}].
\]
Since $Z_2 D = Z_2 C$ and $Z$ is independent of $(Y_1, Y_0)$ by Assumption \ref{assump:randchoice} (i), it follows that $\mathbbm{E}(U|Z) = Z_2 \mathbbm{E}\left[C\left\{(Y_1 - Y_0) - \text{TOT}\right\}|Z\right]$.
Thus, by iterated expectations,
\begin{align*}
\mathbbm{E}\left[C\left\{(Y_1 - Y_0) - \text{TOT}\right\}|Z\right] 
&= \mathbbm{P}(C=1)\left[\mathbbm{E}(Y_1 - Y_0|C=1) - \text{TOT} \right]=0
\end{align*}
since $Z$ is independent of $(Y_0, Y_1)$ given $C$, an implication of Assumption \ref{assump:randchoice} (i).

For part (ii), since $Z_2(1 - C)= Z_2(1 - D)$ and $(Z_1 + Z_2) = 1 - Z_0$,  Assumption \ref{assump:randchoice} (iii) implies $Y= Y_1 - Z_0 (Y_1 - Y_0) - Z_2(1 - D) (Y_1 - Y_0)$.
Define 
\[
V \equiv [Y_1 - \mathbb{E}(Y_1)] - Z_0[(Y_1 - Y_0) - \text{ATE}] - Z_2(1 - D)[(Y_1 - Y_0) - \text{TUT}].
\]
Since $Z_2(1 - D) = Z_2 (1 - C)$ and $Z$ is independent of $(Y_0, Y_1)$ by Assumption \ref{assump:randchoice} (i),
$\mathbb{E}(V|Z) = -Z_2\mathbb{E}[(1 - C)\left\{(Y_1 - Y_0) - \text{TUT}\right\}|Z]$.
Thus, by iterated expectations,
\begin{align*}
\mathbb{E}[(1 - C)\left\{(Y_1 - Y_0) - \text{TUT}\right\}|Z]
&= \mathbb{P}(C=0|Z) \left[ \mathbb{E}(Y_1 - Y_0|C=0) - \text{TUT}\right] = 0
\end{align*}
since $Z$ is independent of $(Y_0, Y_1)$ given $C$, an implication of Assumption \ref{assump:randchoice} (i).
\end{proof}

Since $\text{ASG} = \text{TOT} - \text{TUT}$, the preceding proposition provides a consistent estimate of the $\text{ASG}$ effect.
The $\text{ASB}$ effect, $\mathbbm{E}(Y_0|C=1) - \mathbb{E}(Y_0|C=0)$, can likewise be estimated by taking the difference of coefficients across two linear IV regressions with \emph{no intercept} and instrument sets $(Z_0, Z_2)$, as shown in the following proposition.

\begin{prop}
\label{prop:ASBreg}
Under Assumption \ref{assump:randchoice}
\begin{enumerate}[(i)]
    \item $(1 - D) Y = \mathbbm{E}(Y_{0}) \times Z_0 + \mathbbm{E}(Y_{0}|C=0) \times (1 - D)Z_2 + U_{0}$
    \item  $(1 - D) Y = \mathbbm{E}(Y_{0}) \times(Z_0 + Z_2) + \mathbbm{E}(Y_{0}|C=1)\times -DZ_2 + U_{1}$
\end{enumerate}
where $\mathbbm{E}(U_0|Z) = \mathbbm{E}(U_1|Z) = 0$.
\end{prop}

\begin{proof}
Assumption \ref{assump:randchoice} (ii) implies $(1 - D) = Z_0 + Z_2(1 - C)$. Hence, 
\begin{align*}
(1 - D)Y 
&= Z_0 Y_0 + Z_2 (1 - C) [(1 - C) Y_0 + C Y_1] = Z_0 Y_0 + Z_2 (1 - C) Y_0
\end{align*}
by Assumption \ref{assump:randchoice} (iii), since $Z_j^2 = Z_j$ for any $j$ and $Z_j Z_k = 0$ for any $j \neq k$ and, similarly, $(1 - C)^2 = (1 - C)$ and $C (1 - C) = 0$.
Therefore, since $Z_2 (1 - C) = Z_2 (1 - D)$,
\[
(1 - D)Y = Z_0 Y_0 + Z_2 (1 - D) Y_0, \quad
(1 - D)Y = (Z_0 + Z_2) Y_0 + (-DZ_2) Y_0. 
\]
Now, define 
\begin{align*}
U_0 &\equiv Z_0 [Y_0 - \mathbbm{E}(Y_0)] + Z_2(1 - D)[Y_0 - \mathbbm{E}(Y_0|C=0)]\\
U_1 &\equiv (Z_0 + Z_2)[Y_0 - \mathbbm{E}(Y_0)] + (-Z_2 D)[Y_0 - \mathbbm{E}(Y_0|C=1)].
\end{align*}
Since $Z_2(1-D) = Z_2(1 - C)$, and $Z$ is independent of $Y_0$,  
\begin{align*}
\mathbbm{E}(U_0|Z) 
&= Z_2 \, \mathbbm{P}(C = 0 | Z)\, \mathbbm{E}[Y_0 - \mathbbm{E}(Y_0|C=0) | C = 0, Z] = 0
\end{align*}
by iterated expectations and the fact that $Z$ is conditionally independent of $Y_0$ given $C$. 
Since $Z_2 D = Z_2 C$, a nearly identical argument gives
\begin{align*}
\mathbbm{E}(U_1|Z) 
&= -Z_2 \, \mathbbm{P}(C = 1 | Z)\, \mathbbm{E}[Y_0 - \mathbbm{E}(Y_0|C=1) | C = 1, Z] = 0. \qedhere
\end{align*}
\end{proof}

The final result in this section implies that the $\text{ASL}$ effect, $\mathbbm{E}(Y_1|C=1) - \mathbb{E}(Y_1|C=0)$, can be estimated as the difference of coefficients across two linear IV regressions with \emph{no intercept} and instrument set $(Z_1, Z_2)$.

\begin{prop}
\label{prop:ASLreg}
Under Assumption \ref{assump:randchoice},
\begin{enumerate}[(i)]
\item $D Y = \mathbbm{E}(Y_{1}) \times(Z_1 + Z_2) + \mathbbm{E}(Y_{1}|C=0)\times (D - 1)Z_2 + V_{0}$
\item  $D Y = \mathbbm{E}(Y_{1}) \times Z_1 + \mathbbm{E}(Y_{1}|C=1)\times D Z_2+ V_{1}$
\end{enumerate}
where $\mathbbm{E}(V_0|Z) = \mathbbm{E}(V_1|Z) = 0$.
\end{prop}

\begin{proof}
By Assumption \ref{assump:randchoice}, $D = Z_1 + Z_2 C$. Hence, by Assumption \ref{assump:randchoice} (iii),
\begin{align*}
DY 
&= Z_1 Y_1 + Z_2 C[(1 -C) Y_0 + C Y_1] = Z_1 Y_1 + Z_2 C Y_1
\end{align*}
because $Z_j^2 = Z_j$ for any $j$ and $Z_j Z_k = 0$ for any $j \neq k$ and, similarly, $(1 - C)^2 = (1 - C)$ and $C (1 - C) = 0$. Therefore, since $Z_2 (1 - C) = Z_2 (1 - D)$,
\[
 DY = (Z_1 + Z_2)Y_1 + Z_2 (D - 1) Y_1, \quad
 DY = Z_1 Y_1 + Z_2 D Y_1.
\]
Now, define
\begin{align*}
V_0 &= (Z_1 + Z_2)[Y_1 - \mathbbm{E}(Y_1)] + Z_2(D - 1)[Y_1 - \mathbbm{E}(Y_1|C=0)]\\
V_1 &= Z_1[Y_1 - \mathbbm{E}(Y_1)] + Z_2D[Y_1 - \mathbbm{E}(Y_1|C=1)].
\end{align*}
Since $Z_2(1 - D) = Z_2(1 - C)$ and $Z$ is independent of $Y_1$,
\begin{align*}
\mathbbm{E}(V_0|Z) 
&= -Z_2\, \mathbbm{P}(C = 0 | Z)\, \mathbbm{E}[Y_1 - \mathbbm{E}(Y_1|C=0)|C=0, Z] = 0
\end{align*}
by iterated expectations and the fact that  $Z$ is conditionally independent of $Y_1$ given $C$. 
Since $Z_2 D = Z_2 C$, a similar argument gives
\begin{align*}
\mathbbm{E}(V_1|Z) 
&= Z_2\, \mathbbm{P}(C = 1 | Z)\, \mathbbm{E}[Y_1 - \mathbbm{E}(Y_1|C=1)|C=1, Z] = 0. \qedhere
\end{align*}
\end{proof}

\normalsize
\normalsize

\subsection{Inference for ASG, ASB, and ASL}
\label{subsec:inference}

We now explain how to carry out cluster-robust inference for the ASG, ASB, and ASL effects, as implemented in our companion STATA package.
Each of these effects can be expressed as a difference of coefficients from two just-identified linear IV regressions.
The ASG effect is the difference of the \text{TOT} and \text{TUT} effects from Proposition \ref{prop:TOTandTUTregs}. 
Similarly, the ASB effect is the difference of $\mathbb{E}(Y_0|C=1)$ and $\mathbbm{E}(Y_0|C=0)$ from Proposition \ref{prop:ASBreg} while the ASL effect is the difference of $\mathbbm{E}(Y_1|C=1)$ and $\mathbbm{E}(Y_1|C=0)$ from Proposition \ref{prop:ASLreg}.
Within each pair of IV regressions the outcome variable and instrument set is identical; only the regressors differ. 
Since our estimators of all three effects share the same structure, our discussion abstracts from the specific regressors and instruments used in each case.

Let $g = 1, ..., G$ index clusters and $i = 1, ..., N_g$ index individuals within a particular cluster $g$. 
In our experiment, a cluster is a branch-day combination and the experimentally-assigned treatment (control, mandatory, or choice arm) is assigned at the cluster level.
We assume that observations are iid across clusters but potentially correlated within cluster.
Now consider a pair of just-identified linear IV regressions given by
$Y_{ig} = \mathbf{X}_{1,ig}' \boldsymbol{\theta}_0 + U_{ig}$ and $Y_{ig} = \mathbf{X}_{0,ig}' \boldsymbol{\theta}_1 + V_{ig}$ with common instrument vector $\mathbf{W}_{ig}$. 
Stacking observations in the usual manner, e.g.\
$\mathbf{W}_g' \equiv \begin{bmatrix}
\mathbf{W}_{1g} & \cdots & 
\mathbf{W}_{N_gg} 
\end{bmatrix}$ and 
$\mathbf{W}' = \begin{bmatrix}
\mathbf{W}_1' & \cdots & \mathbf{W}_G'
\end{bmatrix}$
we can write the preceding equations in matrix form as $\mathbf{Y} = \mathbf{X}_1\boldsymbol{\theta_1} + \mathbf{U}$ and $\mathbf{Y} = \mathbf{X}_0\boldsymbol{\theta_0} + \mathbf{V}$ with instrument matrix $\mathbf{W}$.
Now, the IV estimators for $\boldsymbol{\theta}_1$ and $\boldsymbol{\theta}_0$ can be expressed as 
\begin{align*}
\widehat{\boldsymbol{\theta}}_1 
&= \left(\mathbf{W}'\mathbf{X}_1\right)^{-1}\mathbf{W}'\mathbf{Y}  = \boldsymbol{\theta}_1 + \left(\mathbf{W}'\mathbf{X}_1\right)^{-1}\mathbf{W}'\mathbf{U}\\ 
\widehat{\boldsymbol{\theta}}_0 
&= \left(\mathbf{W}'\mathbf{X}_0\right)^{-1}\mathbf{W}'\mathbf{Y} = \boldsymbol{\theta}_0 + \left(\mathbf{W}'\mathbf{X}_0\right)^{-1}\mathbf{W}'\mathbf{V}.
\end{align*}

By our experimental design and exclusion restriction, $\mathbf{W}_{ig}$ is independent of $U_{ig}$ both unconditionally and conditional on cluster size. 
Hence, by a standard argument and under mild regularity conditions, the following expression provides a consistent, cluster robust estimator of $\widehat{\text{Avar}}(\widehat{\boldsymbol{\theta}}_1 - \widehat{\boldsymbol{\theta}}_0)$ 
\[
\widehat{\text{Avar}}(\widehat{\boldsymbol{\theta}}_1 - \widehat{\boldsymbol{\theta}}_0) 
= \begin{bmatrix}
\left(\mathbf{W}'\mathbf{X}_1\right)^{-1} &
-\left(\mathbf{W}'\mathbf{X}_0\right)^{-1} 
\end{bmatrix}
\begin{bmatrix}
\mathbf{S}_{UU}& \mathbf{S}_{UV}\\
\mathbf{S}_{UV}'& \mathbf{S}_{VV}
\end{bmatrix} 
\begin{bmatrix}
\left(\mathbf{X}_1'\mathbf{W}\right)^{-1} \\
-\left(\mathbf{X}_0'\mathbf{W}\right)^{-1} 
\end{bmatrix}
\]
where we define the IV residuals
$\widehat{\mathbf{U}}_g \equiv \mathbf{Y}_g - \mathbf{X}_{1,g}\widehat{\boldsymbol{\theta}}_1$ and
$\widehat{\mathbf{V}}_g \equiv \mathbf{Y}_g - \mathbf{X}_{0,g}\widehat{\boldsymbol{\theta}}_0$ along with the matrices
$\mathbf{S}_{UU} \equiv \sum_{g=1}^G \mathbf{W}_g' \widehat{\mathbf{U}}_g \widehat{\mathbf{U}}_g' \mathbf{W}_g$,
$\mathbf{S}_{UV} \equiv \sum_{g=1}^G \mathbf{W}_g' \widehat{\mathbf{U}}_g \widehat{\mathbf{V}}_g' \mathbf{W}_g$, and finally
$\mathbf{S}_{VV} \equiv \sum_{g=1}^G \mathbf{W}_g' \widehat{\mathbf{V}}_g \widehat{\mathbf{V}}_g' \mathbf{W}_g$.
In our application the number of clusters, $G$, is large.
If desired, an ad hoc degrees of freedom correction can be applied by multiplying the associated standard errors by $\sqrt{G/(G-1)}$.

\section{Testing for Heterogeneity and Calibration}
\label{append:het}

Here we provide details for the tests for heterogeneous treatment effects and calibration described in Section \ref{sec:RF} of the paper. 
If the treatment effects $Y_{i1} - Y_{i0}$ are constant across $i$, then we must have
\[
\text{ATE}(X_i) = \mathbbm{E}[Y_{i1} - Y_{i0}|X_i] = \mathbbm{E}[Y_{i1} - Y_{i0}] \equiv \text{ATE}
\]
for any covariates $X_i$ that vary across $i$. If, on the other hand, $\text{ATE}(X_i)$ can be predicted using some scalar function $\tau(\cdot)$ of $X_i$, then the average treatment effect function is not constant so there must be treatment effect heterogeneity.

We operationalize this idea using a two-step approach proposed by \cite{chernozhukov2018generic}. We begin by randomly dividing the participants in the forced arms of the experiment ($Z_i \neq 2$) into two groups: a training set and a test set. These sets are constructed to ensure that all observations from a given branch-day cluster are allocated to the same set. This avoids inferential problems that could arise from correlated unobservables within clusters.

In the first step, we apply the generalized random forest approach of \cite{atheygrf} to the training set to estimate two proxy predictors: $\psi(\cdot\mid\text{Training})$ approximates the untreated potential outcome function, $\mathbbm{E}[Y_{i0}\mid X_i] = \mathbbm{E}[Y_i\mid Z_i=0, X_i]$, while $\tau(\cdot\mid\text{Training})$ approximates the ATE function
\[
\text{ATE}(X_i) = \mathbbm{E}[Y_i\mid Z_i=1, X_i] - \mathbbm{E}[Y_i\mid Z_i=0, X_i].
\]
The proxy predictors need not be unbiased or even consistent estimators of the functions they aim to approximate: the goal is merely to find a scalar function of $X_i$ that \emph{accurately predicts} $\text{ATE}(X_i)$.

In the second step, we fit a linear regression model \emph{on the test set} using regressors constructed from the proxy functions $\psi(\cdot\mid\text{Training})$ and $\tau(\cdot\mid\text{Training})$ obtained in the first step. In particular, we estimate
\begin{equation}
Y_i = \alpha_0 + \alpha_1 \psi_i
+ \beta_1 \bigl(Z_i - \mathbbm{E}[Z_i]\bigr)
+ \beta_2 \bigl(Z_i - \mathbbm{E}[Z_i]\bigr)\bigl(\tau_i - \mathbbm{E}[\tau_i]\bigr)
+ \varepsilon_i,
\label{eq:chernreg}
\end{equation}
where $\psi_i \equiv \psi(X_i\mid\text{Training})$ and $\tau_i \equiv \tau(X_i\mid\text{Training})$ (in practice we implement centering with test-sample means).%
\footnote{This is a slightly simpler regression than the one proposed in equation (3.1) of \cite{chernozhukov2018generic}, which involves propensity score weights. Because the random assignment of $Z$ in our experiment does \emph{not} condition on $X$, the propensity score weights in our case are constant over $X$ and drop out.}

As shown by \cite{chernozhukov2018generic}, the coefficients $\beta_1$ and $\beta_2$ from \eqref{eq:chernreg} identify the \emph{best linear predictor} of the conditional ATE based on $\tau(\cdot\mid\text{Training})$.
If treatment effects are homogeneous we must have $\beta_2 = 0$. Rejecting this hypothesis establishes that $\tau_i$ predicts $\text{ATE}(X_i)$ and hence that $\Delta_i$ varies. Because $\tau_i$ and $\psi_i$ are learned on the training set and do not depend on the \emph{test-set outcomes}, inference for the regression in \eqref{eq:chernreg} is straightforward conditional on the Training/Test split.

Our estimate for $\beta_2$ is 1.81 with a heteroskedasticity-robust (HC3) standard error of 0.33. Thus we easily reject the null hypothesis of homogeneous treatment effects.
The same regression allows us to test the \emph{calibration} of our causal forest model. If the model is correctly calibrated, $\beta_1$ should equal the unconditional ATE. We are unable to reject the null hypothesis that the two are equal.\footnote{Using the \texttt{grf} package, we implement a version of the calibration test that is parameterized slightly differently. Under correct calibration, the parameter $\theta_1 \equiv \beta_1 / \text{ATE}$ should equal one. Our point estimate is $1.01$ with a standard error of $0.12$.}

\section{Can we target paternalism accurately?} \label{targeting}

In this appendix, we evaluate the performance of a feasible targeting rule that assigns contracts to borrowers based on their observed covariates, using the causal forest approach from Section \ref{sec:RF}.
We ask whether it is possible to achieve better aggregate outcomes by assigning structured repayment to \emph{only} those borrowers predicted to benefit from it.

The answer depends on how well our covariates predict treatment effects. But if the goal is to develop a targeting protocol that could be applied in the real world settings, we must limit ourselves to covariates that are either verifiable or hard to manipulate. This leaves us with only age, gender, high school education or above, desired loan size, and whether that individual has ever pawned before. 
We call these the ``narrow'' covariate set, to contrast them with the full set of survey variables, which we call the ``wide'' covariate set. 
Figure \ref{wide_narrow_forests} compares CATEs estimated using the wide covariate set---the same as in Figure \ref{heterogeneous_effects} from Section \ref{sec:RF}---against those estimated using only the narrow set. 

\begin{figure}[h!]
   \caption{Conditional ATEs from ``wide'' and ``narrow'' covariate sets.}
    \begin{center}
    \begin{subfigure}{0.75\textwidth}
        \centering
        \includegraphics[width=\textwidth]{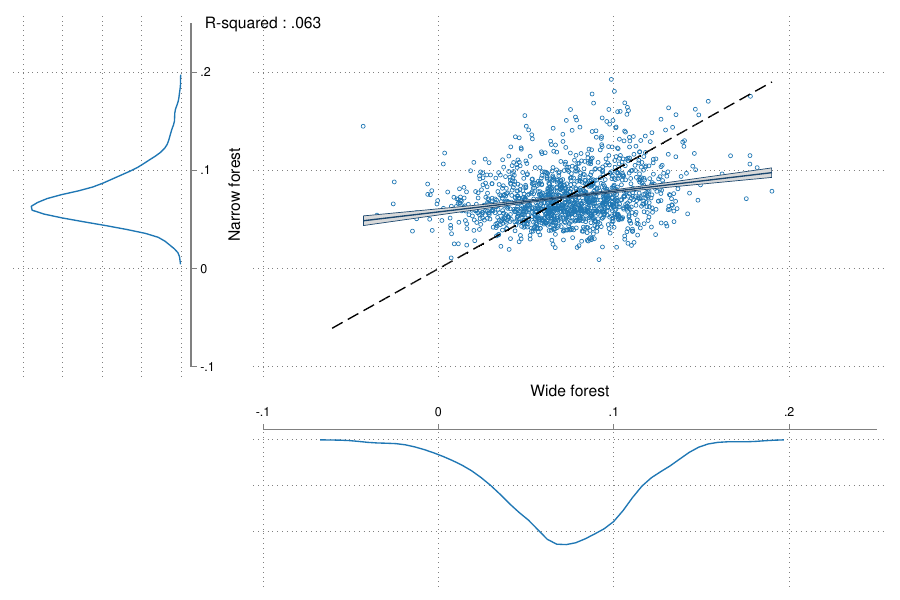}
    \end{subfigure}
    \end{center}
        \footnotesize{This figure plots the relationship between the causal forest conditional ATE estimates from Section \ref{sec:RF} that use the ``wide'' set of covariates (all intake survey responses) and those based on a restricted ``narrow'' set of covariates (age, gender, HS education, desired loan size, and previous borrowing). The outcome variable is CR benefit as a proportion of loan value, so 0.10 equals 10 percentage points, i.e.\ CR multiplied by -1.}
    \label{wide_narrow_forests}

\end{figure}

To explore the performance of targeting, we treat CATE estimates based on the wide covariate set as ground truth. We then construct a dummy variable that equals one if a given borrower has a positive CATE estimate based on the wide covariate set. We then attempt to predict this dummy variable using the narrow covariate set. We consider two approaches: a simple logistic regression model and a random forest classification model. We then consider the in-sample performance of targeting rules based on these two models against a policy of universal mandated structure.

\begin{table}[H]
\caption{Type I \& II errors using targeting narrow rules}
\label{hit_miss_rule}
\begin{center}
\footnotesize{
\begin{tabular}{lccc}
\toprule
\multicolumn{1}{c}{Rule} & \% incorrectly  & \% incorrectly  & Overall Error Rate \\
      & assigned to control  &  assigned to treatment &  \\
\midrule
\midrule
All to Status-quo & 97.16 & 0     & 97.16 \\
All to Structure & 0     & 2.84  & 2.84 \\
Optimal & 0     & 0     & 0 \\
Narrow rule (RF) & 0.27  & 0.69  & 0.96 \\
Narrow rule (Logit) & 0.49  & 0.88  & 1.38 \\
Allow choice & 99.24 & 24.88 & 90.83 \\
\bottomrule
\bottomrule
\end{tabular}%
}
\end{center}
\scriptsize{This table reports error rates for six rules that allocate borrowers to structured repayment. 
}
\end{table} 

Table \ref{hit_miss_rule} reports error rates for six rules that allocate borrowers to structured repayment. Row 1 assigns all borrowers to the status quo (control); Row 2 assigns all borrowers to structure; Row 3 implements optimal targeting based on the CATE from the wide covariate set (ground truth for this exercise); Rows 4 and 5 implement optimal targeting best on the random forest and logit classifiers, respectively, trained on the narrow covariate set; and Row 6 implements borrowers’ own choices. 
Columns (2) and (3) report mis-assignment rates: the shares incorrectly assigned to control and to treatment, respectively. For rules that fix assignments (Rows 1–5), both columns use the denominator of the sample on which that rule is estimated, so the overall error rate equals their sum. For ``Allow choice'' (Row 6), entries in Columns (2)–(3) are computed within the self-selected groups---non-choosers and choosers, respectively---so the overall error rate is the sample-share–weighted average across groups.

While the narrow RF makes relatively accurate treatment assignments, it only improves the overall correct targeting rate by about 1.9 percentage points relative to universal mandated structure. The Logit assignment rule is slightly less accurate, improving on universal mandated structure by 1.5 percentage points. Self-targeting through choice fails to deliver meaningful improvement over assigning everyone to the Status quo contract, given the low take-up rate and the presence of both Type I and Type II errors in the choice arm. Given the low take-up of structure, the high fraction of borrowers with positive estimated CATEs, and the weak predictive power of the narrow covariate set, universal mandated structure appears to be an attractive alternative to targeting in our context.

\begingroup
  \renewcommand{\refname}{Supplementary Appendix References}
  \putbib
  \clearpage
\endgroup
\end{bibunit}

\end{appendix}

\end{document}